\documentclass[11pt,a4paper,reqno]{amsart}
\usepackage[margin=2.5cm]{geometry}
\usepackage{microtype}
\usepackage{amsmath,amsfonts,amssymb,amsthm,mathrsfs}
\usepackage{graphicx}
\usepackage[colorlinks=true,linkcolor=cyan]{hyperref}
\usepackage{enumitem}
\setlist{noitemsep,topsep=3pt}

\newtheorem{theorem}{Theorem}[section]
\newtheorem{lemma}[theorem]{Lemma}
\newtheorem{proposition}[theorem]{Proposition}
\newtheorem{corollary}[theorem]{Corollary}

\theoremstyle{definition}
\newtheorem{definition}[theorem]{Definition}
\newtheorem{example}[theorem]{Example}
\theoremstyle{remark}
\newtheorem{remark}[theorem]{Remark}

\newcommand\U{{\mathrm U}}
\newcommand\SU{{\mathrm {SU}}}

\newcommand\Z{{\mathbb{Z}}}
\newcommand\N{{\mathbb N}}
\newcommand\C{{\mathbb C}}
\newcommand\R{{\mathbb R}}
\newcommand\E{{\mathbb E}}
\newcommand\Tr{{\mathrm{Tr}}}
\newcommand\supp{{\mathrm{supp}}}
\newcommand\dist{{\mathrm{dist}}}
\newcommand\HK{{\mathrm{HK}}}
\newcommand\Id{{\mathrm{Id}}}

\newcommand\Proj{{\mathrm{Proj}}}

\newcommand\conn{{\mathrm{conn}}}

\newcommand\End{{\mathrm{End}}}
\newcommand\Hom{{\mathrm{Hom}}}
\newcommand\Pfr{{\mathcal{P}}}

\newcommand\Pbb{{\mathbb{P}}}
\newcommand\Cref{{\mathcal C_{\mathrm{ref}}}}

\title{The heat-kernel master field on $\Z^d$ at strong coupling}
\author{Thibaut Lemoine}
\address{Universit\'e de Strasbourg, CNRS, IRMA UMR 7501, Strasbourg, France.}
\email{thibaut.lemoine@math.unistra.fr}



\begin{document}
\maketitle

\begin{abstract}
We solve large-$N$ Yang--Mills theory on $\Z^d$, for every $d\geq2$, at strong coupling, for structure group $\mathrm U(N)$ and for the heat-kernel action. More precisely, we prove that normalized Wilson loop expectations have infinite-volume large-$N$ limits, factorize at leading order, and admit an all-order $1/N$-expansion with exponentially local coefficients, whose leading order characterizes the master field. We also prove an area-law upper bound for the heat-kernel master field, with a stronger coefficientwise version.

The proof is based on a rooted heat-kernel master loop equation. Unlike the Wilson-action equation or the two-dimensional Makeenko--Migdal equation, this equation does not close on Wilson loop observables alone; it closes on an extended space of loop observables coupled to compactly supported plaquette decorations. We prove a strong-coupling, order-truncated rooted trajectory expansion and then identify its leading term with the master field. The main inputs are the universal finite-$N$ duality formulas developed in the companion paper \cite{Lem26a} and large-$N$ heat-kernel estimates from \cite{LemMai25,LM2}.

\end{abstract}

\tableofcontents

\section{Introduction}

\subsection{Yang--Mills: from lattice to continuum}

Lattice Yang--Mills theory is a probabilistic model designed to describe Euclidean Yang--Mills theory in a discrete setting. It relies on the following ingredients:
\begin{itemize}
\item A discrete spacetime: an oriented two-dimensional CW complex $\Lambda=(\Lambda^0,\Lambda^1,\Lambda^2)$, which can be an embedded graph $\mathbb G=(V,E,F)$ in a compact surface (for the two-dimensional setting) or a sublattice $\Lambda\subset\Z^d$ for any $d\geq 2$ (for the general setting). In the latter convention, we denote by $E(\Lambda)=\Lambda^1$ its edges and $P(\Lambda)=\Lambda^2$ its plaquettes.
\item A structure group: a compact connected Lie group $G$, typically $\U(N)$ or $\SU(N)$.

\item An action weight: a family $Q=(Q_t)_{t>0}$ of positive central functions $Q_t:G\to\R$.  The parameter $t$ should be thought of as a temperature or area parameter; in the heat-kernel case it is literally the heat time.
\item A coupling: a family $\beta=(\beta_p)_{p\in P(\Lambda)}$ of positive real numbers for each plaquette.  We use the convention $t_p=1/\beta_p$, so that strong coupling, or small inverse temperature in the usual lattice terminology, corresponds to large $t_p$.
\end{itemize}
The discrete Yang--Mills measure on $\Lambda$ with structure group $G$ and coupling constant $\beta$ is a measure $\mu_{\Lambda,\beta,G}$ on the configuration space $G^{E(\Lambda)}$: if we denote by $U=(U_e)_{e\in E(\Lambda)}$ a typical configuration and $dU=\prod_{e\in E(\Lambda)}dU_e$ the associated uniform probability measure, then
\[
d\mu_{\Lambda,\beta,G}(U)=\frac{1}{Z}\prod_{p\in P(\Lambda)}Q_{1/\beta_p}(U_{\partial p})dU,
\]
where $U_{\partial p}$ is the ordered product of the edge variables $U_e$ for $e\in\partial p$, and $Z$ is a normalization constant called partition function.

This convention reconciles two standard notational traditions.  In the Wilson-action literature one usually writes the inverse-temperature parameter $\beta$ in the exponent,
\[
Q_{1/\beta}^{\mathrm W}(g)=e^{-\beta\Re\Tr(I-g)},
\]
whereas in heat-kernel Yang--Mills one usually writes the plaquette weight in terms of the heat time,
\[
Q_t^{\HK}(g)=p_t(g).
\]
Equivalently, the Wilson family may be written as $Q_t^{\mathrm W}(g)=e^{-t^{-1}\Re\Tr(I-g)}$.  The Wilson action was introduced in \cite{Wil74}; the heat-kernel action, first considered by Migdal~\cite{Mig75}, uses the fundamental solution $(p_t)_{t\geq 0}$ of the heat equation
\[
\partial_t p_t=\frac12\Delta p_t,\qquad \lim_{t\downarrow 0}p_t(g)dg=\delta_{1_G}.
\]
With this heat-equation convention, the large-$N$ heat-kernel plaquette weight used in the sequel is the $N$-scaled heat kernel $Q^{\HK}_{T,N}=p_{T/N}$, whose Fourier coefficient is fixed once and for all in~\eqref{eq:Fourier-coef-HK}.  Thus the parameter denoted by $T$ below is the rescaled heat time, and strong coupling corresponds to large $T$.
Both actions are natural approximations of the formal continuum Yang--Mills action, defined as the squared $L^2$ norm of the curvature of a connection of a principal $G$-bundle over the underlying spacetime. There have been good reasons to consider the Wilson action and the heat-kernel action in each situation: on the one hand, the Wilson action is usually easier to analyze, and it links lattice Yang--Mills with quite general random matrix models \cite{CPS25}; on the other hand, the heat-kernel action is used in the construction of a \emph{continuum} two-dimensional Yang--Mills measure, either directly in the continuum \cite{Dri89,Sen97,Lev03} or as a scaling limit of discrete models \cite{ChevShen26,DangNohra26,LemNoh26}. In this sense the heat-kernel action is not merely a convenient alternative to Wilson's action, but the canonical heat-semigroup representative in the two-dimensional Brownian Yang--Mills setting\footnote{See also the paper by Chevyrev--Garban \cite{ChevGar25} for a partial result in higher dimensions.}. The present paper asks whether the heat-kernel action can also be controlled in the higher-dimensional large-$N$ strong-coupling regime. Let us note that, although some partial progress has been made toward higher-dimensional continuum Yang--Mills theory \cite{CCHS24,CC24,ChevShen26b}, these results are formulated in terms of Yang--Mills heat flow rather than in terms of a specific lattice plaquette action. More generally, a rigorous construction for a nonabelian structure group in dimension 3 or 4 remains a major challenge today \cite{JW}. 

\subsection{Large-$N$ regime}

Since 't Hooft's seminal paper \cite{Hoo74}, a special regime has been investigated: the large-$N$ limit. It consists in taking $G$ as a group of $N\times N$ matrices, typically $\U(N)$ or $\SU(N)$, and letting $N$ tend to infinity. The main observables of interest are either the partition function
\[
Z=\int_{G^{E(\Lambda)}}\prod_{p\in P(\Lambda)}Q_{1/\beta_p}(U_{\partial p}) dU,
\]
or Wilson loop expectations. If $L=(\ell_1,\ldots,\ell_k)$ is a family of loops traced in $\Lambda$ and $G$ is a subgroup of $\U(N)$, the Wilson loop functional is the functional $W_{\Lambda,L}:G^{E(\Lambda)}\to\C$ defined by
\[
W_{\Lambda,L}(U)=\Tr(U_{\ell_1})\ldots\Tr(U_{\ell_k}),
\]
where $U_\ell$ is defined as the ordered product of the edge variables $U_e$ traversed by the loop $\ell$. The associated Wilson loop expectation is then the matrix integral
\[
\E_{\Lambda,\beta,G}[W_{\Lambda,L}(U)]=\frac1Z\int_{G^{E(\Lambda)}}\Tr(U_{\ell_1})\ldots\Tr(U_{\ell_k})\prod_{p\in P(\Lambda)}Q_{1/\beta_p}(U_{\partial p})dU.
\]
It is often more convenient to work with \emph{normalized} Wilson loop expectations:
\[
\Phi_{\Lambda,\beta,G}(L)=N^{-k}\E_{\Lambda,\beta,G}[W_{\Lambda,L}(U)].
\]
A major conjecture is the existence of the \emph{master field}, a functional $\phi_{\Lambda,\beta,\infty}$ on the space $\mathcal{L}$ of loops in $\Lambda$ (or a functional $\Phi_\infty$ on the space of loops in the ambient surface if we work with the continuous Yang--Mills measure in two dimensions) such that
\[
\lim_{N\to\infty}\Phi_{\Lambda,\beta,\U(N)}(\ell_1,\ldots,\ell_k)=\Phi_{\Lambda,\beta,\infty}(\ell_1)\ldots\Phi_{\Lambda,\beta,\infty}(\ell_k).
\]
In other words, the normalized Wilson loop expectations should not only converge to finite values, but their limit should also satisfy an asymptotic factorization property.

For the two-dimensional continuous Yang--Mills theory, the existence of the master field was first investigated by Singer \cite{Sing95}, and proved successively in the plane, the sphere, the torus and higher genus surfaces in a series of papers \cite{Lev17,DN,DL23,DL25,PPSY23,Dah26}. The proofs combine harmonic analysis techniques on $\U(N)$ (or other compact classical groups) and loop equations known as Makeenko--Migdal equations \cite{MM79}, whose large-$N$ version has been proved to admit a unique solution: the master field.

Historically, the master field has also been understood as a non-commutative probabilistic object, rather than only as a collection of scalar Wilson loop limits.  This point of view is already present in Singer's formulation of the problem \cite{Sing95}, and it is particularly explicit in the two-dimensional theory: in the planar case, after choosing elementary lassos, the planar master field can be described through free unitary Brownian motion \cite{Lev17}. In that setting, independence of the finite-$N$ lasso variables becomes freeness at large $N$. This is no longer the case on compact surfaces, where elementary lasso variables are no longer independent. Among the notable consequences, the master field on a two-dimensional torus of area $T$ provides a new kind of interpolation between independent $(T=0)$ and free $(T=\infty)$ Haar unitaries \cite{DL25}. The present paper provides a higher-dimensional construction of the master field as a tracial non-commutative law on lattice loop holonomies.  This distinction is discussed in Subsection~\ref{subsec:concluding-remarks}. An extension of the planar master field from the heat semigroup to other L\'evy semigroups was obtained in \cite{CDG17}.

In the case of higher-dimensional lattice Yang--Mills, the existence of the master field for the Wilson action was first proved for several group types by Chatterjee and Jafarov \cite{Jaf16,Cha19}, and novel proofs have been obtained since then using different techniques \cite{SZZ23,BCSK24,BCSK25}. The rigorous large-$N$ results currently known in $\Z^d$ are proved in the \emph{strong coupling} regime, meaning that we consider $\beta$ very small, or $t=1/\beta$ very large.

Let us finally mention that, along with the existence of the master field, another aspect of large-$N$ theory is the conjectured link with string theory: partition functions and Wilson loop expectations should be asymptotically described in terms of string theory, or more loosely stated, ``surface-sums''. This principle, often called ``gauge/string duality'', could already be attributed to 't Hooft \cite{Hoo74} but was really developed by Gross and Taylor \cite{GT,GT2}; it has been investigated both in the lattice and in the 2d continuum settings \cite{Lev08,Jaf16,BaGan18,Cha19,Nov2024,LemMai25,CPS25,LM2}. See \cite{CPS25,LM2,Lem26a} for more detailed accounts on this program, which is currently better understood in two dimensions because it has deep relationships with Hurwitz and Gromov--Witten theories, which describe genuine string theories.

\subsection{Confinement and mass gap}

One of the main motivations for studying Wilson loop expectations is the problem of confinement. In physical terms, confinement means that isolated colour charges cannot be observed at large distances.  On the lattice, Wilson proposed to detect this phenomenon through the decay of Wilson loop expectations: a confining phase should satisfy an \emph{area law}
\[
|\E[W_\ell]|\le C^{|\ell|}\exp\{-\sigma A(\ell)\},
\]
where $A(\ell)$ is a minimal lattice filling area and $\sigma>0$ is the string tension.  This should be contrasted with a perimeter law, where the decay is governed by the length of the loop rather than by the area it encloses.  The area law is the lattice manifestation of a linear potential between static test charges.

The area-law problem is closely related to the mass-gap problem \cite{JW}, and we refer the reader to \cite{Cha21} for a recent probabilistic overview. The rigorous strong-coupling theory of confinement goes back to the classical lattice gauge theory literature: Wilson's original argument predicted area decay at strong coupling \cite{Wil74}; Osterwalder--Seiler developed an early constructive framework for lattice gauge fields and proved strong-coupling area-law estimates \cite{OS78}, and Seiler's monograph gives a systematic account of the constructive and statistical-mechanical approach to gauge theories \cite{Sei82}.  Cao--Nissim--Sheffield recently enlarged the parameter regimes where Wilson area law is known for pure $\U(N)$ lattice Yang--Mills, especially at large $N$, by treating the master loop equation as a linear inhomogeneous equation for Wilson string expectations and controlling the merger term through a truncated model \cite{CNS25a}.  In a subsequent note \cite{CNS25}, they proved Wilson's area law in the 't Hooft regime by verifying a mass-gap criterion of Durhuus--Fr\"ohlich through the dynamical approach of Shen--Zhu--Zhu \cite{DF80,SZZ23}. These results provide a sharp contemporary benchmark for strong-coupling confinement in lattice Yang--Mills; the modern large-$N$ results just mentioned, however, are only formulated for the Wilson action.

\subsection{Main results}

Before stating the main results, let us emphasize the main message of the paper: the universal finite-$N$ formalism of \cite{Lem26a} provides a new way of analyzing lattice Yang--Mills actions beyond the Wilson case, provided that the corresponding actions can be controlled in the large-$N$ regime. The present paper shows that this formalism, combined with the large-$N$ heat-kernel estimates obtained in a series of works in collaboration with Myl\`ene Ma\"ida \cite{LemMai25,LM2}, yields a full strong-coupling large-$N$ analysis of heat-kernel lattice Yang--Mills on $\Z^d$, $d\ge 2$. In this sense the results below should be viewed as the first major application of the universal duality machinery to a non-Wilson action.

Our first result is a large-$N$ analysis of normalized Wilson loop expectations $\Phi_{\Lambda,T,N}:=\Phi_{\Lambda,T,\U(N)}$ in the heat-kernel model with plaquette weight $Q^{\HK}_{T,N}=p_{T/N}$, with the same rescaled parameter $T>0$ on every plaquette, for any finite family $L=(\ell_1,\ldots,\ell_k)$ of loops in $\Z^d$ with finite support, and any finite lattice $\Lambda\subset\Z^d$ that contains $L$. As in the case of the Wilson action, we shall consider only strong coupling, i.e. $T$ large enough. In particular, it provides the existence and characterization of the master field in $\Z^d$ for the heat-kernel action at strong coupling.

\begin{theorem}[Construction of the master field]\label{thm:main}
Let $d\ge2$. There exists $T_0(d)<\infty$ such that, for every $T>T_0(d)$ and every fixed finite loop family $L$, the normalized Wilson expectations admit, for every $R\ge0$, an expansion
\begin{equation}\label{eq:intro-full-expansion}
\Phi_{\Lambda,T,N}(L)
=\sum_{r=0}^R N^{-r}\Phi^{(r)}_{\Lambda,T}(L)+O_{R,L,T}(N^{-R-1}),
\end{equation}
with a remainder uniform over all finite volumes $\Lambda$ containing $L$. Each coefficient has an infinite-volume limit $\Phi^{(r)}_{\infty,T}(L)\in\C$, and there are constants $C_{r,L,T},c>0$ such that
\begin{equation}\label{eq:intro-exp-locality}
|\Phi^{(r)}_{\Lambda,T}(L)-\Phi^{(r)}_{\infty,T}(L)|
\le C_{r,L,T}e^{-c\operatorname{dist}(\operatorname{supp}L,\partial\Lambda)}.
\end{equation}
Consequently, for every sequence of finite volumes $\Lambda_n$ and ranks $N_n$ satisfying
\[
\operatorname{dist}(\operatorname{supp}L,\partial\Lambda_n)\longrightarrow\infty,
\qquad N_n\longrightarrow\infty,
\]
one has the joint limit
\begin{equation}\label{eq:intro-joint-limit}
\Phi_{\Lambda_n,T,N_n}(L)\longrightarrow\Phi^{(0)}_{\infty,T}(L).
\end{equation}
In particular this contains the iterated thermodynamic/large-$N$ limit along every exhaustion. The leading coefficient satisfies
\begin{equation}
\Phi^{(0)}_{\infty,T}(\ell_1,\ldots,\ell_k)=\prod_{i=1}^k\Phi^{(0)}_{\infty,T}(\ell_i).
\end{equation}
It defines a functional $\Phi^{(0)}_{\infty,T}$ on loop families on $\Z^d$ that we call the \emph{heat-kernel master field}.
\end{theorem}

Theorem~\ref{thm:main-expanded} gives the coefficientwise construction underlying this statement and the quantitative finite-volume estimates used in its proof. A discussion on why the resulting asymptotic expansion does not directly yield a gauge/string surface expansion is given in Subsection~\ref{subsec:concluding-remarks}.

The second main result is the confinement of the heat-kernel master field at strong coupling. Let $\mathcal A(\ell)$ denote the lattice filling area of a loop $\ell$, see Definition~\ref{def:lattice-area}.

\begin{theorem}[Area law for the heat-kernel master field]\label{thm:main-area-law}
Let $d\ge2$. There exists $T_{\mathrm{area}}(d)<\infty$ such that, for every $T>T_{\mathrm{area}}(d)$, there are constants $C(T)<\infty$ and $\sigma(T)>0$, with $\sigma(T)$ growing linearly as $T\to\infty$, such that every lattice loop $\ell$ satisfies
\begin{equation}\label{eq:mf-area-law1}
\left|\phi^{\HK}_{\infty,T}(\ell)\right|\le C(T)^{|\ell|}\exp\{-\sigma(T)\mathcal A(\ell)\},
\end{equation}
where $\phi^{\HK}_{\infty,T}(\ell):=\Phi^{(0)}_{\infty,T}(\ell)$ is the one-loop heat-kernel master field. For simple non-trivial loops, after increasing the strong-coupling threshold if necessary, the boundary-length prefactor can be absorbed into the area term, giving a purely exponential bound
\begin{equation}\label{eq:mf-area-law2}
\left|\phi^{\HK}_{\infty,T}(\ell)\right|\le \exp\{-\widetilde\sigma(T)\mathcal A(\ell)\}
\end{equation}
for some $\widetilde\sigma(T)>0$.
\end{theorem}

\subsection{Overview of the proof}
\label{subsec:beyond-wilson}

Let us describe the proof at a structural level. The main difference with the Wilson-action strong-coupling theory is that the heat-kernel plaquette weight is Fourier-theoretic rather than trace-polynomial. Consequently, the integration-by-parts identity does not close on Wilson loop observables alone. We enlarge the state space by allowing compactly supported plaquette decorations, and we prove a rooted master loop equation on this extended space.

The first input is the finite-$N$ duality formalism of \cite{Lem26a}. It rewrites Wilson loop expectations as sums over plaquette representation fields, with heat-kernel spectral weights and universal topological coefficients. After gauge fixing, these topological coefficients admit a local-channel representation on the dual incidence graph. We revisit this representation in a projection-supported form, in which edge integrations remain Haar projections and the non-spectral part of a local-channel summand is separated from the heat-kernel Casimir weights.

The key estimate is a pinned block bound for connected local-channel components attached to the Wilson insertion. Its proof has two steps. First, a peeling argument removes tree appendices and cancels the corresponding plaquette dimension factors before absolute values are taken. Second, the remaining residual core is controlled by a fixed-degree recoupling estimate and by uniform reciprocal-dimension summability along stable Young-graph paths. Together with the large-$N$ heat-kernel estimates of \cite{LemMai25,LM2}, this gives exponentially summable pinned components at strong coupling.

We then use the coefficientwise master loop equation of \cite{Lem26a}. Rooting it at a single occurrence gives a finite-$N$ identity with two parts: a loop-only cut-and-join operator and a local loop--plaquette transfer operator. After adjunction with respect to the heat-kernel reference law, the plaquette-transfer operator satisfies a strong-coupling total-variation bound. Since the loop-only part does not have a uniformly bounded global resolvent in the natural norm, we use an order-truncated rooted trajectory expansion instead of an infinite resolvent series.

This trajectory expansion yields, uniformly in the finite volume, a full $1/N$-expansion for rooted decorated observables. Exponential locality of the coefficients gives infinite-volume limits, and the leading connected-cumulant estimate gives factorization. This proves the existence of the heat-kernel master field on $\mathbb Z^d$ at strong coupling.

Finally, the area-law estimate follows from the reduced central-charge selection rule. In the stable small-mass sector, the projective charge carried by the plaquette decorations forms an integral plaquette filling of the loop. Its total projective mass is therefore at least the lattice filling area. Combining this filling constraint with the heat-kernel decay and the mass-sensitive summability of rooted histories gives the coefficientwise area law, and hence the stated area-law upper bound for the master field.

Let us also emphasize the modularity of the proof.  The finite-$N$ state-sum, the local-channel formalism, the pinned Haar-projection estimates, and the rooted coefficientwise loop equation are essentially action-agnostic.  The heat-kernel action is used through the spectral coercivity and tail estimates which control the plaquette-transfer estimates.  This suggests that the same strategy should apply to other central actions satisfying analogous large-$N$ spectral estimates.  Similarly, the argument is written here for $\U(N)$ because the required finite-$N$ local-channel formalism and the stable heat-kernel estimates are available in this form, but the representation-theoretic mechanisms are not intrinsically unitary: the extension to the other compact classical groups should follow by combining the corresponding finite-$N$ duality formalism with the large-$N$ heat-kernel estimates for the classical series developed in \cite{LM2}.

\subsection{Standing lattice conventions}\label{subsec:lattice-loop-conventions}

We close the introduction by fixing the lattice conventions used throughout the paper.

We use the standard cubical cell structure of $\Z^d$.  Let $(\mathbf e_1,\ldots,\mathbf e_d)$ be the canonical basis.  An oriented nearest-neighbour edge is a pair
\[
e=(x,i),\qquad x\in\Z^d,\quad i\in\{1,\ldots,d\},
\]
with initial vertex $e^-=x$ and terminal vertex $e^+=x+\mathbf e_i$.  Its inverse orientation is denoted by $e^{-1}$; it has initial vertex $e^+$ and terminal vertex $e^-$.  We write $\bar e=\{e,e^{-1}\}$ for the corresponding unoriented edge.  If a gauge configuration assigns a matrix $U_e$ to the chosen positive orientation of $\bar e$, then throughout the paper $U_{e^{-1}}=U_e^{-1}.$ An oriented plaquette is an oriented elementary square.  More explicitly, if $1\le i<j\le d$ and $x\in\Z^d$, the positively oriented plaquette in the $(i,j)$-plane based at $x$ has boundary word
\[
\partial p=(x,i)(x+\mathbf e_i,j)(x+\mathbf e_j,i)^{-1}(x,j)^{-1},
\]
and the opposite orientation is denoted by $p^{-1}$.  We shall often identify a plaquette with its underlying unoriented square when orientation is irrelevant.  The signed incidence number between an oriented edge and an oriented plaquette is denoted by $\varepsilon(e,p)$; it is equal to $\pm 1$ if $e^{\pm 1}$ occurs in $\partial p$, and $0$ otherwise. This is the convention used later in the discrete coboundary operator $\partial^*$.

\begin{figure}[h!]
\centering
\includegraphics[width=0.7\linewidth]{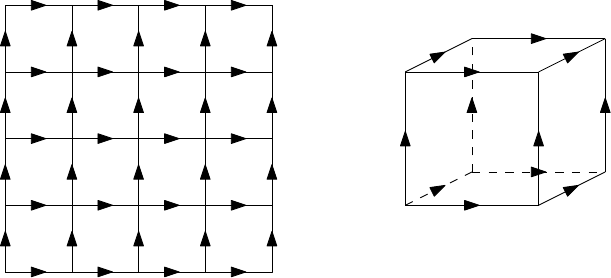}
\caption{Orientation convention of sublattices of $\Z^2$ and $\Z^3$.}
\label{fig:orientation-lattice}
\end{figure}

A finite lattice $\Lambda\subset\Z^d$ is a finite connected two-dimensional subcomplex of this cubical complex.  We denote by $\overline E(\Lambda)$ and $\overline P(\Lambda)$ its unoriented edges and plaquettes.  Once and for all we choose one orientation of every unoriented edge and plaquette, and we denote the resulting oriented sets by $E(\Lambda)$ and $P(\Lambda)$, see Figure~\ref{fig:orientation-lattice}.  Thus $E(\Lambda)$ and $P(\Lambda)$ contain one representative of each unoriented cell; inverse orientations such as $e^{-1}$ and $p^{-1}$ are used as formal oriented cells when writing words and boundaries, but they do not introduce additional integration variables or plaquette weights.  The vertex set $V(\Lambda)$ consists of all endpoints of edges in $\overline E(\Lambda)$.  Connectedness means connectedness of the one-skeleton.  The final gauge-invariant quantities are independent of these auxiliary orientation choices.

The boundary $\partial\Lambda$ is the set of unoriented edges of $\Lambda$ which are adjacent, in the ambient cubical complex of $\Z^d$, to at least one elementary plaquette not belonging to $\overline P(\Lambda)$.  Equivalently, $\bar e\in\partial\Lambda$ if some plaquette of $\Z^d$ contains $\bar e$ and is not part of the finite plaquette set of $\Lambda$.  This convention is used only in locality estimates such as $\dist(\supp L,\partial\Lambda)$; replacing it by any equivalent cell-boundary convention changes such distances by at most a constant depending only on $d$.

A loop in $\Lambda$ is a finite cyclically composable word of oriented edges
\[
\ell=e_1\ldots e_n,\qquad e_r^+=e_{r+1}^-\ (1\le r<n),\qquad e_n^+=e_1^- .
\]
Its length is $|\ell|=n$, and its holonomy is $U_\ell=U_{e_1}\cdots U_{e_n}.$ The trace $\Tr(U_\ell)$ is invariant under cyclic shifts of the word, but we keep the word-level occurrences of edges because rooted loop equations distinguish a chosen occurrence.  A finite loop family is an ordered tuple $L=(\ell_1,\ldots,\ell_k)$; we write
\[
\begin{aligned}
\#L&=k,\qquad \ell(L)=\sum_{i=1}^k|\ell_i|,\\
j_L(e)&=\#\{\text{occurrences of }e\text{ in }L\}
-\#\{\text{occurrences of }e^{-1}\text{ in }L\}.
\end{aligned}
\]
The associated Wilson functional is $W_{\Lambda,L}(U)=\prod_{i=1}^k \Tr(U_{\ell_i}).$ We say that $\Lambda$ contains $L$ if every edge occurrence in the words $\ell_i$ is of the form $e^{\pm1}$ for some chosen representative $e\in E(\Lambda)$, equivalently if every underlying unoriented edge traversed by $L$ belongs to $\overline E(\Lambda)$.

The support of a loop $\ell$ is the finite set of unoriented edges it traverses, together with their endpoints; the support of a loop family is the union of the supports of its components.  If spectral plaquette labels or compactly supported decorations are present, their support also includes the corresponding plaquettes.  Distances between supports are measured in the incidence graph of the cubical complex: two cells are adjacent if one lies in the boundary of the other, and $\dist(A,B)$ is the minimal graph distance between cells of $A$ and cells of $B$.  With this convention, all supports used below are finite cell sets.

\section{Stable representations and heat-kernel estimates}
\label{sec:heat-kernel-estimates}

This section collects the harmonic-analysis and representation-theoretic estimates used in the proofs of the main results. It recalls and complements the analytical large-$N$ estimates for the heat kernel \cite{LemMai25,LM2} as well as algebraic constructions related to the finite-$N$ construction of the companion paper~\cite{Lem26a}.

\subsection{Harmonic analysis and heat-kernel normalization}

We begin with the Fourier conventions for central functions on $\U(N)$. These conventions fix the normalization of dimensions, Casimir eigenvalues and heat-kernel coefficients used in all later estimates. The set $\widehat{\U(N)}$ of irreducible representations of $\U(N)$ is countable, and it is in bijection with the set of highest weights
\[
\lambda=(\lambda_1,\ldots,\lambda_N)\in\Z^N\quad \text{such that}\quad \lambda_1\geq\ldots\geq\lambda_N.
\]
We shall denote by $(\rho_\lambda,V_\lambda)$ an irreducible representation associated to the highest weight $\lambda$; by Schur's lemma, all such representations are equivalent to each other, hence they have the same character $\chi_\lambda$. We also denote by $d_\lambda=d_{V_\lambda}$ the associated dimension. By the Peter--Weyl theorem, any square integrable central function $f:\U(N)\to\R$ admits the decomposition in $L^2$
\[
f(x)=\sum_{\lambda\in\widehat{\U(N)}}\widehat{f}(\lambda)\chi_\lambda(x),\qquad \widehat{f}(\lambda)=\langle f,\chi_\lambda\rangle_{L^2}=\int_{\U(N)}f(g)\overline{\chi_\lambda(g)}dg.
\]
For an irreducible representation $\lambda\in\widehat{\U(N)}$, the quantity $\widehat{f}(\lambda)$ is the non-commutative Fourier coefficient of $f$ associated to $\lambda$. Another noteworthy property of the irreducible characters $\chi_\lambda$ is that they are eigenfunctions of the Laplace--Beltrami operator on $\U(N)$:
\begin{equation}
\Delta\chi_\lambda=-C_2(\lambda)\chi_\lambda,\qquad \lambda\in\widehat{\U(N)},
\end{equation}
with eigenvalues $C_2(\lambda)\geq 0$ called \emph{Casimir eigenvalues}, or \emph{Casimir numbers}. The name Casimir refers to the fact that they are also related to the Casimir operator, an operator in the center of the enveloping algebra of $\mathfrak u(N)$. The dimension and Casimir of an irreducible representation of highest weight $\lambda$ admit explicit formulas\footnote{For the Casimir number, we choose a different normalization than in \cite{Lem22,LemMai25,LM2}, where there is a $1/N$ prefactor, but the prefactor will appear back in the time scaling, so that we are effectively in the same regime.}:
\begin{equation}\label{eq:dim-cas}
d_\lambda=\prod_{1\leq i<j\leq N}\frac{\lambda_i-i+j-\lambda_j}{j-i},\qquad C_2(\lambda)=\sum_{i=1}^N\lambda_i(\lambda_i+N+1-2i).
\end{equation}

Let $Q=(Q_t)_{t>0}$ be a family of smooth positive central functions on $\U(N)$.  The parameter $t$ is the action parameter of the family $Q_t$.  In this paper the lattice coupling is denoted by $\beta$, and we set $t_p=\frac1{\beta_p}.$ This convention reconciles the two standard traditions:
\begin{itemize}
\item In Wilson-action lattice gauge theory, the inverse-temperature parameter is usually denoted by $\beta$ and appears in the exponent,
\[
Q^{\mathrm W}_{1/\beta}(g)=\exp\{-\beta\Re\Tr(I-g)\}.
\]
\item In heat-kernel Yang--Mills, the natural parameter is instead the heat time, so that $Q_t^{\HK}=p_t$.
\end{itemize}
In the strong-coupling heat-kernel part of the paper we use the large-$N$ heat-kernel normalization
\begin{equation}\label{eq:Fourier-coef-HK}
Q^{\HK}_{T,N}=p_{T/N},
\qquad
\widehat Q^{\HK}_{T,N}(\lambda)=d_\lambda\exp\left\{-\frac{T}{2N}C_2(\lambda)\right\}.
\end{equation}
Thus $T$ is the $N$-scaled heat time.  In the informal coupling notation, this corresponds to a constant rescaled coupling parameter rather than to an unscaled heat time independent of $N$.

\subsection{Stable coordinates and tail bounds}

There is a coordinate system, closely related to mixed Schur--Weyl duality, in which the large-$N$ heat-kernel coefficient has a transparent form. It is based on a decomposition of highest weights in terms of a determinant shift, playing the role of a highest weight of $\U(1)$, and two integer partitions playing the role of irreducible representations of symmetric groups.

Recall that an integer partition is a finite nonincreasing family of positive integers $\mu=(\mu_1,\ldots,\mu_k)$; we denote its length by $\ell(\mu)=k$ and its size by $\vert\mu\vert=\sum_i\mu_i$. We also write $\mu\vdash n$ if $\mu$ is a partition of size $\vert\mu\vert=n$. We write $\Pfr_n=\{\mu\mid \mu\vdash n\}$ the set of partitions of size $n$, and $\Pfr=\bigsqcup_{n\geq 0}\Pfr_n$ the set of all integer partitions, with the convention that the only partition of size $0$ is the empty partition $\varnothing.$ Given two partitions $\lambda^+,\lambda^-$, for any $N\geq \ell(\lambda^+)+\ell(\lambda^-)$, one can define the \emph{stable highest weight}
\[
[\lambda^+,\lambda^-]_N=(\lambda_1^+,\ldots,\lambda_{\ell(\lambda^+)}^+,0,\ldots,0,-\lambda_{\ell(\lambda^-)}^-,\ldots,-\lambda_1^-)\in\widehat{\U(N)},
\]
obtained by adding the coefficients of $\lambda^+$ to the first coefficients of the flat highest weight $(0,\ldots,0)\in\widehat{\U(N)}$, and removing the coefficients of $\lambda^-$ in reverse order to its last coefficients. Using the representation of partitions in terms of Young diagrams, a display of a stable highest weight is given in Figure~\ref{fig:stable-highest-weight}. We denote by $\widehat{\U(N)}_{\mathrm{st}}$ the set of stable highest weights.

\begin{figure}[h!]
\centering
\includegraphics[width=0.3\linewidth]{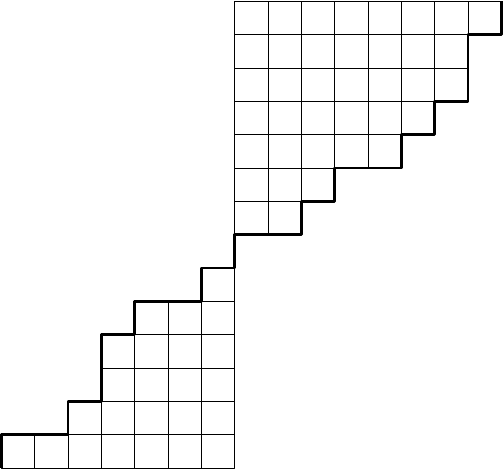}
\caption{An example of stable highest weight $[\lambda^+,\lambda^-]_N$ with only one zero.}
\label{fig:stable-highest-weight}
\end{figure}

For any $N$, introduce two length cutoffs on partitions
\begin{equation}\label{eq:length-cutoff}
A_N=\lfloor (N+1)/2\rfloor-1,\quad B_N=N-1-A_N.
\end{equation}
As discussed in \cite{Lem22}, or more recently in \cite{LemMai25,LM2}, there is a bijection $\lambda_N:\Omega_N\to\widehat{\U(N)}$, where
\[
\Omega_N=\{(\lambda^+,\lambda^-,q)\in\Pfr\times\Pfr\times\Z\mid \ell(\lambda^+)\leq A_N,\ \ell(\lambda^-)\leq B_N\},
\]
given by
\[
\begin{aligned}
\lambda_N(\lambda^+,\lambda^-,q)
&=q\mathbf{1}_N+[\lambda^+,\lambda^-]_N\\
&=(\lambda_1^++q,\ldots,\lambda_{\ell(\lambda^+)}^++q,
q,\ldots,q,
q-\lambda_{\ell(\lambda^-)}^-,\ldots,q-\lambda_1^-).
\end{aligned}
\]

\begin{definition}\label{def:stable-coordinates}
For any $\lambda\in\widehat{\U(N)}$, its \emph{stable coordinates} are $(\lambda^+,\lambda^-,q)$ such that $\lambda=\lambda_N(\lambda^+,\lambda^-,q)$. Its \emph{projective mass} is $m(\lambda)=\vert\lambda^+\vert+\vert\lambda^-\vert$, and its \emph{determinant shift} is $q$.
\end{definition}

In stable coordinates, the Casimir becomes easier to analyze in the large-$N$ regime: for each cell $\square$ in the Young diagram associated to $\lambda^+$ or $\lambda^-$, if we write $(i,j)$ its coordinates in the diagram, the content of the cell is defined by $c(\square)=c(i,j)=j-i.$ The \emph{total content} is a map $K:\Pfr\to\Z$ that assigns to a partition the sum of contents of all its cells:
\[
K(\lambda)=\sum_{\square\in\lambda}c(\square).
\]
The Casimir of a stable highest weight can be rewritten
\[
C_2([\lambda^+,\lambda^-]_N)=Nm([\lambda^+,\lambda^-]_N)+2K(\lambda^+)+2K(\lambda^-),
\]
and for a general highest weight we get
\begin{equation}\label{eq:casimir-det-shift}
C_2(\lambda)=Nq^2+2q(\vert\lambda^+\vert-\vert\lambda^-\vert)+C_2([\lambda^+,\lambda^-]_N).
\end{equation}
We shall use later the following domination estimate \cite[Lemma~4.1]{LemMai25}:
\begin{equation}\label{eq:LM-casimir-domination}
\frac{1}{N}C_2(\lambda)\geq \frac12m(\lambda)+\left(q+\frac{|\lambda^+|-|\lambda^-|}{N}\right)^2 .
\end{equation}
As an elementary consequence of \eqref{eq:LM-casimir-domination},
\begin{equation}\label{eq:spectral-decay}
\exp\left(-\frac{T}{2N}C_2(\lambda)\right)\leq \exp\left(-\frac{T}{4}\mathcal E(\lambda)\right)\exp\left(\frac{T m(\lambda)^2}{2N^2}\right),
\end{equation}
where $\mathcal E(\lambda)=m(\lambda)+q^2$ is an energy functional corresponding to the leading term of the large-$N$ expansion of the Casimir operator \cite{LM2}. Equation~\eqref{eq:spectral-decay} follows by applying \eqref{eq:LM-casimir-domination} and using
$\bigl(q+(|\lambda^+|-|\lambda^-|)/N\bigr)^2
\geq \frac12 q^2-m(\lambda)^2/N^2$.

Following~\cite{LemMai25}, define the \emph{Gaussian measure} on $\widehat{\U(N)}$ by
\[
\widehat\Pbb_{T,N}(\lambda)=\frac{1}{\mathcal Z_{T,N}}e^{-\frac{T}{2N}C_2(\lambda)},\qquad \lambda\in\widehat{\U(N)},\ T\in(0,\infty),\ N\in\N^*,
\]
and for any $A\geq 1$, define the tilted Gaussian measure
\[
\widehat\Pbb_{A,T,N}(\lambda)=\frac{1}{\mathcal Z_{A,T,N}}A^{m(\lambda)}e^{-\frac{T}{2N}C_2(\lambda)},
\]
where $\mathcal Z_{A,T,N}$ is the normalization constant.  If $S$ is a finite plaquette set, we use the product laws
\begin{equation}\label{eq:product-ref-laws}
\widehat\Pbb_{S,T,N}:=\widehat\Pbb_{T,N}^{\otimes S},\qquad
\widehat\Pbb_{S,A,T,N}:=\widehat\Pbb_{A,T,N}^{\otimes S},
\end{equation}
and write $\widehat\E_{S,T,N}$ and $\widehat\E_{S,A,T,N}$ for their expectations.  Thus, for any function $F$ of a plaquette-label field on $S$,
\[
\sum_{\lambda_S}F(\lambda_S)\prod_{p\in S}e^{-\frac{T}{2N}C_2(\lambda_p)}=\mathcal Z_{T,N}^{|S|}\widehat\E_{S,T,N}[F(\lambda_S)],
\]
and
\[
\sum_{\lambda_S}F(\lambda_S)A^{\sum_{p\in S}m(\lambda_p)}\prod_{p\in S}e^{-\frac{T}{2N}C_2(\lambda_p)}=\mathcal Z_{A,T,N}^{|S|}\widehat\E_{S,A,T,N}[F(\lambda_S)].
\]

\begin{lemma}\label{lem:nonstable-weighted-tail}
Let $A\geq1$.  If $T>4\log A$, then there are constants $C(A,T)<\infty$ and $\gamma(A,T)>0$, independent of $N$, such that
\begin{equation}\label{eq:nonstable-weighted-tail}
\widehat\Pbb_{A,T,N}\left(m(\lambda)>\frac{N}{2}\right)\leq C(A,T)e^{-\gamma(A,T)N}.
\end{equation}
\end{lemma}

\begin{proof}
By \eqref{eq:LM-casimir-domination},
\[
\begin{aligned}
\exp\left(-\frac{T}{2N}C_2(\lambda)\right)
&\leq
\exp\left(-\frac{T}{4}(|\lambda^+|+|\lambda^-|)\right)\\
&\qquad\times
\exp\left[-\frac{T}{2}\left(q+\frac{|\lambda^+|-|\lambda^-|}{N}\right)^2\right].
\end{aligned}
\]
Set $\rho=Ae^{-T/4}$.  The assumption $T>4\log A$ is exactly $\rho<1$.  Since
\[
\Theta_T:=\sup_{x\in\mathbb R}\sum_{n\in\Z}e^{-\frac{T}{2}(n+x)^2}<\infty,
\]
we obtain, after forgetting the length restrictions on $\lambda^\pm$,
\begin{align*}
\sum_{m(\lambda)>N/2}A^{m(\lambda)}e^{-TC_2(\lambda)/(2N)}&\leq \Theta_T \sum_{\substack{a,b\geq0\\ a+b>N/2}}p(a)p(b)\rho^{a+b}.
\end{align*}
If $a+b>N/2$, then either $a>N/4$ or $b>N/4$.  Hence
\[
\sum_{a+b>N/2}p(a)p(b)\rho^{a+b}\leq 2\left(\sum_{a\geq0}p(a)\rho^a\right)\left(\sum_{b>N/4}p(b)\rho^b\right).
\]
The first factor is finite because $\rho<1$ and the generating function of the number of partitions has convergence radius 1, and the second is exponentially small in $N$ by \cite[Lemma~2.4]{LemMai25}. This proves \eqref{eq:nonstable-weighted-tail}.
\end{proof}

\begin{lemma}[Product tilted-mass tail]\label{lem:product-tilted-mass-tail}
Let $A\ge1$ and $T>4\log A$. There exist constants $C(A,T)<\infty$ and $\gamma(A,T)>0$, independent of $N$ and of the finite set $S$, such that
\begin{equation}\label{eq:product-tilted-mass-tail}
\widehat\Pbb_{S,A,T,N}\left(\sum_{p\in S}m(\lambda_p)>\frac N2\right)
\le C(A,T)^{|S|}e^{-\gamma(A,T)N}.
\end{equation}
\end{lemma}

\begin{proof}
Write $\widehat\E_{A,T,N}$ for expectation under the one-site law $\widehat\Pbb_{A,T,N}$. Choose $\theta>0$ so small that
\[
Ae^{\theta}e^{-T/4}<1,
\]
which is possible because $T>4\log A$. Under the one-site tilted law,
\[
\widehat\E_{A,T,N}\bigl[e^{\theta m(\lambda)}\bigr]
=\frac{\mathcal Z_{Ae^{\theta},T,N}}{\mathcal Z_{A,T,N}}.
\]
The denominator is at least the contribution of the trivial representation, hence at least $1$. The proof of Lemma~\ref{lem:nonstable-weighted-tail}, with $A$ replaced by $Ae^{\theta}$ and without imposing a tail event, gives the uniform bound
\[
\mathcal Z_{Ae^{\theta},T,N}
\le \Theta_T\left(\sum_{a\ge0}p(a)\bigl(Ae^{\theta}e^{-T/4}\bigr)^a\right)^2
=:C_\theta(A,T)<\infty.
\]
Consequently, by the product structure~\eqref{eq:product-ref-laws} and exponential Markov inequality,
\begin{align*}
\widehat\Pbb_{S,A,T,N}\left(\sum_{p\in S}m(\lambda_p)>\frac N2\right)
&\le e^{-\theta N/2}
\prod_{p\in S}\widehat\E_{A,T,N}\bigl[e^{\theta m(\lambda_p)}\bigr]\\
&\le C_\theta(A,T)^{|S|}e^{-\theta N/2}.
\end{align*}
This is~\eqref{eq:product-tilted-mass-tail} with $\gamma(A,T)=\theta/2$.
\end{proof}

\subsection{Mixed Schur--Weyl duality}
\label{subsec:mixed-schur-weyl}

We next recall the fixed-degree mixed Schur--Weyl model in which the local channel coordinates are taken, following for instance \cite{Dah26,Lem26a}. Let $V_N=\C^N$, with standard basis $(e_i)_{1\le i\le N}$, and let $V_N^*=(\C^N)^*$ be the dual representation, with dual basis $(e^i)_{1\le i\le N}$.  For any $n^+,n^-\ge0$, define the mixed tensor space
\[
T_N^{n^+,n^-}=V_N^{\otimes n^+}\otimes (V_N^*)^{\otimes n^-},
\qquad
T_N^{0,0}=\C.
\]
For $n^+,n^-\ge1$ and $1\le i\le n^+$, $1\le j\le n^-$, let $c_{i,j}:T_N^{n^+,n^-}\longrightarrow T_N^{n^+-1,n^--1}$ be the contraction which evaluates the $j$-th dual factor on the $i$-th covariant factor and leaves all other tensor factors in their induced order.  Dually, the coevaluation map $\operatorname{coev}_{i,j}:T_N^{n^+-1,n^--1}\longrightarrow T_N^{n^+,n^-}$ is obtained by inserting the invariant tensor $\sum_{a=1}^N e_a\otimes e^a$ in the $i$-th covariant and $j$-th contravariant slots.  If $n^+=0$ or $n^-=0$, there are no contractions and we set $\mathring T_N^{n^+,n^-}=T_N^{n^+,n^-}$. The \emph{mixed Schur--Weyl duality} proved by Koike~\cite{Koi89} states that the space of traceless tensors
\[
\mathring{T}_N^{n^+,n^-}:=\bigcap_{\substack{1\leq i\leq n^+\\ 1\leq j\leq n^-}}\ker c_{i,j}
\]
admits a decomposition
\begin{equation}\label{eq:Koike}
\mathring{T}_N^{n^+,n^-}=\bigoplus_{\substack{\lambda^+\vdash n^+,\lambda^-\vdash n^-\\ \ell(\lambda^+)+\ell(\lambda^-)\leq N}}V^{[\lambda^+,\lambda^-]}\otimes V_{[\lambda^+,\lambda^-]_N},
\end{equation}
where the vector spaces of the right-hand side are defined as follows:
\begin{itemize}
\item $V_{[\lambda^+,\lambda^-]_N}$ denotes the space associated to the stable representation $[\lambda^+,\lambda^-]_N$ of $\U(N)$;
\item for $n\ge1$, the irreducible representations of $S_n$ are labelled by $V^\lambda$, $\lambda\vdash n$; for $n=0$ we use the convention $S_0=\{1\}$ and $V^\varnothing=\C$.  We set $V^{[\lambda^+,\lambda^-]}=V^{\lambda^+}\otimes V^{\lambda^-}$ as a representation of $S_{n^+}\times S_{n^-}$.
\end{itemize}
Denote by $P_N^{[\lambda^+,\lambda^-]}\in\End(T_N^{n^+,n^-})$ the orthogonal projector onto the $(\lambda^+,\lambda^-)$-component in the decomposition
\begin{equation}\label{eq:Koike2}
T_N^{n^+,n^-}=(\mathring{T}_N^{n^+,n^-})^\perp\oplus\bigoplus_{\substack{\lambda^+\vdash n^+,\lambda^-\vdash n^-\\ \ell(\lambda^+)+\ell(\lambda^-)\leq N}}V^{[\lambda^+,\lambda^-]}\otimes V_{[\lambda^+,\lambda^-]_N}.
\end{equation}
Let $\rho_{n^+,n^-}$ be the representation of $\U(N)$ on $T_N^{n^+,n^-}$ obtained by tensorizing $n^+$ times the fundamental representation and $n^-$ times the dual representation:
\[
U\cdot v_1\otimes\ldots\otimes v_{n^+}\otimes v_1^*\otimes\ldots\otimes v_{n^-}^*=\bigotimes_{i=1}^{n^+}U v_i\otimes\bigotimes_{j=1}^{n^-}v_j^*(U^{-1}).
\]
The irreducible character associated to the stable representation $[\lambda^+,\lambda^-]_N$ admits the trace representation
\begin{equation}\label{eq:trace-character}
\chi_{[\lambda^+,\lambda^-]_N}(U)=\Tr_{T_N^{n^+,n^-}}
\left(P_N^{[\lambda^+,\lambda^-]}\rho_{n^+,n^-}(U)\right).
\end{equation}

We now fix the diagrammatic notation used to expand $P_N^{[\lambda^+,\lambda^-]}$.  Put
\[
I_+=\{1,\ldots,n^+\},\qquad I_-=\{-1,\ldots,-n^-\},\qquad I=I_+\sqcup I_-.
\]
A walled Brauer diagram of type $(n^+,n^-)$ is a perfect matching of the two-row vertex set $I\times\{\mathrm t,\mathrm b\}$ such that vertical strings preserve the side of the wall and horizontal strings cross the wall.  Equivalently, if two vertices $(a,\epsilon)$ and $(b,\epsilon')$ are matched, then
\[
\epsilon\ne\epsilon' \Rightarrow ab>0,\qquad\epsilon=\epsilon' \Rightarrow ab<0.
\]
We denote the finite set of such diagrams by $\mathcal B_{n^+,n^-}$.  The corresponding endomorphism $\rho_N(\tau)\in\End(T_N^{n^+,n^-})$ is defined by Kronecker deltas as follows.  For a multi-index $\mathbf i=(i_a)_{a\in I}\in\{1,\ldots,N\}^I$, set
\[
E_{\mathbf i}:=e_{i_1}\otimes\cdots\otimes e_{i_{n^+}}\otimes e^{i_{-1}}\otimes\cdots\otimes e^{i_{-n^-}}.
\]
For top and bottom multi-indices $\mathbf j,\mathbf i\in\{1,\ldots,N\}^I$, define
\[
\delta_\tau(\mathbf j,\mathbf i):=\prod_{\{(a,\epsilon),(b,\epsilon')\}\in\tau}\delta_{\kappa_{a,\epsilon},\kappa_{b,\epsilon'}},\qquad\kappa_{a,\mathrm t}:=j_a,\quad\kappa_{a,\mathrm b}:=i_a.
\]
Then
\[
\rho_N(\tau)E_{\mathbf i}=\sum_{\mathbf j\in\{1,\ldots,N\}^I}\delta_\tau(\mathbf j,\mathbf i)E_{\mathbf j}.
\]
This is the usual walled-Brauer action on mixed tensors: vertical strands permute covariant indices among themselves and contravariant indices among themselves, while horizontal strands implement contractions and coevaluations across the wall. An illustration of walled-Brauer diagrams is given in Figure~\ref{fig:Brauer}.

\begin{figure}[h!]
\centering
\includegraphics[width=0.3\linewidth]{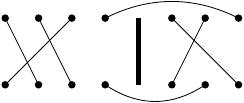}
\caption{An example of walled-Brauer diagram $\tau\in\mathcal B_{4,3}$.}
\label{fig:Brauer}
\end{figure}

The following result will be important later; it is the standard walled-Brauer expansion of the mixed Schur--Weyl projector, in the form used for instance in \cite{Dah26,Lem26a}.
\begin{proposition}\label{prop:expansion-projector}
For any fixed mixed degree $(n^+,n^-)$ and any stable label $[\lambda^+,\lambda^-]_N$ with $\lambda^+\vdash n^+$ and $\lambda^-\vdash n^-$, the projector $P_N^{[\lambda^+,\lambda^-]}$ belongs to the image of the walled-Brauer algebra. Equivalently, after fixing the diagrammatic spanning family $\mathcal B_{n^+,n^-}$, there are coefficients $c_N^{[\lambda^+,\lambda^-]}(\tau)$, $\tau\in\mathcal B_{n^+,n^-}$, such that
\begin{equation}\label{eq:projector-walled-brauer-expansion}
P_N^{[\lambda^+,\lambda^-]}=\sum_{\tau\in\mathcal B_{n^+,n^-}}c_N^{[\lambda^+,\lambda^-]}(\tau)\rho_N(\tau).
\end{equation}
For fixed $(n^+,n^-)$ and fixed $(\lambda^+,\lambda^-)$, the coefficients can be chosen as rational functions of $N$ for all sufficiently large $N$; in particular they admit controlled asymptotic expansions to all orders in $1/N$.
\end{proposition}

\begin{proof}
The inclusion in the image of the walled-Brauer algebra is the standard mixed Schur--Weyl commutant statement, in the form used in~\cite{Dah26,Lem26a}. For the asymptotic assertion, keep the mixed degree fixed and take $N\ge n^++n^-$. In this stable range the diagram representation of the walled-Brauer algebra is faithful. The diagrammatic multiplication table and trace pairing have entries which are polynomials in the parameter $N$: every closed index loop contributes one factor $N$. The central/isotypic idempotent corresponding to the fixed pair $(\lambda^+,\lambda^-)$ is therefore obtained by solving a finite linear system whose matrix has polynomial entries in $N$. Since the idempotent exists and is unique for all sufficiently large $N$, the relevant determinant is not the zero polynomial. Cramer's rule consequently gives a choice of its diagram coefficients as rational functions of $N$. Expanding these rational functions at $N=\infty$ yields the claimed controlled $1/N$ expansions. The finitely many ranks below the stable threshold play no role in the fixed-degree asymptotic arguments below.
\end{proof}

\subsection{Fixed-degree coordinates and orthonormalization}
\label{subsec:channel-recoupling}
\label{subsec:fixed-degree-coordinates}

The walled-Brauer coordinates are algebraically convenient, but the estimates in Sections~\ref{sec:finiteN} and~\ref{sec:pinned-haar-projection} are taken in orthonormal channel bases. We therefore record the fixed-degree coordinate and orthonormalization conventions.

Throughout the paper, when we say that a quantity admits a controlled asymptotic expansion to all orders in powers of $1/N$, we mean that, after subtracting the terms up to any prescribed order $R$, the remainder is $O_R(N^{-R-1})$. For local data of fixed support and fixed stable mass, the constant in the remainder is uniform over the corresponding finite set of choices and over the ambient finite volume. For infinite stable sums, the coefficient sums are absolutely convergent and the remainder is bounded by $N^{-R-1}$ times a summable stable majorant, up to large-mass or unstable contributions which are exponentially small in $N$. No pointwise expansion is asserted for arbitrary unstable representation labels.

We now make explicit the coordinate conventions used below.  Let
$\varepsilon=(\varepsilon_1,\ldots,\varepsilon_r)\in\{+,-\}^r$, put
$V_N^+=V_N$, $V_N^-=V_N^*$, and set $E_N^\varepsilon=\bigotimes_{a=1}^r V_N^{\varepsilon_a}.$ More generally, a \emph{fixed mixed tensor space} means any tensor product obtained from finitely many copies of $V_N$ and $V_N^*$, with a degree independent of $N$.  If a reduced stable label $\lambda$ occurs in such a space $E_N$, we write
\[
M_\lambda(E_N):=\Hom_{\U(N)}(V_\lambda,E_N),
\qquad
\langle D,D'\rangle_\lambda
:=d_\lambda^{-1}\Tr_{V_\lambda}(D^*D') .
\]
Thus the isotypic decomposition can be written
\begin{equation}\label{eq:isotypic-decomposition}
E_N\simeq \bigoplus_\lambda M_\lambda(E_N)\otimes V_\lambda .
\end{equation}
The multiplicity space $M_\lambda(E_N)$ is the space in which coordinate choices are made.

In the special case $E_N=T_N^{n^+,n^-}$, the Koike decomposition~\eqref{eq:Koike2} identifies the relevant multiplicity spaces with the symmetric-group modules $V^{[\lambda^+,\lambda^-]}$, together with the possible traceless-complement part.  Choosing a basis in these multiplicity spaces gives what we call \emph{mixed Schur--Weyl coordinates}. They must be opposed to the \emph{walled-Brauer coordinates}, which are the diagrammatic coordinates associated with the spanning family $(\rho_N(\tau))_{\tau\in\mathcal B_{n^+,n^-}}$ of Subsection~\ref{subsec:mixed-schur-weyl}: if an equivariant operator $A_N\in\End(T_N^{n^+,n^-})$ is written in the form
\[
A_N=\sum_{\tau\in\mathcal B_{n^+,n^-}} a_N(\tau)\rho_N(\tau),
\]
then the scalars $a_N(\tau)$ are its walled-Brauer coordinates in this chosen spanning family.  For instance, the coefficients $c_N^{[\lambda^+,\lambda^-]}(\tau)$ in~\eqref{eq:projector-walled-brauer-expansion} are the walled-Brauer coordinates of the projector $P_N^{[\lambda^+,\lambda^-]}$.  Tensor products, contractions, coevaluations and compositions of such diagrammatic operators give the algebraic coordinate families used below.  We shall refer to them collectively as \emph{algebraic diagrammatic coordinates in fixed mixed degree}.

Finally, an \emph{orthonormal channel basis} of $M_\lambda(E_N)$ is an orthonormal basis for the inner product $\langle\cdot,\cdot\rangle_\lambda$.  If $\mathcal I_\lambda(E_N)=(I^{E,\lambda}_{a,N})_a$ is such a basis, then
\[
I^{E,\lambda}_{a,N}:V_\lambda\longrightarrow E_N,
\qquad
(I^{E,\lambda}_{a,N})^*I^{E,\lambda}_{b,N}=\delta_{ab}\Id_{V_\lambda}.
\]
If $A_N:E_N\to F_N$ is a $\U(N)$-equivariant map between fixed mixed tensor spaces and if orthonormal channel bases have been chosen in $E_N$ and $F_N$, its orthonormal channel matrix coefficients are
\begin{equation}\label{eq:orthonormal-channel-coordinates}
A^{\mathcal I}_{\lambda;b,a}(N)
:=d_\lambda^{-1}\Tr_{V_\lambda}\!\left((I^{F,\lambda}_{b,N})^*A_NI^{E,\lambda}_{a,N}\right).
\end{equation}
The following elementary orthonormalization device explains how we pass from algebraic diagrammatic coordinates to orthonormal channel coordinates without losing controlled expansions.

\begin{lemma}\label{lem:fixed-degree-orthonormalization}
Fix a mixed tensor space $E_N$ and a fixed reduced stable label $\lambda$ occurring in $E_N$. Let $D^{(N)}_1,\ldots,D^{(N)}_m\in M_\lambda(E_N)=\Hom_{\U(N)}(V_\lambda,E_N)$ be an algebraic diagrammatic coordinate family in fixed mixed degree, and suppose that its scalar coefficients in these coordinates, equivalently its matrix coefficients after the fixed mixed-tensor realization of $V_\lambda$, have controlled asymptotic expansions in powers of $1/N$. Define the normalized Gram matrix
\[
G^{(N)}_{ij}=\langle D^{(N)}_i,D^{(N)}_j\rangle_\lambda
=d_\lambda^{-1}\Tr_{V_\lambda}\!\left((D^{(N)}_i)^*D^{(N)}_j\right).
\]
Assume that $G^{(N)}=G^{(\infty)}+O(N^{-1}),$ with $G^{(\infty)}>0$. Then, for $N$ large enough, $G^{(N)}$ is invertible and the orthonormalized intertwiners
\[
I^{(N)}_a=
\sum_{i=1}^m D^{(N)}_i (G^{(N)})^{-1/2}_{ia}
\]
form an orthonormal channel basis of their span and have controlled asymptotic expansions to all orders in powers of $1/N$.
\end{lemma}

\begin{proof}
The entries of $G^{(N)}$ have controlled asymptotic expansions by the fixed-degree assumption. Since $G^{(\infty)}$ is positive definite, the spectrum of $G^{(N)}$ stays in a compact subset of $(0,\infty)$ for $N$ large enough. Holomorphic functional calculus applied to $z\mapsto z^{-1/2}$ gives a controlled expansion for $(G^{(N)})^{-1/2}$. Multiplying the expansions of the $D^{(N)}_i$'s by this matrix gives the announced expansions for the $I^{(N)}_a$'s. The displayed normalization of the Gram matrix gives $(I^{(N)}_a)^*I^{(N)}_b=\delta_{ab}\Id_{V_\lambda}$, hence orthonormality.
\end{proof}

In all later applications the hypothesis $G^{(\infty)}>0$ is arranged at the coordinate-selection stage. Indeed, the relevant diagrammatic families have fixed finite cardinality and controlled Laurent expansions in $1/N$. After discarding linear dependencies and, when necessary, multiplying individual coordinates by their leading powers of $N$, one chooses a maximal subfamily whose limiting Gram matrix has non-zero determinant. Its limit is then positive definite because every $G^{(N)}$ is a Gram matrix. Thus Lemma~\ref{lem:fixed-degree-orthonormalization} applies to every fixed-degree channel family used below.

\subsection{Orthonormal channel resolutions and contraction--coevaluation projectors}
\label{subsec:orthonormal-channel-recoupling}

We now fix the concrete channel-recoupling conventions used in the pinned block estimates of Section~\ref{sec:pinned-haar-projection}. Let $E_N$ be a fixed mixed tensor space, possibly with some fixed external tensor factors of bounded degree. As we have seen, it satisfies the isotypic decomposition~\eqref{eq:isotypic-decomposition}, with multiplicity spaces $M_\eta(E_N)=\Hom_{\U(N)}(V_\eta,E_N)$ for all $\eta\in\widehat{\U(N)}$. Choosing an orthonormal basis $(I_{\eta,a})_{a\in A_\eta(E_N)}$ of $M_\eta(E_N)$, viewed as isometric intertwiners $I_{\eta,a}:V_\eta\to E_N$, gives the orthogonal decomposition
\[
\Id_{E_N}=
\sum_{\eta\in\widehat{\U(N)}}\sum_{a\in A_\eta(E_N)}I_{\eta,a}I_{\eta,a}^* .
\]
We call such a choice an \emph{orthonormal Clebsch--Gordan resolution} of $E_N$. If $E_1\otimes\cdots\otimes E_r$ is a tensor product with $r>2$, a resolution is obtained by choosing a binary bracketing and applying the preceding decomposition recursively. The internal irreducible labels which appear along the binary channel tree are called \emph{internal channels}, and the corresponding basis indices are called multiplicity indices. Given two binary bracketings of the same tensor product, the corresponding orthonormal Clebsch--Gordan bases are related by a unitary change-of-bracketing matrix. In fixed mixed degree, the coefficients of these matrices, and of all intertwiners obtained by composing Clebsch--Gordan maps, duality contractions, coevaluations and algebraic diagrammatic coordinate maps, have controlled asymptotic expansions in powers of $1/N$, by Lemma~\ref{lem:fixed-degree-orthonormalization}.

\begin{example}\label{ex:CG-resolution}
\label{ex:first-change-of-bracketing}
Let $V_N=\mathbb C^N$, and assume $N\geq 3$. Consider the tensor product $E=V_N\otimes V_N\otimes V_N$. There are two natural binary Clebsch--Gordan resolutions of $E$, corresponding to the two binary bracketings
\[
((V_N\otimes V_N)\otimes V_N),\qquad (V_N\otimes (V_N\otimes V_N)).
\]
Using the classical decompositions
\[
\begin{aligned}
V_N\otimes V_N
&\simeq \operatorname{Sym}^2V_N\oplus \wedge^2V_N,\\
\operatorname{Sym}^2V_N\otimes V_N
&\simeq V_{(3)}\oplus V_{(2,1)},\\
\wedge^2V_N\otimes V_N
&\simeq V_{(2,1)}\oplus V_{(1,1,1)}.
\end{aligned}
\]
we obtain for the left bracketing
\[
V_N^{\otimes3}\simeq V_{(3)}\oplus V_{(2,1)}^{(12,+)}\oplus V_{(2,1)}^{(12,-)}\oplus V_{(1,1,1)}.
\]
The superscripts indicate whether the first two factors have first been symmetrized or antisymmetrized. The right bracketing gives similarly
\[
V_N^{\otimes3}\simeq V_{(3)}\oplus V_{(2,1)}^{(23,+)}\oplus V_{(2,1)}^{(23,-)}\oplus V_{(1,1,1)}.
\]
The two decompositions agree on the totally symmetric component $V_{(3)}$ and on the totally antisymmetric component $V_{(1,1,1)}$. The only non-trivial change of basis takes place in the two-dimensional multiplicity space of $V_{(2,1)}$.
\end{example}

If $V_\eta$ is an irreducible unitary representation and $(e_a)_{a=1}^{d_\eta}$ is an orthonormal basis of $V_\eta$, let $(e^a)_{a=1}^{d_\eta}$ be the dual basis. Define evaluation $\operatorname{ev}_\eta:V_\eta\otimes V_\eta^*\to\mathbb C$ and coevaluation $\operatorname{coev}_\eta:\mathbb C\to V_\eta\otimes V_\eta^*$ by
\[
\operatorname{ev}_\eta(v\otimes f)=f(v),\qquad\operatorname{coev}_\eta(1)=\sum_{a=1}^{d_\eta}e_a\otimes e^a.
\]
The normalized invariant vector in $V_\eta\otimes V_\eta^*$ is $\Omega_\eta=d_\eta^{-1/2}\operatorname{coev}_\eta(1)$. We denote by
\begin{equation}\label{eq:contraction-coevaluation-projector}
\Pi^{\operatorname{inv}}_\eta
:=\operatorname{Proj}_{(V_\eta\otimes V_\eta^*)^{\U(N)}}
=\Omega_\eta\Omega_\eta^*
=d_\eta^{-1}\operatorname{coev}_\eta\operatorname{ev}_\eta
\end{equation}
the associated contraction--coevaluation projector. Every insertion of the contraction--coevaluation projector associated with a channel $\eta$ carries the scalar normalization $d_\eta^{-1}$. We call this scalar the \emph{reciprocal dimension factor} of the channel.

We now record an elementary Haar-projection mechanism that will be used as a black-box estimate in the peeling argument of Section~\ref{sec:pinned-haar-projection}.

\begin{lemma}\label{lem:partial-trace-haar-projection}
Let $(\rho_\lambda,V_\lambda)$ be an irreducible unitary representation of $\U(N)$, and let $(\pi_H,H)$ be a finite-dimensional unitary representation of $\U(N)$.  Set
\[
\Pi_{\lambda,H}=\int_{\U(N)}\rho_\lambda(g)\otimes\pi_H(g) dg,
\]
the orthogonal projection onto $(V_\lambda\otimes H)^{\U(N)}$.  Then $d_\lambda\Tr_{V_\lambda}(\Pi_{\lambda,H})$ is the orthogonal projection of $H$ onto its $V_{\lambda^*}$-isotypic component. In particular,
\[
\left\|d_\lambda\Tr_{V_\lambda}(\Pi_{\lambda,H})\right\|_{2\to2}\le1.
\]
\end{lemma}

\begin{proof}
Decompose $H$ into isotypic components:
\[
H\simeq\bigoplus_{\mu\in\widehat{\U(N)}}M_\mu\otimes V_\mu .
\]
The invariant subspace of $V_\lambda\otimes H$ is zero except on the summand $M_{\lambda^*}\otimes V_{\lambda^*}$.  On this summand it is
\[
M_{\lambda^*}\otimes (V_\lambda\otimes V_{\lambda^*})^{\U(N)}.
\]
The normalized invariant vector in $V_\lambda\otimes V_{\lambda^*}$ is $\Omega_\lambda=d_\lambda^{-1/2}\sum_i e_i\otimes e_i^*$, and the orthogonal projection onto $\mathbb C\Omega_\lambda$ has partial trace $d_\lambda^{-1}\Id_{V_{\lambda^*}}$.  Therefore $d_\lambda\Tr_{V_\lambda}(\Pi_{\lambda,H})$ is the identity on the $V_{\lambda^*}$-isotypic component and zero on every other component.
\end{proof}

An opened channel is a place where an internal channel of a tensor contraction has been separated by inserting the resolution of the identity on an irreducible summand. If the opened channel has label $\eta$, then the contraction--coevaluation projector $\Pi^{\mathrm{inv}}_\eta$ contributes the reciprocal dimension factor $d_\eta^{-1}$. We shall use this convention in Section~\ref{sec:pinned-haar-projection}, where the fully resolved channel summands and their labelled channel graphs are defined in the concrete local-channel setting.

\subsection{Stable Young graphs and reciprocal dimension bounds}
\label{subsec:stable-young-graphs}

We shall repeatedly tensor irreducible $\U(N)$-representations by the fundamental representation $V_N=\C^N$ or by its dual. The relevant Clebsch--Gordan rule is the Pieri rule. Let $\lambda,\mu\in\widehat{\U(N)}$ be highest weights. We write $\lambda\nearrow \mu$ if $\mu$ is obtained from $\lambda$ by incrementing one coefficient by one:
\[
\exists i_0\in\{1,\ldots,N\}:\quad \mu_{i_0}=\lambda_{i_0}+1,\qquad \mu_i=\lambda_i\ \text{for }i\neq i_0 .
\]
Then the Pieri decomposition reads
\begin{equation}\label{eq:Pieri}
V_\lambda\otimes V_N \simeq\bigoplus_{\lambda\nearrow\mu} V_\mu,\qquad V_\lambda\otimes V_N^*\simeq\bigoplus_{\mu\nearrow\lambda} V_\mu .
\end{equation}
Taking traces of
$\rho_\lambda(x)\otimes x$ and
$\rho_\lambda(x)\otimes x^*$ in these decompositions gives the character
form
\begin{equation}\label{eq:Pieri-trace}
\Tr(x)\chi_\lambda(x)=\sum_{\lambda\nearrow\mu}\chi_\mu(x),\qquad\Tr(x^*)\chi_\lambda(x)=\sum_{\mu\nearrow\lambda}\chi_\mu(x),\qquad x\in \U(N).
\end{equation}
Let us now translate this rule into stable coordinates. Suppose first that $\lambda=[\lambda^+,\lambda^-]_N$ is reduced stable and that the resulting labels stay in the stable range. Then tensoring by $V_N$ either adds one box to the positive partition or removes one box from the negative partition\footnote{There is a boundary case when $N=\ell(\alpha)+\ell(\beta)$, but it is irrelevant in the stable range considered below.}:
\begin{equation}\label{eq:Pieri-stable}
V_{[\lambda^+,\lambda^-]_N}\otimes V_N\simeq\bigoplus_{\lambda^+\nearrow\mu^+}V_{[\mu^+,\lambda^-]_N}\ \oplus\ \bigoplus_{\mu^-\nearrow\lambda^-} V_{[\lambda^+,\mu^-]_N}.
\end{equation}
Similarly, tensoring by $V_N^*$ either removes one box from the positive partition or adds one box to the negative partition:
\begin{equation}\label{eq:Pieri-stable2}
V_{[\lambda^+,\lambda^-]_N}\otimes V_N^*\simeq\bigoplus_{\mu^+\nearrow\lambda^+}V_{[\mu^+,\lambda^-]_N}\ \oplus\ \bigoplus_{\lambda^-\nearrow\mu^-}V_{[\lambda^+,\mu^-]_N}.
\end{equation}
The same formulas hold after adding a determinant shift $q\mathbf 1_N$: the shift $q$ is unchanged, and only the reduced pair $(\lambda^+,\lambda^-)$ moves by one box. Terms which do not satisfy the length cutoffs are omitted. For partitions $\rho$ and $\sigma$, we write $\rho\nearrow\sigma$ if $\sigma$ is obtained from $\rho$ by adding one box.

\begin{definition}
Let $\lambda\in\widehat{\U(N)}$ be a highest weight, and let $(\lambda^+,\lambda^-,q)$ be its stable coordinates. A \emph{stable Pieri move} of $\lambda$ consists in adding or removing a box in either $\lambda^+$ or $\lambda^-$.
\end{definition}

We now encode these Pieri moves graph-theoretically.  Let $\mathbb Y$ be the Young graph: its vertices are partitions, and two partitions are adjacent if one is obtained from the other by adding one box.  Equivalently, $\mathbb Y$ is the Hasse graph of Young's lattice after forgetting orientations.  If $k\geq 1$, we denote by $\mathbb Y_{\le k}\subset \mathbb Y$ the induced subgraph on partitions of length at most $k$; we call it the $k$-truncated Young graph.  An illustration is given in Figure~\ref{fig:Young-Lattice}.

\begin{figure}[h!]
\centering
\includegraphics[width=0.4\linewidth]{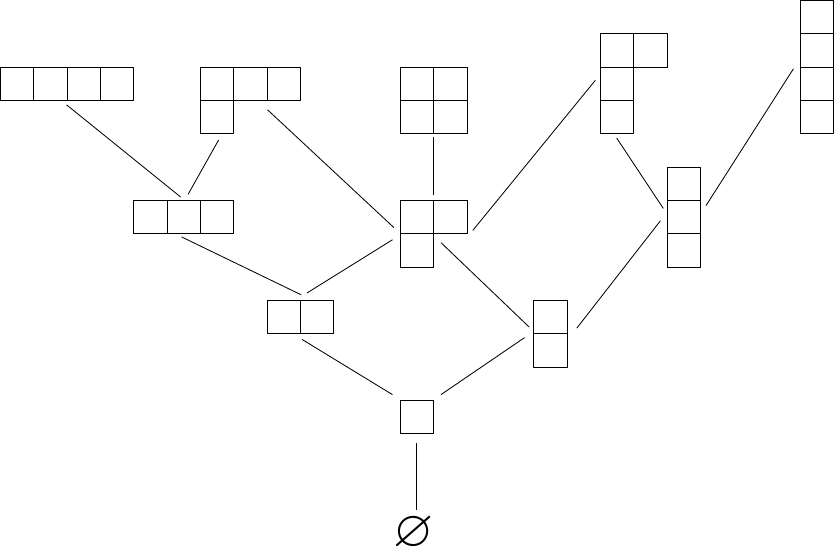}
\hspace{2em}
\includegraphics[width=0.4\linewidth]{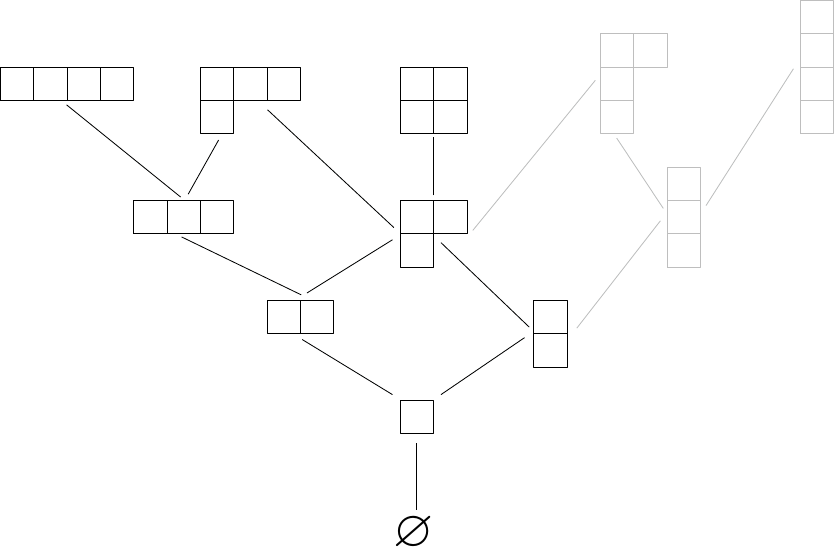}
\caption{A portion of the Young graph (left), and the corresponding portion of the $2$-truncated Young graph (right).}
\label{fig:Young-Lattice}
\end{figure}

We shall use the standard Cartesian product of graphs.  If $G$ and $H$ are two unoriented graphs, their Cartesian product $G\square H$ is the graph with vertex set
\[
V(G\square H)=V(G)\times V(H),
\]
and with adjacency relation
\[
(g,h)\sim (g',h')
\]
if and only if either $g=g'$ and $h\sim h'$ in $H$, or $h=h'$ and $g\sim g'$ in $G$.  In other words, a path in $G\square H$ is an interleaving of two paths in the factors: at each elementary step, exactly one coordinate moves and the other one remains fixed.  This is the product which is adapted to mixed tensoring by $\C^N$ and $(\C^N)^*$, since a stable Pieri move changes only one of the two stable partitions.

We define the reduced stable Young graph of rank $N$ by
\[
\mathbb Y_N^{\mathrm{red}}
:=
\mathbb Y_{\le A_N}\square\mathbb Y_{\le B_N},
\]
where $A_N$ and $B_N$ are the length-cutoffs defined in~\eqref{eq:length-cutoff}. Its vertices are pairs
\[
\gamma=(\gamma^+,\gamma^-),
\qquad
\ell(\gamma^+)\le A_N,
\qquad
\ell(\gamma^-)\le B_N,
\]
which we identify with the stable representation $[\gamma^+,\gamma^-]_N$.  Two vertices $\gamma$ and $\eta$ are adjacent if and only if exactly one of the two coordinates moves by one Young-graph step:
\[
\eta^+=\gamma^+,
\quad \eta^-\sim\gamma^- ,
\qquad\text{or}\qquad
\eta^- =\gamma^-,
\quad \eta^+\sim\gamma^+ .
\]
Thus a path in $\mathbb Y_N^{\mathrm{red}}$ is equivalently a pair of lazy paths on the truncated Young graphs $\mathbb Y_{\le A_N}$ and $\mathbb Y_{\le B_N}$, coupled by interleaving their non-trivial moves.  We shall use this graph only in the stable mass range $m(\gamma)\le N/2$.

This terminology is classical in a nearby form.  Directed paths in Young's lattice from $\varnothing$ to a partition $\lambda$ are standard Young tableaux of shape $\lambda$ and they are related to representation theory of the symmetric group \cite{Sag01}. If one allows both upward and downward steps, one obtains oscillating tableaux, also called up-down tableaux.  These paths appear naturally in the representation theory of classical groups, starting with Berele's Schensted-type correspondence for the symplectic group \cite{Ber86} and Sundaram's work on symplectic and orthogonal tableaux \cite{Sun90a,Sun90b}; closely related vacillating tableaux occur in the representation theory of partition and diagram algebras, for example in the work of Halverson--Lewandowski \cite{HalLew05}.  The paths used below are a stable mixed-tensor analogue of these objects: instead of one Young diagram moving up and down, we have a pair $(\gamma^+,\gamma^-)$, and each elementary tensoring by $\C^N$ or $(\C^N)^*$ moves exactly one of the two diagrams by one box.

\begin{definition}\label{def:oscillating-path}
An \emph{elementary stable oscillating path} is a path $\gamma_0,\gamma_1,\ldots,\gamma_r$ in the reduced stable Young graph $\mathbb Y_N^{\mathrm{red}}$. We also denote by $\gamma$ the stable highest weight $[\gamma^+,\gamma^-]_N$; in particular, $m(\gamma)=|\gamma^+|+|\gamma^-|$ is its projective mass and $d_\gamma$ its dimension.
\end{definition}

The estimate below says that reciprocal dimensions make the oscillating walk uniformly summable in the stable mass range, and it will be used in Section~\ref{sec:pinned-haar-projection}.

\begin{lemma}\label{lem:oscillating-path}
There exists a universal constant $C<\infty$ such that, for every $N$, every integer $s\ge0$, and every initial reduced stable label $\gamma_0$ with $m(\gamma_0)\le N/2$,
\begin{equation}\label{eq:oscillating-bound}
\sum_{\gamma_1,\ldots,\gamma_s}\prod_{j=1}^s d_{\gamma_j}^{-1}\le C^s,
\end{equation}
where the sum is over all elementary stable oscillating paths
$\gamma_0,\gamma_1,\ldots,\gamma_s$ staying in total mass at most $N/2$.
\end{lemma}

\begin{proof}
It is useful to think of \eqref{eq:oscillating-bound} as a weighted path-counting estimate on the graph $\mathbb Y_N^{\mathrm{red}}$.  On vertices of mass at most $N/2$, define the positive weighted adjacency operator
\[
(Rf)(\gamma)=\sum_{\eta\sim\gamma : m(\eta)\le N/2}d_\eta^{-1} f(\eta).
\]
For the constant function $\mathbf 1$, the left-hand side of \eqref{eq:oscillating-bound} is precisely $R^s\mathbf 1(\gamma_0).$ Thus it is enough to prove a uniform one-step bound
\[
\sup_{m(\gamma)\le N/2} R\mathbf 1(\gamma)\le C.
\]

Let $m=m(\gamma)$.  A partition of size $a$ has at most $a+1$ addable corners and at most $a$ removable corners.  Applying this to the two partitions $\gamma^+$ and $\gamma^-$ gives the rough bound
\[
\#\{\eta:\eta\sim\gamma\}\le 4(m+1).
\]
On the other hand, every non-zero reduced stable $\U(N)$-label has dimension at least $N$.  Indeed, the one-dimensional representations of $\U(N)$ are determinant powers, and among reduced labels the only determinant power is the zero label; equivalently, the smallest non-trivial reduced representation is the fundamental representation or its dual, both of dimension $N$.

If $\gamma=0$, then all neighbours are the fundamental or dual fundamental label, so
\[
R\mathbf 1(0)\le 2N^{-1}\le 2.
\]
If $m(\gamma)=1$, there may be one step back to the zero label, contributing $1$, and all other neighbours are non-zero, hence contribute at most $N^{-1}$ each.  Therefore
\[
R\mathbf 1(\gamma)\le 1+\frac{C}{N}\le C.
\]
Finally, if $m(\gamma)\ge 2$, every neighbour is non-zero.  Since we only consider the range $m\le N/2$,
\[
R\mathbf 1(\gamma)
\le
\frac{4(m+1)}{N}
\le C.
\]
This proves the uniform one-step estimate.  Iterating it gives
\[
R^s\mathbf 1(\gamma_0)
\le
\|R\mathbf 1\|_\infty^s
\le C^s,
\]
which is exactly \eqref{eq:oscillating-bound}.
\end{proof}

For absolute estimates outside the stable small-mass regime we use a coarser version which does not rely on stable coordinates.  Let $\mathbb P_N^{\mathrm{full}}$ be the unoriented Pieri graph on $\widehat{\U(N)}$: two highest weights are adjacent when one occurs in the tensor product of the other with $V_N$ or $V_N^*$.

\begin{lemma}[Full Pieri reciprocal-dimension bound]\label{lem:full-pieri-path}
There is a universal constant $C<\infty$ such that, for every $N$, every $s\ge0$, and every initial irreducible label $\lambda_0\in\widehat{\U(N)}$,
\begin{equation}\label{eq:full-pieri-path}
\sum_{\lambda_1,\ldots,\lambda_s}
\prod_{j=1}^s d_{\lambda_j}^{-1}
\le C^s,
\end{equation}
where the sum is over all length-$s$ paths in $\mathbb P_N^{\mathrm{full}}$ starting at $\lambda_0$.
\end{lemma}

\begin{proof}
By the Pieri rule, tensoring an irreducible representation by $V_N$ has at most $N$ irreducible summands, and the same is true for tensoring by $V_N^*$.  Hence every vertex of $\mathbb P_N^{\mathrm{full}}$ has at most $2N$ neighbours, counted without multiplicity.

The one-dimensional irreducible representations are the determinant powers $\det^q$.  A fixed highest weight can be adjacent to at most two determinant powers, one through a $V_N$-move and one through a $V_N^*$-move.  Every other irreducible representation has dimension at least $N$; this follows directly from the Weyl dimension formula, since after removing a determinant twist the smallest non-trivial dominant weight is a fundamental weight, whose representation has dimension at least $N$.

Consequently, for every $\lambda$,
\[
\sum_{\mu\sim\lambda}d_\mu^{-1}
\le 2+\frac{2N}{N}\le4.
\]
Iterating this one-step estimate gives~\eqref{eq:full-pieri-path} with $C=4$.
\end{proof}

\section{\texorpdfstring{Finite-$N$}{Finite-N} Fourier expansion and local channels}
\label{sec:finiteN}

This section collects the finite-$N$ objects imported from~\cite{Lem26a} and puts them in the form used by the large-$N$ estimates. In particular, we rewrite the Wilson loop expectation as a ratio of expectations of observables of random plaquette decorations $\lambda:P(\Lambda)\to\widehat{\U(N)}$.

\subsection{State-sum formula}

Given a plaquette decoration $\lambda:P(\Lambda)\longrightarrow \widehat{\U(N)},$ define the spectral coefficient
\begin{equation}\label{eq:kappa-def}
\kappa_{\Lambda,Q}(\lambda):=\prod_{p\in P(\Lambda)}\widehat Q_{1/\beta_p}(\lambda_p),
\end{equation}
and the universal topological coefficient
\begin{equation}\label{eq:kappa-W}
\widehat W_{\Lambda,L}(\lambda):=\int_{\U(N)^{E(\Lambda)}}\prod_{i=1}^k\Tr(U_{\ell_i})\prod_{p\in P(\Lambda)}\chi_{\lambda_p}(U_{\partial p}) dU .
\end{equation}
The case $L=\varnothing$ is the vacuum topological coefficient. The state-sum expansion~\cite[Theorem~1.1]{Lem26a} reads
\begin{equation}\label{eq:state-sum}
\E_{\Lambda,Q}[W_{\Lambda,L}]=\frac{1}{Z}\sum_{\lambda:P(\Lambda)\to\widehat{\U(N)}}\kappa_{\Lambda,Q}(\lambda)\widehat W_{\Lambda,L}(\lambda).
\end{equation}
It follows naturally that the normalized Wilson expectation is
\begin{equation}\label{eq:normalized-state-sum}
\Phi_{\Lambda,Q,N}(L)
:=N^{-\#L}\E_{\Lambda,Q}[W_{\Lambda,L}]
=\frac{1}{Z}
\sum_{\lambda}
\kappa_{\Lambda,Q}(\lambda)N^{-\#L}\widehat W_{\Lambda,L}(\lambda).
\end{equation}

Two comments will be used repeatedly.  First, $\kappa_{\Lambda,Q}$ is a Boltzmann weight on plaquette representation fields.  Second, $\widehat W_{\Lambda,L}(\lambda)$ is a universal invariant integral: it depends on $\Lambda$, $L$ and $\lambda$, but not on the chosen plaquette action.  This action/topology separation is the very starting point of~\cite{Lem26a}. For this reason, all statements obtained for topological coefficients do not depend specifically on the fact that we work with the heat-kernel action; in particular, the pinned Haar-projection estimates from Section~\ref{sec:pinned-haar-projection} are action-agnostic. Throughout this paper, $\lambda$ denotes a full highest-weight field, while $\alpha$ will be reserved for its determinant-reduced stable part when we pass to stable coordinates. In the heat-kernel case with constant time $T$, we will abbreviate
\begin{equation}\label{eq:heat-kernel-kappa}
\kappa_{\Lambda,T,N}(\lambda):=\prod_{p\in P(\Lambda)}\widehat Q^{\HK}_{T,N}(\lambda_p).
\end{equation}

\subsection{Central charge constraints}

In this subsection we record the simple conservation law imposed by the central characters of the plaquette representations. Let us first introduce some terminology.

\begin{definition}
Let $\lambda\in\widehat{\U(N)}$ be a highest weight, and let $(\lambda^+,\lambda^-,q)$ be its stable coordinates. Its \emph{determinant charge} is the integer $Nq$, its \emph{reduced central charge} is the integer $s(\lambda)=\vert\lambda^+\vert-\vert\lambda^-\vert$, and its \emph{full central charge} is the integer $c(\lambda)=Nq+s(\lambda).$
\end{definition}

Let $j_L(e)$ be the signed number of traversals of the oriented edge $e$ by the loop family $L$.  If a plaquette $p$ carries the label $\lambda_p=\lambda_N(\lambda_p^+,\lambda_p^-,q_p)$, set $c_p=c(\lambda_p)$ and $s_p=s(\lambda_p)$. Define the signed plaquette-edge incidence operator $\partial^*:\Z^{P(\Lambda)} \longrightarrow \Z^{E(\Lambda)}$ by
\[
(\partial^*f)(e)=\sum_{p\in P(\Lambda)}\varepsilon(e,p)f(p).
\]

\begin{lemma}[Central selection rule]\label{lem:central-selection}
If $\widehat W_{\Lambda,L}(\lambda)\neq0$, then
\begin{equation}\label{eq:central-selection}
\partial^*c+j_L=0.
\end{equation}
If moreover $2dM+\ell(L)<N$, where $M=\sum_p m(\lambda_p)=\sum_p(|\lambda_p^+|+|\lambda_p^-|)$, then \eqref{eq:central-selection} is equivalent to
\begin{equation}\label{eq:split-selection}
\partial^*q=0,
\qquad
\partial^*s+j_L=0.
\end{equation}
\end{lemma}

\begin{proof}
Replace, in the integral \eqref{eq:kappa-W}, each edge variable by $z_eU_e$, with $z_e\in\U(1)$.  Haar invariance leaves the value of the integral unchanged.  The Wilson insertion is multiplied by $\prod_e z_e^{j_L(e)}$.  The plaquette character at $p$ is multiplied by the scalar character $z_{\partial p}^{c_p}$, because the determinant shift contributes charge $Nq_p$ and the rational reduced part contributes charge $s_p$.  Averaging independently over the scalars $z_e$ kills the integral unless the exponent of every $z_e$ vanishes, which is exactly \eqref{eq:central-selection}.

If $2dM+\ell(L)<N$, then for every edge $e$ one has $| (\partial^*s)_e+j_L(e)|<N$.  Since \eqref{eq:central-selection} says
\[
N(\partial^*q)_e+(\partial^*s)_e+j_L(e)=0,
\]
the multiple of $N$ must vanish and the remaining term must vanish.  This proves \eqref{eq:split-selection}.
\end{proof}

This is the only charge constraint needed at this stage.  Its heat-kernel cutoff consequence will be stated after active supports and projective masses have been defined in Subsection~\ref{subsec:pinned-block-observables}.  We now turn from this global charge constraint to the local-channel representation of the topological coefficients.

\subsection{Dual incidence graph and local channels}\label{subsec:local-channels}

The coefficient $\widehat W_{\Lambda,L}(\lambda)$ is an integral of products of matrix coefficients over one independent Haar variable per edge. Its combinatorial study should be viewed as a refined version of Weingarten calculus. In the latter, Haar integration produces permutations or pairings between tensor indices. Here the plaquettes carry arbitrary irreducible representations, so Haar integration produces orthogonal projections onto invariant tensors and multiplicity spaces.  A \emph{local channel} is the finite collection of choices which makes all these projections explicit.

We start by gauge-fixing redundant variables: fix a spanning tree $\mathcal T\subset E(\Lambda)$ and set all edge variables $U_e$ to $I_N$, for $e\in\mathcal T$. Define the dual incidence graph $D_{\mathcal T}(\Lambda)$ as the bipartite graph with vertex set
\[
V(D_{\mathcal T}(\Lambda))=P(\Lambda)\sqcup(E(\Lambda)\setminus\mathcal T),\]
and with an incidence $(p,e)$ whenever $e$ occurs in the boundary of $p$.  Plaquette vertices carry representation labels; non-tree edge vertices carry Haar projections. An example of such graph is displayed in Figure~\ref{fig:dual-incidence}: the bold black graph is the chosen spanning tree $\mathcal T$, the red disks are plaquette vertices of the dual incidence graph $D_{\mathcal T}(\Lambda)$, the grey squares are non-tree edge vertices, and the blue segments are incidence edges.  The orange rectangle represents a Wilson loop insertion $\ell$.

\begin{figure}[h!]
\centering
\includegraphics[width=0.3\linewidth]{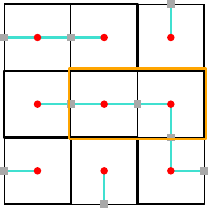}
\caption{The dual incidence graph of a lattice $\Lambda\subset\Z^2$ with a Wilson loop insertion.}
\label{fig:dual-incidence}
\end{figure}

The construction proceeds locally: a plaquette character $\chi_{\lambda_p}(U_{\partial p})$ is resolved into tensor legs distributed along the boundary incidences $(p,e)$, then each non-tree edge variable $U_e$ is integrated independently.  The integral at $e$ is the orthogonal projection onto the invariant subspace of the tensor product of all representation spaces incident to $e$, together with the possible Wilson strands crossing $e$.  Schematically, an edge kernel has the form
\[
\Pi_e
=\int_{\U(N)}
\bigotimes_{(p,e)}\rho_{p,e}(g^{\varepsilon(p,e)})
\otimes
\bigotimes_{\text{Wilson legs at }e} g^{\pm1} dg,
\]
after the incident plaquette tensors have been placed in suitable mixed tensor spaces.  Choosing bases in the relevant multiplicity spaces turns these projections into scalar local kernels.

The local-channel representation~\cite[Theorem~1.3]{Lem26a} says that, for every plaquette label field $\lambda$, the topological coefficient $\widehat{W}_{\Lambda,L}(\lambda)$ can be decomposed on a finite set $\mathcal C_\Lambda^L(\lambda)$ of local channel fields $\Gamma$.  A channel field records, at every incidence $(p,e)$ and every non-tree edge $e$, the local resolution of tensor contractions coming from the integration with respect to the edge variable $U_e$. The representation can be stated as follows: there are scalar plaquette coefficients $A_p^N(\Gamma_{\partial p};\lambda_p)$ and local edge kernels $K_e^N(\Gamma_e;\lambda,L)$ such that
\begin{equation}\label{eq:local-channel}
\widehat W_{\Lambda,L}(\lambda)=\sum_{\Gamma\in\mathcal C^L_\Lambda(\lambda)}\left(\prod_{p\in P(\Lambda)}A_p^N(\Gamma_{\partial p};\lambda_p)\right)\left(\prod_{e\in E(\Lambda)\setminus\mathcal T}K^N_e(\Gamma_e;\lambda,L)\right).
\end{equation}
Two points will be used repeatedly:
\begin{itemize}
\item The local-channel representation depends on a choice of spanning tree $\mathcal T$, but the topological coefficient does not: the representation is gauge-fixed but the coefficient itself is gauge-invariant.
\item The vacuum data are allowed: if $\lambda_p$ is trivial and no Wilson strand interacts with $p$, the scalar plaquette coefficient is the vacuum coefficient; similarly, an edge kernel can be the vacuum scalar. Before scalarization, these coefficients come from plaquette tensors and edge Haar projections.
\end{itemize}
An illustration of local channels is given in Figure~\ref{fig:local-channels}: plaquette-edge incidences are filled with tensor legs coming out of plaquettes, and Haar integration over adjacent edges results in pairings of tensor legs, eventually with other tensor legs coming from the Wilson loop insertion (in orange). If we zoom in the dual incidence graph from Figure~\ref{fig:dual-incidence} and assign local channels to all incidences, we get Figure~\ref{fig:dual-incidence-2}. Note that, in this case, there are two loop insertions through non-tree edges.

\begin{figure}[h!]
\centering
\includegraphics[width=0.45\linewidth]{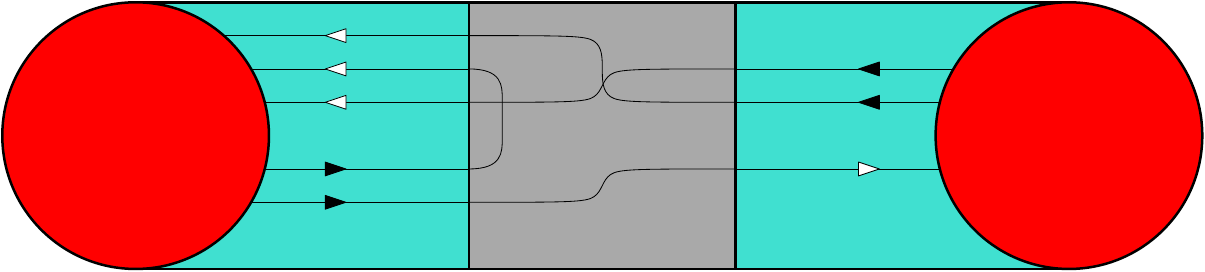}
\includegraphics[width=0.45\linewidth]{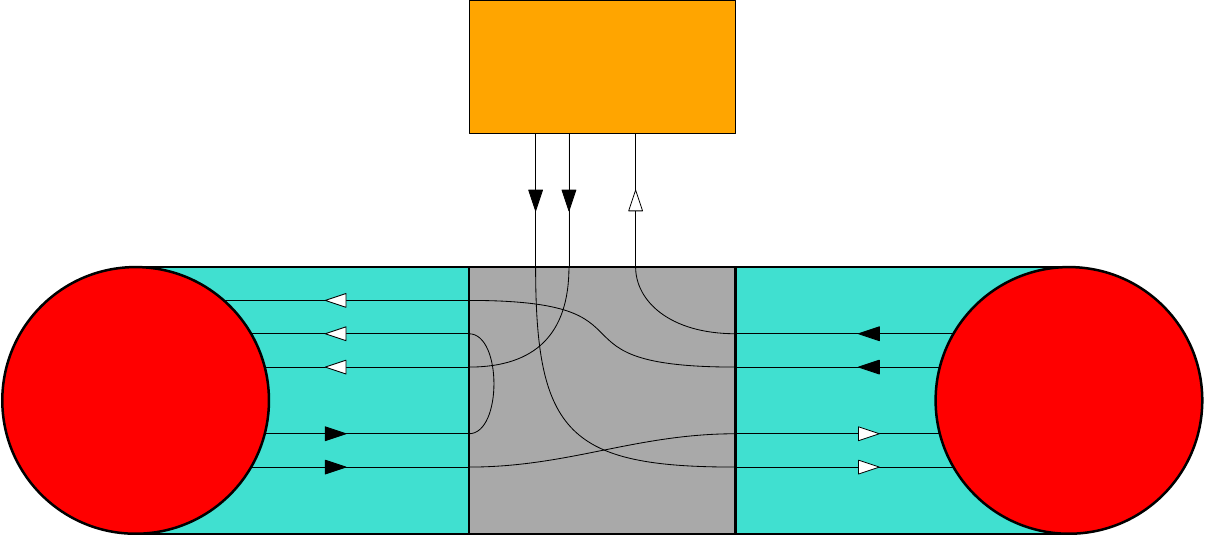}
\caption{Two local channel examples: one with no Wilson loop insertion (left) and one with a Wilson loop insertion (right).}
\label{fig:local-channels}
\end{figure}

\begin{figure}[h!]
\centering
\includegraphics[width=0.4\linewidth]{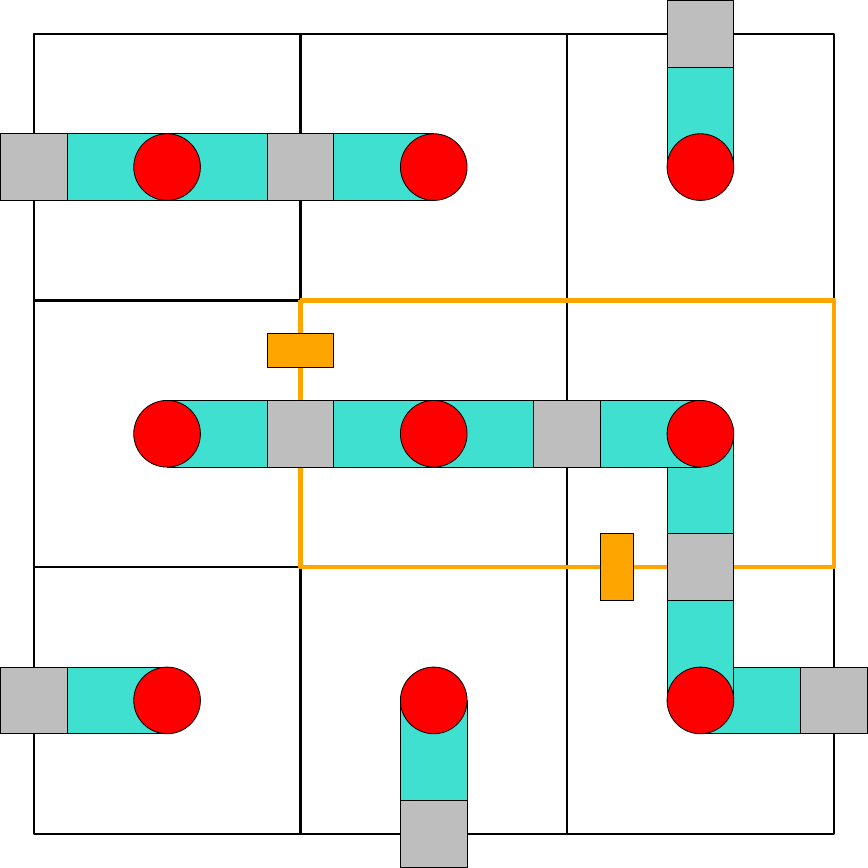}
\caption{A local channel representation of the dual incidence graph illustrated in Figure~\ref{fig:dual-incidence}.}
\label{fig:dual-incidence-2}
\end{figure}

\subsection{Local channel expansion revisited}\label{subsec:trace-topological-coefficients}

We now reproduce the proof of the local-channel expansion~\eqref{eq:local-channel} slightly differently from~\cite{Lem26a}, writing the scalar edge kernels as traces of endomorphisms of explicit Hilbert spaces. We work in the determinant-reduced stable sector; determinant shifts will be reinserted later through the central selection rule and through the reference law $\widehat\Pbb_{T,N}$.

Fix a finite lattice $\Lambda$, a spanning tree $\mathcal T\subset E(\Lambda)$ and a loop family $L=(\ell_1,\ldots,\ell_k)$.  Let
\[
\alpha:P(\Lambda)\longrightarrow \widehat{\U(N)}_{\mathrm{red},\mathrm{st}}
\]
be a reduced stable plaquette decoration.
After gauge fixing, all tree-edge variables are set equal to $I_N$. Thus the topological coefficient becomes
\[
\widehat W_{\Lambda,L}(\alpha)=\int_{\U(N)^{E(\Lambda)\setminus\mathcal T}}\prod_{i=1}^k \Tr(U_{\ell_i})\prod_{p\in P(\Lambda)} \chi_{\alpha_p}(U_{\partial p})\prod_{e\in E(\Lambda)\setminus\mathcal T} dU_e .
\]
Tree variables which occur in a plaquette boundary or in a loop word are omitted from the notation below.

We shall use the following convention. If $V$ is a finite-dimensional Hilbert space and $n\geq 1$, define the cyclic operator
\[
\operatorname{Cyc}_V^{(n)}(v_1\otimes\cdots\otimes v_n)=v_n\otimes v_1\otimes\cdots\otimes v_{n-1}.
\]
For $n=0$, we set $V^{\otimes 0}=\C$ and
$\operatorname{Cyc}_V^{(0)}=\Id_{\C}$. Then
\[
\Tr_{V^{\otimes n}}\Bigl(\operatorname{Cyc}_V^{(n)}(A_1\otimes\cdots\otimes A_n)\Bigr)=\Tr_V(A_1\cdots A_n).
\]

Let $p\in P(\Lambda)$, and write the gauge-fixed boundary word as
\[
\partial^{\mathcal T} p=e_{p,1}^{\varepsilon_{p,1}}\cdots e_{p,b_p}^{\varepsilon_{p,b_p}}, \qquad e_{p,r}\in E(\Lambda)\setminus\mathcal T,\quad \varepsilon_{p,r}\in\{+1,-1\}.
\]
Write $\alpha_p=[\lambda_p^+,\lambda_p^-]_N,$ $n_p^\pm=|\lambda_p^\pm|,$ $T_{\alpha_p}=T^{n_p^+,n_p^-}_N.$ Let $P_{\alpha_p}=P_N^{[\lambda_p^+,\lambda_p^-]}$ be the mixed Schur--Weyl projector defined in subsection~\ref{subsec:mixed-schur-weyl}. By the trace realization of stable characters,
\[
\chi_{\alpha_p}(U_{\partial p})=\Tr_{T_{\alpha_p}}\Bigl(P_{\alpha_p}\rho_{n_p^+,n_p^-}(U_{e_{p,1}}^{\varepsilon_{p,1}})\cdots\rho_{n_p^+,n_p^-}(U_{e_{p,b_p}}^{\varepsilon_{p,b_p}})\Bigr).
\]
Since $P_{\alpha_p}$ commutes with the $\U(N)$-action and is a projection, we insert a copy of $P_{\alpha_p}$ between consecutive boundary factors. By Proposition~\ref{prop:expansion-projector}, we have
\[
P_{\alpha_p}=\sum_{\tau\in \mathcal B_{n_p^+,n_p^-}}c_N^{\alpha_p}(\tau)\rho_N(\tau).
\]
After expanding each copy of $P_{\alpha_p}$, the character becomes a finite sum indexed by tuples $\tau_p=(\tau_{p,1},\ldots,\tau_{p,b_p}) \in \mathcal B_{n_p^+,n_p^-}^{ b_p}.$ We now distribute the resulting elementary mixed-tensor legs among the plaquette-edge incidences. Put $r_p=n_p^++n_p^-$, and let
\[
\kappa_p(u)=
\begin{cases}
+, & 1\leq u\leq n_p^+,\\
-, & n_p^+<u\leq r_p,
\end{cases}
\]
where $V_N^+=V_N$ and $V_N^-=V_N^*$. The effective type of the elementary leg $(r,u)$ is obtained by reversing the type when the boundary occurrence is inverse:
\[
\eta_{p,r,u}=
\begin{cases}
\kappa_p(u), & \varepsilon_{p,r}=+1,\\
-\kappa_p(u), & \varepsilon_{p,r}=-1.
\end{cases}
\]
Here $-(+)=-$ and $-(-)=+$. For $e\in\partial^{\mathcal T} p$, set
\[
S_{p,e}=\{(r,u):1\leq r\leq b_p,\ 1\leq u\leq r_p,\ e_{p,r}=e\},
\]
ordered lexicographically, and define the incidence space
\[
\mathcal H_{p,e}(\tau_p;\alpha_p):=\bigotimes_{(r,u)\in S_{p,e}} V_N^{\eta_{p,r,u}} .
\]
This tensor product depends only on $\alpha_p$, on the gauge-fixed boundary word, and on the incidence $(p,e)$; the notation keeps $\tau_p$ because the operator inserted on this space depends on $\tau_p$. Set
\[
\mathcal K_p(\tau_p;\alpha_p):=\bigotimes_{e\in\partial^{\mathcal T} p} \mathcal H_{p,e}(\tau_p;\alpha_p).
\]
After the canonical reordering from the boundary-occurrence ordering to the edge-incidence ordering, the previous walled-Brauer expansion gives
\[
\chi_{\alpha_p}(U_{\partial p})=\sum_{\tau_p}c_N^{\alpha_p}(\tau_p)\Tr_{\mathcal K_p(\tau_p;\alpha_p)}\left[\Phi^N_{p,\tau_p}\left(\bigotimes_{e\in\partial^{\mathcal T} p}\pi^N_{p,e,\tau_p}(U_e)\right)\right],
\]
where $c_N^{\alpha_p}(\tau_p)=\prod_{j=1}^{b_p}c_N^{\alpha_p}(\tau_{p,j})$, $\Phi^N_{p,\tau_p}\in\End(\mathcal K_p(\tau_p;\alpha_p))$, and $\pi^N_{p,e,\tau_p}$ is the tensor product of the tautological representations on the elementary $V_N$- and $V_N^*$-legs of $\mathcal H_{p,e}(\tau_p;\alpha_p)$.

Choose, for every incidence $(p,e)$, finite orthonormal families in the spaces $\End\bigl(\mathcal H_{p,e}(\tau_p;\alpha_p)\bigr),$ induced by the orthonormal channel coordinates of Subsection~\ref{subsec:orthonormal-channel-recoupling}. Expanding $c_N^{\alpha_p}(\tau_p)\Phi^N_{p,\tau_p}$ in the corresponding tensor-product coordinates gives
\begin{equation}\label{eq:decomp-cphi}
\begin{aligned}
c_N^{\alpha_p}(\tau_p)\Phi^N_{p,\tau_p}
&=\sum_{\iota_{\partial p}}
a_p^N(\tau_p,\iota_{\partial p};\alpha_p)\\
&\qquad\times
\bigotimes_{e\in\partial^{\mathcal T}p}
C^N_{p,e}(\tau_p,\iota_{p,e};\alpha_p).
\end{aligned}
\end{equation}
where $\iota_{\partial p}=(\iota_{p,e})_{e\in\partial^{\mathcal T} p}$ runs over a finite set, and
\[
C^N_{p,e}(\tau_p,\iota_{p,e};\alpha_p)
\in\End\bigl(\mathcal H_{p,e}(\tau_p;\alpha_p)\bigr).
\]

\begin{definition}
A \emph{local plaquette resolution} is a pair $r_p=(\tau_p,\iota_{\partial p})$ for which $a_p^N(\tau_p,\iota_{\partial p};\alpha_p)\neq 0$ in~\eqref{eq:decomp-cphi}. We denote the finite set of such resolutions by $\mathcal R_p(\alpha_p)$, and write
\[
a_p^N(r_p;\alpha_p):=a_p^N(\tau_p,\iota_{\partial p};\alpha_p),\qquad C^N_{p,e}(r_p;\alpha_p):=C^N_{p,e}(\tau_p,\iota_{p,e};\alpha_p),
\]
\[
\mathcal H_{p,e}(r_p;\alpha_p):=\mathcal H_{p,e}(\tau_p;\alpha_p),\qquad\pi^N_{p,e,r_p}:=\pi^N_{p,e,\tau_p}.
\]
\end{definition}

\begin{definition}
The \emph{incidence channel} associated with $r_p$ and $e\in\partial^{\mathcal T} p$ is the finite record
\[
\Gamma_{p,e}(r_p;\alpha_p)=\bigl(\mathcal H_{p,e}(r_p;\alpha_p),C^N_{p,e}(r_p;\alpha_p)\bigr),
\]
together with the ordered $V_N/V_N^*$-types inherited from the incidence $(p,e)$. The \emph{boundary-channel tuple} of $r_p$ is
\[
\Gamma_{\partial p}(r_p;\alpha_p):=\bigl(\Gamma_{p,e}(r_p;\alpha_p)\bigr)_{e\in\partial^{\mathcal T} p}.
\]
\end{definition}
With the notations above, we find
\begin{align*}
\chi_{\alpha_p}(U_{\partial p})
={}&\sum_{r_p\in\mathcal R_p(\alpha_p)}a_p^N(r_p;\alpha_p)
\Tr_{\mathcal K_p(r_p;\alpha_p)}\Bigg[
\left(\bigotimes_{e\in\partial^{\mathcal T}p}C^N_{p,e}(r_p;\alpha_p)\right)\\
&\hspace{5cm}\times
\left(\bigotimes_{e\in\partial^{\mathcal T}p}\pi^N_{p,e,r_p}(U_e)\right)
\Bigg].
\end{align*}
where
\[
\mathcal K_p(r_p;\alpha_p)=\bigotimes_{e\in\partial^{\mathcal T} p}\mathcal H_{p,e}(r_p;\alpha_p).
\]
We next add the Wilson factors. If the gauge-fixed word of $\ell_i$ is
\[
\ell_i=f_{i,1}^{\delta_{i,1}}\cdots f_{i,n_i}^{\delta_{i,n_i}},\qquad f_{i,r}\in E(\Lambda)\setminus\mathcal T,\quad \delta_{i,r}\in\{+1,-1\},
\]
set
\[
\mathcal W_e(L):=\bigotimes_{(i,r): f_{i,r}=e} V_N^{\delta_{i,r}},
\]
with $\mathcal W_e(L)=\C$ if $e$ does not occur in the gauge-fixed loop family. The edge variable acts on $\mathcal W_e(L)$ through the tensor product representation
\[
\pi^W_{e,N}(g):=\bigotimes_{(i,r): f_{i,r}=e} g^{\delta_{i,r}},
\]
where $g^{+1}$ acts on $V_N$ and $g^{-1}$ acts on $V_N^*$.

The cyclic Wilson trace tensor is a fixed operator
\[
C^W_{\Lambda,L,N}\in
\End\left(\bigotimes_{e\in E(\Lambda)\setminus\mathcal T}\mathcal W_e(L)\right).
\]
Choose once and for all a finite decomposable expansion
\[
C^W_{\Lambda,L,N}=\sum_{\omega\in\Omega^W_{\Lambda,L}}\bigotimes_{e\in E(\Lambda)\setminus\mathcal T} C^W_{e,\omega,N},\qquad C^W_{e,\omega,N}\in\End(\mathcal W_e(L)),
\]
absorbing the scalar coefficient of each decomposable term into one of the factors $C^W_{e,\omega,N}$. If $L=\emptyset$, this expansion consists of the single scalar term $1$.

A scalar local-channel field is obtained by choosing a local plaquette resolution $r_p\in\mathcal R_p(\alpha_p)$ for every plaquette $p$, together with a Wilson scalarization index $\omega\in\Omega^W_{\Lambda,L}$, and then grouping the data by non-tree edge. For each non-tree edge $e$, set $\Gamma_e:=\left(\bigl(\Gamma_{p,e}(r_p;\alpha_p)\bigr)_{p:\,e\in\partial^{\mathcal T} p},\omega\right).$ The associated edge space is
\[
\mathcal H_e(\Gamma_e;\alpha,L):=\left(\bigotimes_{p: e\in\partial^{\mathcal T} p}\mathcal H_{p,e}(r_p;\alpha_p)\right)\otimes \mathcal W_e(L),
\]
and the edge representation is
\[
\pi^\Gamma_{e,N}(g):=\left(\bigotimes_{p: e\in\partial^{\mathcal T} p}\pi^N_{p,e,r_p}(g)\right)\otimes \pi^W_{e,N}(g).
\]
The edge Haar projection is
\[
\Pi^\Gamma_{e,N}(\alpha,L):=\int_{\U(N)}\pi^\Gamma_{e,N}(g) dg=\Proj_{\mathcal H_e(\Gamma_e;\alpha,L)^{U(N)}} .
\]
The non-Haar part of the edge tensor is
\[
C_{e,N}(\Gamma_e;\alpha,L):=\left(\bigotimes_{p: e\in\partial^{\mathcal T} p}C^N_{p,e}(r_p;\alpha_p)\right)\otimes C^W_{e,\omega,N}\in\End\bigl(\mathcal H_e(\Gamma_e;\alpha,L)\bigr).
\]
We define the scalar edge kernel by
\[
K^N_e(\Gamma_e;\alpha,L):=\Tr_{\mathcal H_e(\Gamma_e;\alpha,L)}\bigl(C_{e,N}(\Gamma_e;\alpha,L) \Pi^\Gamma_{e,N}(\alpha,L)\bigr).
\]
For an admissible plaquette boundary-channel tuple $\eta_{\partial p}$, define the compressed scalar plaquette coefficient
\[
A_p^N(\eta_{\partial p};\alpha_p):=\sum_{\substack{r_p\in\mathcal R_p(\alpha_p)\\ \Gamma_{\partial p}(r_p;\alpha_p)=\eta_{\partial p}}}a_p^N(r_p;\alpha_p).
\]
Equality of boundary-channel tuples includes equality of the incidence spaces, ordered $V_N/V_N^*$-types and incidence coordinate operators. In particular, the
edge kernels are unchanged by this grouping.

Let $\mathcal C^L_\Lambda(\alpha)$ be the finite set of scalar local-channel fields $\Gamma$ obtained in this way, after the preceding grouping of local plaquette resolutions by boundary-channel tuples. Substituting the decomposable plaquette expansions and the decomposable Wilson expansion into the gauge-fixed integral, then integrating independently over the non-tree edge variables, gives the local-channel expansion~\eqref{eq:local-channel}.

\subsection{Block observables and pinned supports}\label{subsec:pinned-block-observables}\label{subsec:local-channel-expansion}

Let us now pass from local-channel formula to the connected pinned blocks which will be estimated in Section~\ref{sec:pinned-haar-projection}. Again, we work first in the determinant-reduced stable sector $\widehat{\U(N)}_{\mathrm{red},\mathrm{st}}$: a reduced stable label is a highest weight of the form $\lambda=[\lambda^+,\lambda^-]_N$, with $m(\lambda)=|\lambda^+|+|\lambda^-|\le N/2$.  Determinant shifts will be reinserted only after the determinant-reduced block estimate has been proved.

The projection-supported local-channel expansion of Subsection~\ref{subsec:trace-topological-coefficients} gives more than the scalar formula~\eqref{eq:local-channel}: after fixing a scalar channel $\Gamma$, every edge kernel is the trace of a deterministic local contraction operator against an orthogonal Haar projection.  For a full plaquette decoration $\lambda:P(\Lambda)\to\widehat{\U(N)}$ and a channel $\Gamma\in\mathcal C_\Lambda^L(\lambda)$, set
\[
\mathcal H^\Gamma_\Lambda(\lambda,L)
:=
\bigotimes_{e\in E(\Lambda)\setminus\mathcal T}
\mathcal H_e(\Gamma_e;\lambda,L),
\qquad
\Pi^\Gamma_{\Lambda,N}(\lambda,L)
:=
\bigotimes_{e\in E(\Lambda)\setminus\mathcal T}
\Pi^\Gamma_{e,N}(\lambda,L),
\]
and
\[
C^\Gamma_{\Lambda,N}(\lambda,L):=\bigotimes_{e\in E(\Lambda)\setminus\mathcal T}C_{e,N}(\Gamma_e;\lambda,L).
\]
Thus
\[
\begin{aligned}
\prod_{e\in E(\Lambda)\setminus\mathcal T}K_e^N(\Gamma_e;\lambda,L)
&=\Tr_{\mathcal H^\Gamma_\Lambda(\lambda,L)}\Bigl(
C^\Gamma_{\Lambda,N}(\lambda,L)\\
&\hspace{4.5cm}\times
\Pi^\Gamma_{\Lambda,N}(\lambda,L)\Bigr).
\end{aligned}
\]
The normalized non-spectral local-channel summand of the Fourier-side observable is
\begin{equation}\label{eq:observable-channel-summand}
\begin{aligned}
\mathcal O^L_{\Lambda,N}(\lambda;\Gamma)
:={}&N^{-\#L}
\left(
\prod_{p\in P(\Lambda)}d_{\lambda_p}A_p^N(\Gamma_{\partial p};\lambda_p)
\right)\\
&\times
\Tr_{\mathcal H^\Gamma_\Lambda(\lambda,L)}
\left(
\mathcal C^\Gamma_{\Lambda,N}(\lambda,L)
\Pi^\Gamma_{\Lambda,N}(\lambda,L)
\right).
\end{aligned}
\end{equation}
We then set
\begin{equation}\label{eq:local-channel-observable}
\mathcal O^L_{\Lambda,N}(\lambda)
:=
\sum_{\Gamma\in\mathcal C_\Lambda^L(\lambda)}
\mathcal O^L_{\Lambda,N}(\lambda;\Gamma).
\end{equation}
The coefficients $A_p^N$, $K_e^N$, and the operators $C_{e,N}$, $\Pi^\Gamma_{e,N}$ are topological quantities: they do not contain the heat-kernel Casimir exponential.  The dimension factors $d_{\lambda_p}$ and the Wilson normalization $N^{-\#L}$ have been included in $\mathcal O^L_{\Lambda,N}$, because this is the normalization that appears in Wilson expectations. Recall that we denote by $\widehat\E_{\Lambda,T,N}$ expectation with respect to the reference measure $\widehat\Pbb_{\Lambda,T,N}$ defined in~\eqref{eq:product-ref-laws}. Combining the heat-kernel state-sum formula with \eqref{eq:local-channel-observable}, the normalized Wilson expectation may be written as the Fourier-side expectation ratio
\begin{equation}\label{eq:observable-reference-ratio}
\Phi_{\Lambda,T,N}(L)
=
\frac{
\widehat\E_{\Lambda,T,N}\bigl[\mathcal O^L_{\Lambda,N}(\lambda)\bigr]
}{
\widehat\E_{\Lambda,T,N}\bigl[\mathcal O^\varnothing_{\Lambda,N}(\lambda)\bigr]
}.
\end{equation}
Indeed, the factor $\mathcal Z_{T,N}^{|P(\Lambda)|}$ produced by the change of measure cancels between numerator and vacuum denominator. In the stable small-mass range, writing $\lambda_p=q_p\mathbf 1_N+\alpha_p$ with $\alpha_p$ reduced stable, the central selection rule separates determinant shifts from reduced labels.  The determinant shifts remain only in the Casimir factor, while $d_{\lambda_p}=d_{\alpha_p}$.  Thus the non-spectral quantity to control is the determinant-reduced connected part of the observable $\mathcal O^L_{\Lambda,N}$, with the Casimir exponential removed.

A scalar channel $\Gamma$ decomposes its non-vacuum data into connected components in the dual incidence graph.  Components which do not meet the Wilson insertion are vacuum components and cancel in the expectation ratio~\eqref{eq:observable-reference-ratio}. We therefore estimate only connected components pinned to the Wilson insertion.

We use the following local terminology.  After gauge fixing, a \emph{marked Wilson occurrence} is an occurrence of a non-tree edge, or its inverse, in one of the gauge-fixed loop words.  A non-tree edge vertex is \emph{marked} if it carries at least one marked Wilson occurrence.  A \emph{marked trace component} is one of the normalized factors $N^{-1}\Tr(U_{\ell_i})$ which contributes at least one marked occurrence to the active support under consideration.  The \emph{Wilson collar} at an edge is the ordered list of its marked Wilson occurrences, their signs, and the cyclic trace-pairing data inherited from the fixed tensor $C^W_{\Lambda,L,N}$.

The vacuum local-channel datum is the one in which the plaquette label is trivial, all incidence spaces are $\C$, all incidence coordinate operators are the identity scalar, and every unmarked edge Haar projection is the scalar projection $1$ on $\C$.  A plaquette coordinate, incidence coordinate or edge coordinate is called non-vacuum precisely when it differs from this datum; Wilson-collar data are always recorded separately.

\begin{definition}\label{def:active-support}
Let $\alpha$ be a plaquette decoration, let $\Gamma\in\mathcal C_\Lambda^L(\alpha)$, and let $L$ be a loop family. The \emph{active incidence graph} $\mathcal I(\alpha,\Gamma;L)\subset D_{\mathcal T}(\Lambda)$ is the subgraph generated by the following vertices and incidences:
\begin{itemize}
\item a plaquette vertex $p$ is active if the corresponding plaquette coordinate is not the vacuum coordinate;
\item an edge vertex $e\in E(\Lambda)\setminus\mathcal T$ is active if the corresponding edge Haar projection is not the vacuum scalar, or if $e$ carries a marked Wilson occurrence;
\item an incidence $(p,e)$ is active if the scalar local-channel coordinate has a non-vacuum dependence between the plaquette coordinate at $p$ and the edge projection at $e$.
\end{itemize}
A \emph{connected active support} is a connected component of $\mathcal I(\alpha,\Gamma;L)$, and it is pinned if it contains at least one edge vertex carrying a marked Wilson occurrence.

Let $Y$ be a connected active support.  We denote by $P(Y)$ its plaquette vertices, by $V_E(Y)$ its edge vertices, and by $I(Y)$ its active incidences.  We set
\[
|Y|:=|P(Y)|+|V_E(Y)|+|I(Y)|.
\]
The bounded degree of $\Z^d$ implies that any of these three quantities controls the other two up to constants depending only on $d$, after allowing a loop-dependent additive error for marked Wilson occurrences.
\end{definition}

For a reduced stable plaquette decoration $\alpha_Y:P(Y)\to\widehat{\U(N)}_{\mathrm{red},\mathrm{st}}$, let $\bar\alpha_Y$ be its extension by the trivial representation outside $P(Y)$. Let $\mathcal C_Y^L(\alpha_Y)\subset\mathcal C_\Lambda^L(\bar\alpha_Y)$ be the set of local channel fields whose active support, in the presence of $L$, is exactly $Y$. Outside $Y$, all plaquette labels and all channel data are vacuum data.

For $\Gamma\in\mathcal C_Y^L(\alpha_Y)$, we use the projection-supported notation of Subsection~\ref{subsec:trace-topological-coefficients}: set
\[
\begin{aligned}
\mathcal H_Y^\Gamma(\bar\alpha_Y,L)
&:=\bigotimes_{e\in V_E(Y)}\mathcal H_e(\Gamma_e;\bar\alpha_Y,L),\\
\Pi^\Gamma_{Y,N}(\bar\alpha_Y,L)
&:=\bigotimes_{e\in V_E(Y)}\Pi^\Gamma_{e,N}(\bar\alpha_Y,L).
\end{aligned}
\]
and
\[
C_{Y,N}^\Gamma(\bar\alpha_Y,L):=\bigotimes_{e\in V_E(Y)}C_{e,N}(\Gamma_e;\bar\alpha_Y,L).
\]
The determinant-reduced pinned block observable is the connected summand
\begin{equation}\label{eq:scalarization-of-block}
\begin{aligned}
\mathcal O_{Y,N}^{\mathrm{red},L}(\alpha_Y,\Gamma)
:={}&N^{-\#L}
\left(\prod_{p\in P(Y)}d_{\alpha_p}A_p^N(\Gamma_{\partial p};\alpha_p)\right)\\
&\times
\Tr_{\mathcal H_Y^\Gamma(\bar\alpha_Y,L)}\left(
C_{Y,N}^\Gamma(\bar\alpha_Y,L)
\Pi^\Gamma_{Y,N}(\bar\alpha_Y,L)
\right).
\end{aligned}
\end{equation}

The scalar local-channel fields are auxiliary coordinates for a canonical finite-dimensional tensor contraction.  In the estimates below we shall not count these coordinates individually.  Instead, for a fixed reduced label field and a fixed connected support, we first resum all local scalarizations with that support and only then take absolute values:
\begin{equation}\label{eq:resummed-reduced-block}
\mathcal O_{Y,N}^{\mathrm{red},L}(\alpha_Y)
:=
\sum_{\Gamma\in\mathcal C_Y^L(\alpha_Y)}
\mathcal O_{Y,N}^{\mathrm{red},L}(\alpha_Y,\Gamma).
\end{equation}
This resummation is important.  The spaces of walled-Brauer or Schur--Weyl coordinates in mixed degree $m$ may have super-exponential raw cardinality, whereas the tensor contraction which they resolve has a uniform operator description.  All absolute estimates below are therefore performed either before scalarization or after the exact resummation~\eqref{eq:resummed-reduced-block}; multiplicity sums in the residual core are controlled by reciprocal dimensions along canonical Pieri paths rather than by counting basis vectors.

For any finite plaquette set $S$ and any label field $\alpha_S:S\to\widehat{\U(N)}$, we write $M(\alpha_S):=\sum_{p\in S}m(\alpha_p),$ where $m(\alpha_p)=|\alpha_p^+|+|\alpha_p^-|$ is the projective mass of Definition~\ref{def:stable-coordinates}, with stable coordinates understood through the bijection $\lambda_N(\lambda^+,\lambda^-,q)$.  This convention applies equally to reduced labels and to labels with determinant shift.

We shall also use the following consequence of the central selection rule and of the heat-kernel decay.  It is stated in the notation of the pinned block observables introduced above.

\begin{lemma}\label{lem:central-cutoff}
Fix $0<\varepsilon_d<(4d)^{-1}$.  Let $Y$ be a connected active support and let $M=M(\alpha_Y)$ be the total projective mass of the reduced plaquette decoration on $P(Y)$.  For every fixed loop family $L$ there exists $N_0(L,d)<\infty$ such that, whenever $N\ge N_0(L,d)$ and $M\leq \varepsilon_d N$, the central selection rule \eqref{eq:split-selection} holds on the part of the stable expansion supported on $Y$. Moreover, suppose that the corresponding non-spectral stable contribution on $Y$, after summing over all internal local-channel variables of total projective mass $M$ but before inserting the reduced heat-kernel factor, is bounded by $C_LC^{|Y|+M}$.  Then the complementary stable range $M>\varepsilon_dN$ contributes at most
\begin{equation}\label{eq:stable-large-mass-tail}
C_L C^{|Y|}\exp\{-(cT-\log C)\varepsilon_d N\}
\end{equation}
to the corresponding support, for a constant $c>0$ independent of $N$ and $\Lambda$.
\end{lemma}

\begin{proof}
If $M\leq\varepsilon_dN$ and $\varepsilon_d<(4d)^{-1}$, then $2dM+\ell(L)<N$ for all $N\ge N_0(L,d)$ after choosing $N_0(L,d)$ sufficiently large.  Lemma~\ref{lem:central-selection} then gives the split selection rule.  The finitely many values $N<N_0(L,d)$ play no role in the later asymptotic estimates and are absorbed into their constants, but no exact split-selection assertion is made for them.

For the tail, Equation~\eqref{eq:spectral-decay} gives the reduced heat-kernel decay $e^{-cTM}$ in the stable range.  Multiplying by the assumed non-spectral bound gives a summand bounded by $C_LC^{|Y|}(Ce^{-cT})^M$.  Summing over $M>\varepsilon_dN$ gives \eqref{eq:stable-large-mass-tail} once $T$ is large enough, after changing the constants.
\end{proof}

For fixed projective mass $M$, define the signed block coefficient
\begin{equation}\label{eq:signed-block-observable}
\mathcal O_{Y,N}^{\mathrm{red},L}[M]
:=
\sum_{\substack{\alpha_Y:P(Y)\to\widehat{\U(N)}_{\mathrm{red},\mathrm{st}}\\ M(\alpha_Y)=M}}
\mathcal O_{Y,N}^{\mathrm{red},L}(\alpha_Y).
\end{equation}
We estimate the absolute label mass
\begin{equation}\label{eq:absolute-block-observable}
\left\|\mathcal O_{Y,N}^{\mathrm{red},L}\right\|_{M,1}
:=
\sum_{\substack{\alpha_Y:P(Y)\to\widehat{\U(N)}_{\mathrm{red},\mathrm{st}}\\ M(\alpha_Y)=M}}
\left|\mathcal O_{Y,N}^{\mathrm{red},L}(\alpha_Y)\right|.
\end{equation}
Thus absolute values are taken after the finite local-channel resummation at fixed label field.  Clearly
$|\mathcal O_{Y,N}^{\mathrm{red},L}[M]|
\le \|\mathcal O_{Y,N}^{\mathrm{red},L}\|_{M,1}$.
Controlled asymptotic expansions in powers of $1/N$ will be asserted for the signed block coefficient.

We write $b(Y,L)$ for the number of distinct marked trace components of $L$ which meet $Y$.  In particular $Y$ is pinned if and only if $b(Y,L)\ge1$.  All constants denoted by $C_d$ depend only on $d$; constants denoted by $C_L$ may depend on the fixed loop family $L$, but never on $N$, $\Lambda$, $Y$ or $M$.

\section{Pinned block observables estimates}
\label{sec:pinned-haar-projection}

The purpose of this section is to prove the local large-$N$ estimate for the pinned block observables $\mathcal O_{Y,N}^{\mathrm{red},L}$ and their absolute mass. This estimate is the input which makes the heat-kernel master-loop transfer contractive in Section~\ref{sec:rooted-master-loop}.

\subsection{Pinned block bounds}

The trace form of $\mathcal O_{Y,N}^{\mathrm{red},L}(\alpha_Y,\Gamma)$ will be used throughout the proof. Since $\Pi^\Gamma_{Y,N}=\bigotimes_{e\in V_E(Y)}\Pi^\Gamma_{e,N}$ is a tensor product of orthogonal Haar projections, every active edge contribution is a matrix coefficient of a projection. More explicitly, after expanding the finite-rank operator
\[
N^{-\#L}
\left(\prod_{p\in P(Y)}d_{\alpha_p}A_p^N(\Gamma_{\partial p};\alpha_p)\right)
C_{Y,N}^\Gamma(\bar\alpha_Y,L)
\]
into rank-one terms, there are a finite set $\mathcal J(Y,\alpha_Y,\Gamma;L)$ and vectors
\[
\xi^{\rm in}_{j,N},\xi^{\rm out}_{j,N}
\in
\mathcal H_Y^\Gamma(\bar\alpha_Y,L),
\qquad j\in\mathcal J(Y,\alpha_Y,\Gamma;L),
\]
such that
\begin{equation}\label{eq:operator-block-contraction}
\mathcal O_{Y,N}^{\mathrm{red},L}(\alpha_Y,\Gamma)=\sum_{j\in\mathcal J(Y,\alpha_Y,\Gamma;L)}\left\langle\xi^{\rm out}_{j,N},\Pi^\Gamma_{Y,N}(\bar\alpha_Y,L)\xi^{\rm in}_{j,N}\right\rangle .
\end{equation}
All plaquette coefficients, dimension factors, Wilson normalization and fixed channel-coordinate maps are contained in the vectors. The only feature used below is that the edge integrations remain the orthogonal projections $\Pi^\Gamma_{e,N}$, and that the local maps which connect plaquette data to these projections have the fixed-degree bounds recorded in Lemma~\ref{lem:local-operator-complexity}.

A local Schur--Weyl resolution of a plaquette trace is summarized, as in the companion paper, by a finite set $\mathcal R_p(\lambda)$ of coordinate choices.  For $r\in\mathcal R_p(\lambda)$, we write $\mathcal I_p(r)$ for the corresponding finite set of incidence channel types.  The following lemma records the local complexity and stable-expansion bounds needed below.

\begin{lemma}[Local operator complexity]\label{lem:local-operator-complexity}
There is a constant $C_d<\infty$ with the following property.  Let $p$ be a plaquette and let $\lambda=[\lambda^+,\lambda^-]_N$ be a reduced stable label of mass $m$.  The projection-supported plaquette tensor and every local incidence map which occurs in the construction can be realized, without expanding the full walled-Brauer or multiplicity spaces into a raw basis, as a composition of at most $C_d(1+m)$ elementary operations: orthogonal Schur--Weyl projections, permutations of tensor factors, normalized duality maps, and bounded-valence incidence maps.  Consequently every such local operator $T_{\mathrm{loc}}$ satisfies
\[
\|T_{\mathrm{loc}}\|_{2\to2}\le C_d^{1+m}.
\]
For fixed $m$, its matrix coefficients in any of the orthonormal stable channel coordinates used below admit controlled asymptotic expansions to all orders in powers of $1/N$:
\[
\langle u,T_{\mathrm{loc}}(N)v\rangle
= \sum_{r=0}^{R} a_r(u,v)N^{-r}+O_{R,m}(N^{-R-1}),
\]
uniformly over the finitely many local incidence types of the lattice.

No exponential bound on the \emph{raw number} of walled-Brauer diagrams or multiplicity-basis vectors is asserted.  Whenever internal multiplicity variables have to be summed in the sequel, we use a fixed sequential Pieri resolution and the reciprocal-dimension estimate of Lemma~\ref{lem:oscillating-path}.
\end{lemma}

\begin{proof}
Realize $V_\lambda$ inside
\[
(\C^N)^{\otimes|\lambda^+|}
\otimes
((\C^N)^*)^{\otimes|\lambda^-|}.
\]
Keep the orthogonal isotypic projector onto $V_\lambda$ as a norm-one operator; in particular, do not expand it into all walled-Brauer diagrams.  Since a plaquette has bounded boundary length, distributing the elementary mixed-tensor legs among its incidences requires $O_d(1+m)$ permutations, contractions, coevaluations and orthogonal projections.  We use normalized evaluation--coevaluation maps whenever an invariant line is opened, so each elementary operator has norm bounded by a universal constant.  Their composition therefore has norm at most $C_d^{1+m}$.

For fixed $m$ all mixed tensor degrees are fixed.  Algebraic diagrammatic coordinate maps have rational dependence on $N$, and Lemma~\ref{lem:fixed-degree-orthonormalization} shows that passage to orthonormal stable channel coordinates preserves controlled asymptotic expansions.  Since only $O_d(1+m)$ fixed-degree compositions are involved, the asserted expansions follow.  The final sentence is a convention for the rest of the proof: multiplicity spaces are never estimated by their dimension.
\end{proof}

The same convention is used at edge Haar projections.  We keep each Haar integral as the orthogonal projection onto the invariant subspace until the residual core is resolved.  At that stage the incident plaquette representations are opened by tensoring their elementary $V_N$- and $V_N^*$-strands in a fixed lexicographic order and with a fixed left-associated bracketing.  The stable Pieri rules are multiplicity-free at each elementary step, so all multiplicity information is encoded by the resulting oscillating paths.  We shall estimate the corresponding sums by Lemma~\ref{lem:oscillating-path}; in particular, no estimate of the dimension of an invariant multiplicity space is used.

\begin{lemma}[Support resummation]\label{lem:support-resummation}
Fix a reduced label field $\alpha_Y$ and a connected support $Y$.  The resummed block
$\mathcal O_{Y,N}^{\mathrm{red},L}(\alpha_Y)$ from~\eqref{eq:resummed-reduced-block}
can be written as a linear combination of at most $C_LC_d^{|Y|}$ projection-supported tensor contractions with the following properties:
\begin{itemize}
\item every edge Haar integral is kept as an orthogonal projection;
\item every plaquette is represented by its unsplit isotypic projector together with the local incidence operators of Lemma~\ref{lem:local-operator-complexity};
\item the coefficients of the linear combination have absolute value at most $1$.
\end{itemize}
In particular, passing from the scalar-channel description to fixed support $Y$ introduces only an exponential-in-$|Y|$ cost and no multiplicity-space cardinality.
\end{lemma}

\begin{proof}
The scalar local-channel representation is obtained by expanding a finite collection of plaquette, incidence and Wilson operators in tensor-product coordinates with one distinguished vacuum coordinate at every local position.  Summing all scalar coefficients restores the original finite-dimensional operator exactly.

To impose that the active support is precisely $Y$, at each local coordinate belonging to $Y$ we project onto the complement of its vacuum coordinate, while outside $Y$ we project onto the vacuum coordinate.  Writing each non-vacuum projector as $\mathrm{Id}-P_{\mathrm{vac}}$ gives an inclusion--exclusion expansion.  The number of local positions adjacent to a fixed connected support is at most $C_d|Y|+C_L$ by bounded lattice valence and by the fixed Wilson collar.  Hence the expansion has at most $C_LC_d^{|Y|}$ terms.  All vacuum projectors, complementary projectors and edge Haar projections have norm at most one, and the coefficients in the inclusion--exclusion are $\pm1$.  This proves the claim.
\end{proof}

The main result is the following, which bounds the absolute $\ell^1$-mass of the pinned block observable without producing any positive power of $N$.

\begin{theorem}\label{thm:pinned-block}
There exists $C_d<\infty$, depending only on $d$, such that the following holds.  Let $\Lambda\subset\Z^d$ be finite, let $L$ be fixed, and let $Y$ be a connected pinned support.  For every $N$ and every $M\le N/2$,
\begin{equation}
\label{eq:pinned-block-bound}
\left\|\mathcal O_{Y,N}^{\mathrm{red},L}\right\|_{M,1}
\le
C_L C_d^{|Y|+M}N^{-\delta(Y,L)},
\end{equation}
where $\delta(Y,L)\ge0$, and where $\delta(Y,L)\ge1$ if $b(Y,L)\ge2$.  Moreover, for fixed $Y$ and $M$, the signed block coefficient $\mathcal O_{Y,N}^{\mathrm{red},L}[M]$ admits a controlled asymptotic expansion to all orders in powers of $1/N$, uniformly in the ambient finite volume.
\end{theorem}

The proof is performed after the exact local-channel resummation~\eqref{eq:resummed-reduced-block}.  This is essential for an exponential-in-mass bound: raw multiplicity spaces are not counted.  The only sums which are opened absolutely are canonical stable Pieri path sums, whose reciprocal-dimension weights are controlled by Lemma~\ref{lem:oscillating-path}.

The proof of Theorem~\ref{thm:pinned-block} occupies most of this section. Before turning to it, we first explain how it implies the block estimate used later. If a connected support $Y$ carries a plaquette decoration $\lambda_Y:P(Y)\to\widehat{\U(N)}$, we write $\mathcal E(\lambda_Y):=\sum_{p\in P(Y)}\mathcal E(\lambda_p),$ where $\mathcal{E}(\lambda)=m(\lambda)+q^2$.

\begin{lemma}
\label{lem:minimal-support-size}
There exist $a_d>0$ and $C_L<\infty$ such that, for every connected pinned support $Y$ and every plaquette decoration $\lambda_Y:P(Y)\to\widehat{\U(N)}$,
\[
\mathcal{E}(\lambda_Y)\ge a_d |Y|-C_L .
\]
\end{lemma}

\begin{proof}
Let $P_{\rm nt}(Y)=\{p\in P(Y):\mathcal{E}(\lambda_p)>0\}$ be the set of non-trivial plaquette vertices.  If a plaquette vertex $p\notin P_{\rm nt}(Y)$ is active, then its plaquette label is trivial.  By the definition of the vacuum local-channel datum, such a plaquette can be active only through non-vacuum Wilson-collar data.  Hence it lies in a bounded neighbourhood, in the dual incidence graph, of the fixed loop family $L$, and the number of such plaquette vertices is bounded by $C_L$.

Similarly, every active edge vertex is incident either to a plaquette in $P_{\rm nt}(Y)$, or to a marked Wilson strand.  Since every lattice edge is incident to at most $2(d-1)$ plaquettes, we get
\[
|Y|\le C_d |P_{\rm nt}(Y)|+C_L .
\]
For every $p\in P_{\rm nt}(Y)$, one has $\mathcal E(\lambda_p)\ge1$.  Therefore
\[
\mathcal E(\lambda_Y)
\ge |P_{\rm nt}(Y)|
\ge a_d|Y|-C_L,
\]
after changing the constants.
\end{proof}

Let $\lambda_Y:P(Y)\to\widehat{\U(N)}$ be a plaquette decoration with stable coordinates.  Write
\[
\lambda_p=q_p\mathbf 1_N+\alpha_p,\qquad\alpha_p=[\lambda_p^+,\lambda_p^-]_N,\qquad\alpha_Y=(\alpha_p)_{p\in P(Y)}.
\]
The stable pinned block observable is
\begin{equation}\label{eq:hk-block-observable}
\mathcal O_{Y,N}^{\mathrm{st},L}(\lambda_Y)
:=
\mathbf 1_{\{\partial^*c(\lambda_Y)+j_L=0\}}
\mathcal O_{Y,N}^{\mathrm{red},L}(\alpha_Y),
\end{equation}
where $c(\lambda_Y)=(c(\lambda_p))_{p\in P(Y)}$.  We also set
\[
\mathcal O_{Y,N}^{\mathrm{abs},\mathrm{st},L}(\lambda_Y)
:=
\mathbf 1_{\{\partial^*c(\lambda_Y)+j_L=0\}}
\left|\mathcal O_{Y,N}^{\mathrm{red},L}(\alpha_Y)\right|.
\]
The indicator is the central charge selection rule of Lemma~\ref{lem:central-selection}.  On the small-mass range where this rule splits, the determinant characters cancel from the topological coefficient and the determinant shifts are sampled only through the reference law $\widehat\Pbb_{T,N}$.

The following tilted-tail estimate is the device which lets us pass from stable block estimates to full $\U(N)$ label expectations after paying an exponential non-spectral block cost.  The parameter $A$ stands for any exponential-in-mass factor left by the projection-supported local estimates.

\begin{proposition}\label{prop:nonstable-polymer-tail}
Fix $A\ge1$ and assume $T>4\log A$.  For every connected pinned support $Y$ and every plaquette label field $\alpha=(\alpha_p)_{p\in P(Y)}$, let $B_N^L(Y;\alpha)$ be a non-negative quantity satisfying
\begin{equation}
\label{eq:nonspectral-growth-assumption}
B_N^L(Y;\alpha)
\le
C_LC^{|Y|}A^{M(\alpha)},
\qquad
M(\alpha)=\sum_{p\in P(Y)}m(\alpha_p),
\end{equation}
with constants independent of $N$ and of the finite volume.  Then, for some $C_T<\infty$ and $\gamma(A,T)>0$,
\begin{equation}\label{eq:nonstable-polymer-tail-under-P}
\widehat\E_{Y,T,N}\left[
B_N^L(Y;\alpha)
\mathbf 1_{\{M(\alpha)>N/2\}}
\right]
\le
C_L C_T^{|Y|}e^{-\gamma(A,T)N},
\end{equation}
whenever $T>4\log A$, uniformly in $\Lambda$ and $N$.
\end{proposition}

\begin{proof}
By~\eqref{eq:nonspectral-growth-assumption}, we have $B_N^L(Y;\alpha)A^{-M(\alpha)}\le C_LC^{|Y|}.$ Therefore
\[
\begin{split}
\widehat\E_{Y,A,T,N}\left[B_N^L(Y;\alpha)A^{-M(\alpha)}\mathbf 1_{\{M(\alpha)>N/2\}}\right] &\le C_LC^{|Y|} \widehat\Pbb_{Y,A,T,N}\left(M(\alpha)>N/2\right)  \\
&\le C_LC^{|Y|}C(A,T)^{|P(Y)|}e^{-\gamma(A,T)N}.
\end{split}
\]
The last inequality is Lemma~\ref{lem:product-tilted-mass-tail}, applied with $S=P(Y)$; the factor $C(A,T)^{|P(Y)|}$ is absorbed into $C_T^{|Y|}$. We get
\[
\widehat\E_{Y,A,T,N}\left[B_N^L(Y;\alpha)A^{-M(\alpha)}\mathbf 1_{\{M(\alpha)>N/2\}}\right]\le C_L C_T^{|Y|}e^{-\gamma(A,T)N}
\]
for some $C_T<\infty$ and $\gamma(A,T)>0$, uniformly in $\Lambda$ and $N$. Finally,
\begin{align*}
&\widehat\E_{Y,T,N}\left[
B_N^L(Y;\alpha)\mathbf 1_{\{M(\alpha)>N/2\}}
\right]\\
&\qquad=
\left(\frac{\mathcal Z_{A,T,N}}{\mathcal Z_{T,N}}\right)^{|P(Y)|}
\widehat\E_{Y,A,T,N}\left[
B_N^L(Y;\alpha)A^{-M(\alpha)}
\mathbf 1_{\{M(\alpha)>N/2\}}
\right],
\end{align*}
and the ratio of normalizing constants is bounded by a $T,A$-dependent constant.  This gives~\eqref{eq:nonstable-polymer-tail-under-P}.
\end{proof}

\begin{theorem}
\label{thm:stable-hk-block}
There exist constants $C,c>0$, depending only on $d$, with the following property.  For every fixed loop family $L$, there is $C_L<\infty$ such that, for every connected pinned support $Y$,
\begin{equation}
\label{eq:stable-hk-block}
\widehat\E_{Y,T,N}\left[
\mathcal O_{Y,N}^{\mathrm{abs},\mathrm{st},L}(\lambda_Y)
\right]
\le
C_L\exp\{-\mu(T)|Y|\},
\qquad
\mu(T)=cT-\log C .
\end{equation}

\end{theorem}

\begin{proof}
We estimate the expectation in \eqref{eq:stable-hk-block} by using the determinant-reduced block observable and the product reference law $\widehat\Pbb_{Y,T,N}$. We first restrict to the stable small-mass range $M=\sum_{p\in P(Y)}m(\alpha_p)\le \varepsilon_d N,$ where $\varepsilon_d$ is as in Lemma~\ref{lem:central-cutoff}.  On this range the central selection rule splits as
\[
\partial^*q=0,
\qquad
\partial^*s+j_L=0.
\]
Thus the determinant characters cancel from the local-channel coefficient; the determinant variables remain only in the reference law $\widehat\Pbb_{T,N}$.

We now estimate the contribution at fixed projective mass $M$, after the local-channel variables have been resummed and the reduced labels are summed.  Using the heat-kernel domination bound~\eqref{eq:spectral-decay} plaquettewise, and absorbing the factor $\exp\{\frac{T}{2}\sum_p m(\alpha_p)^2/N^2\}$ into the constants on the range $M\le\varepsilon_dN$, there is a constant $c>0$ such that, after changing constants,
\[
\prod_{p\in P(Y)}
\exp\left\{-\frac{T}{2N}C_2(q_p\mathbf 1_N+\alpha_p)\right\}
\le C_L\exp\{-cT(M+Q)\}.
\]
Here and below,
\[
Q:=\sum_{p\in P(Y)}q_p^2.
\]
The possible factor coming from the cross term $2q_ps_p/N$ is absorbed by decreasing $c$, using $M\le\varepsilon_dN$ and the elementary inequality
\[
\frac{2|q_p|m_p}{N}
\le \frac{q_p^2}{2}+\frac{2m_p^2}{N^2}.
\]

For fixed $M$ and $q_Y$, Theorem~\ref{thm:pinned-block} already includes the sum over all reduced label fields of total mass $M$, with all local channel variables resummed at each fixed label field:
\[
\left\|\mathcal O_{Y,N}^{\mathrm{red},L}\right\|_{M,1}
\le C_LC_d^{|Y|+M}N^{-\delta(Y,L)} .
\]
Discarding the harmless factor $N^{-\delta(Y,L)}\le1$, the absolute contribution of all decorations with reduced mass $M$ and fixed determinant shifts $q_Y$ is therefore bounded by
\[
C_L C^{|Y|+M}\exp\{-cT(M+Q)\} .
\]
For every plaquette decoration, Lemma~\ref{lem:minimal-support-size} gives
\[
\mathcal E(\lambda_Y)=M+Q\ge a_d|Y|-C_L.
\]
Hence, after changing constants,
\[
C_L C^{|Y|+M}e^{-cT(M+Q)}
\le
C_L e^{-cTa_d|Y|/2}
C^{|Y|+M}e^{-cT(M+Q)/2}.
\]
The remaining sum over $M$ and over determinant shifts $q_Y\in\Z^{P(Y)}$ is bounded by $C^{|Y|}$ for $T$ large enough: the reduced labels have already been summed in $\left\|\mathcal O_{Y,N}^{\mathrm{red},L}\right\|_{M,1}$, the residual factor $C^M e^{-cTM/2}$ is summable, and the determinant shifts are controlled by the product Gaussian sum $\sum_q e^{-cTq^2/2}<\infty$.  Thus
\[
\widehat\E_{Y,T,N}\left[
\mathcal O_{Y,N}^{\mathrm{abs},\mathrm{st},L}(\lambda_Y)
\mathbf 1_{\{M(\alpha_Y)\le\varepsilon_dN\}}
\right]
\le
C_L\exp\{-(c'T-\log C')|Y|\}.
\]

It remains to treat the complement of the stable small-mass range.  We split it into
\[
\varepsilon_dN<M\le N/2
\qquad\text{and}\qquad
M>N/2.
\]
In the first range, Theorem~\ref{thm:pinned-block} gives the same non-spectral majorant $C_LC^{|Y|+M}$, and Lemma~\ref{lem:central-cutoff} gives an exponentially small cutoff in $N$ after insertion of the reduced heat-kernel factor.  The determinant variables are again controlled by Gaussian tails, and Lemma~\ref{lem:minimal-support-size} extracts the same support decay from $\mathcal E(\lambda_Y)$.  In the second range, Corollary~\ref{cor:coarse-nonspectral-block} gives the pointwise non-spectral bound
$C_LC^{|Y|}A^{M}$ for a fixed constant $A=A(d)$; the tilted tail estimate of Proposition~\ref{prop:nonstable-polymer-tail} therefore applies directly.  For $T$ large enough both complementary contributions are bounded by
\[
C_L\exp\{-(c''T-\log C'')|Y|\}
\]
after changing the constants, and are in fact exponentially small in $N$.  Combining this with the small-mass estimate proves~\eqref{eq:stable-hk-block}.
\end{proof}

The same proof gives the following coefficientwise refinement: for fixed $Y$, $\widehat\E_{Y,T,N}\left[\mathcal O_{Y,N}^{\mathrm{st},L}(\lambda_Y)\right]$ admits a controlled asymptotic expansion to all orders in powers of $1/N$: for every $R\ge0$ there are coefficients $\mathcal O_Y^{[r]}(L,T)$ such that
\begin{equation}\label{eq:stable-hk-block-expansion}
\widehat\E_{Y,T,N}\left[\mathcal O_{Y,N}^{\mathrm{st},L}(\lambda_Y)\right]=\sum_{r=0}^{R}N^{-r}\mathcal O_Y^{[r]}(L,T)+O_{R,Y,L,T}(N^{-R-1}).
\end{equation}
Indeed, let us fix an order $R\ge0$.  On every fixed stable projective-mass sector and fixed determinant-energy sector, the signed summand is a finite sum of stable Schur--Weyl matrix coefficients multiplied by the reference-law density.  By the controlled-expansion convention and the stable Schur--Weyl discussion in Subsection~\ref{subsec:channel-recoupling}, after subtracting the terms up to order $R$ the sectorwise remainder is bounded by $N^{-R-1}$ times a constant depending on the sector.  The estimates above give a summable majorant for these constants: the projective mass is controlled by the factor $C^M e^{-cTM/2}$, the determinant shifts by the product Gaussian factor, and the support dependence is already absorbed into $C_L e^{-(cT-\log C)|Y|}$.  Hence the coefficient sums obtained by summing the sectorwise coefficients are absolutely convergent, and the total remainder is $O_{R,Y,L,T}(N^{-R-1})$.  The complementary ranges $\varepsilon_dN<M\le N/2$ and $M>N/2$ are exponentially small in $N$, hence are also $O_{R,Y,L,T}(N^{-R-1})$.

In the rest of this section, we prove Theorem~\ref{thm:pinned-block}.

\subsection{The peeling process}\label{subsec:peeling}

The first step in the proof of Theorem~\ref{thm:pinned-block} will rely on a peeling process.  Fix a connected active support $Y$ and a connected active subgraph $K\subset Y$ which contains all marked edge vertices of $Y$.
The connected components of $Y\setminus K$ are the tree appendices which have to be removed, while the operator network supported on $K$ must be left untouched.  The point is to remove those appendices one plaquette at a time, and to cancel the corresponding plaquette dimension factor before any tensor-coordinate absolute value is taken.

The active incidence graph is bipartite, with plaquette vertices and edge vertices; a terminal leaf of a tree appendix can therefore be of two types.  The following terminology will be used throughout the subsection.

\begin{definition}
Let $G$ be an active subgraph of the dual incidence graph. An active edge vertex $e$ of $G$ is an \emph{edge leaf} if it is incident to exactly one active plaquette incidence in $G$.  It is an \emph{unmarked edge leaf} if, in addition, it carries no marked Wilson strand. A plaquette vertex $p$ of $G$ is a \emph{plaquette leaf} if it is incident to exactly one active edge vertex in $G$.  It is a \emph{vacuum plaquette leaf} if its unique adjacent active edge vertex carries no marked Wilson strand.
\end{definition}

If $K\subset Y$ is a connected active subgraph containing all marked edge vertices, a plaquette-leaf peeling order for $Y\setminus K$ directed towards $K$ is an ordering $p_1,\ldots,p_s$ of the plaquette vertices of $Y\setminus K$ such that, after removing $p_1,\ldots,p_{j-1}$ and all unmarked edge vertices which have become isolated, the plaquette $p_j$ is a vacuum plaquette leaf of the remaining active graph.  A pinned support $Y$ is peelable if it admits such an order with $K$ equal to the active subgraph left after all active plaquette vertices have been removed.

\begin{definition}\label{def:anchored-core}
Let $Y$ be a connected pinned support with active incidence graph $\mathcal I_Y$, and let $\mathsf M(Y,L)$ be the set of active edge vertices of $\mathcal I_Y$ carrying marked Wilson occurrences.  The \emph{residual core} $\operatorname{Core}_*(Y)$ is obtained by the following pruning procedure.  Starting from $\mathcal I_Y$, remove successively every vacuum plaquette leaf, together with its unique active incidence, and then remove every unmarked edge vertex which has become isolated.  The procedure terminates because $\mathcal I_Y$ is finite.
\end{definition}

\begin{figure}[h!]
\centering
\includegraphics[width=0.4\linewidth]{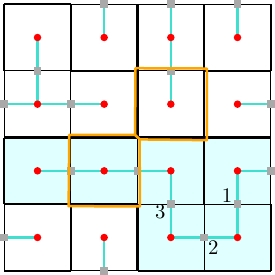}
\caption{A dual incidence graph illustrating the peeling process.  The connected active support $Y_{\mathrm{low}}$ (in light blue) will be peeled in the displayed order: $p_{34}\to p_{44}\to p_{43}$. The remaining residual core is
$\operatorname{Core}_*(Y_{\mathrm{low}})=\{p_{31},p_{32},p_{33}\}$.}
\label{fig:peeling}
\end{figure}

\begin{example}\label{ex:peeling}
Consider the configuration of Figure~\ref{fig:peeling}, and assume that the plaquette labels are indexed by rows and columns, starting from $p_{11}$ and ending at $p_{44}$. The orange loop $\ell$ produces two connected active supports: $Y_{\mathrm{high}}=\{p_{13},p_{23}\}$ and $Y_{\mathrm{low}}=\{p_{31},p_{32},p_{33},p_{43},p_{44},p_{34}\}.$ We focus on $Y_{\mathrm{low}}$.

Let $\mathfrak e_L=(p_{31},p_{32}),$ $\mathfrak e_M=(p_{32},p_{33})$ be the two grey edge vertices crossed by the lower orange Wilson insertion.  Thus the marked edge vertices of $Y_{\mathrm{low}}$ are precisely $\mathfrak e_L$ and $\mathfrak e_M$.  The remaining edge vertices in this support are vacuum vertices; we denote them by
\[
\mathfrak f_1=(p_{33},p_{43}),
\qquad
\mathfrak f_2=(p_{43},p_{44}),
\qquad
\mathfrak f_3=(p_{44},p_{34}).
\]
With this notation the active incidence graph of $Y_{\mathrm{low}}$ is the chain
\[
p_{31}-\mathfrak e_L-p_{32}-\mathfrak e_M-p_{33}-\mathfrak f_1-p_{43}-\mathfrak f_2-p_{44}-\mathfrak f_3-p_{34}.
\]
The vacuum appendix is peeled from the right.  First $p_{34}$ is a vacuum plaquette leaf, attached through the unmarked edge vertex $\mathfrak f_3$.  After removing $p_{34}$ and the now isolated edge vertex $\mathfrak f_3$, the plaquette $p_{44}$ becomes a vacuum plaquette leaf attached through $\mathfrak f_2$.  Removing $p_{44}$ and $\mathfrak f_2$ makes $p_{43}$ a vacuum plaquette leaf attached through $\mathfrak f_1$.  Hence the peeling order is
\[
p_{34}\quad\longrightarrow\quad p_{44}
\quad\longrightarrow\quad p_{43}.
\]
The pruning then stops at $\operatorname{Core}_*(Y_{\mathrm{low}})=\{p_{31},p_{32},p_{33}\},$ together with the marked edge vertices $\mathfrak e_L$ and $\mathfrak e_M$. Indeed $p_{31}$ and $p_{33}$ may be leaves in the remaining graph, but they are not vacuum leaves: they are attached to the Wilson insertion through marked edge vertices.
\end{example}

For a fixed scalar local channel $\Gamma$, the block supported on $Y$ is the scalar product~\eqref{eq:scalarization-of-block}.  In the proof we use only the projection-supported realization~\eqref{eq:operator-block-contraction}: all plaquette coefficients, dimension factors and Wilson-collar normalizations are absorbed into fixed finite-dimensional channel-coordinate maps, while each active edge integration remains the orthogonal projection $\Pi_{e,N}$.  This is precisely the structure which cancels plaquette dimensions before any coordinatewise absolute value is taken.

Let $K\subset Y$ be a connected active subgraph containing all marked edge vertices and let $\mathfrak p=(p_1,\ldots,p_s)$ be a plaquette-leaf peeling order for $Y\setminus K$ directed towards $K$.  Put $G_0:=Y$. For $1\le j\le s$, define
\[
I_j:=\left\{e\in V_E(G_{j-1}):
 e\text{ is unmarked and isolated in }G_{j-1}\setminus\{p_j\}\right\},
\]
and
\begin{equation}\label{eq:peeling-graphs}
G_j:=G_{j-1}\setminus\bigl(\{p_j\}\cup I_j\bigr).
\end{equation}
Thus $G_s=K$.  Let $e_j$ be the unique active edge vertex adjacent to $p_j$ in $G_{j-1}$.  Although the incidence $(p_j,e_j)$ is made of elementary fundamental/dual wires, the leaf plaquette tensor contains the Schur--Weyl projector selecting the irreducible label $\alpha_{p_j}$.  Fusing this whole leaf bundle gives the irreducible channel
\[
\lambda_j=
\begin{cases}
\alpha_{p_j}, & \text{if the boundary orientation gives the channel }V_{\alpha_{p_j}},\\
\alpha_{p_j}^*, & \text{if the boundary orientation gives the dual channel }V_{\alpha_{p_j}}^* .
\end{cases}
\]
Let $H_{j,N}$ be the tensor product of the fundamental and dual tensor factors incident to $e_j$ which remain in $G_{j-1}\setminus\{p_j\}$.  Set
\begin{equation}\label{eq:peeling-projection-definition}
\Pi_{\lambda,H}^{(N)}
:=
\int_{\U(N)}\rho_\lambda(g)\otimes\pi_H(g) dg,
\qquad
\mathsf P_{\lambda,H}^{(N)}
:=
d_\lambda\Tr_{V_\lambda}\bigl(\Pi_{\lambda,H}^{(N)}\bigr).
\end{equation}
The part of the local-channel tensor network involving the leaf plaquette $p_j$, the dimension factor $d_{\alpha_{p_j}}$, and the Haar projection at $e_j$, viewed as an operator on the remaining edge space $H_{j,N}$, has the factorization
\begin{equation}\label{eq:leaf-factorization}
\mathsf L_{j,N}
=
B_{j,N} \mathsf P_{\lambda_j,H_{j,N}}^{(N)} A_{j,N},
\end{equation}
where $A_{j,N}$ and $B_{j,N}$ are the local incidence and fixed-degree channel-coordinate maps attached to $p_j$ and $(p_j,e_j)$.  Lemma~\ref{lem:local-operator-complexity} gives
\begin{equation}\label{eq:leaf-AB-bound}
\|A_{j,N}\|_{2\to2} \|B_{j,N}\|_{2\to2}
\le C_d^{1+m(\alpha_{p_j})}.
\end{equation}

The residual operator network associated with $(K,\mathfrak p)$ is obtained by replacing every compressed leaf operator $\mathsf L_{j,N}$ by $\mathsf P_{\lambda_j,H_{j,N}}^{(N)}$.  The remaining plaquette tensors, edge projections, Wilson-collar tensor and fixed half-strand pairings define a residual edge space
\[
\mathcal H_{K,\mathfrak p,N}
=
\left(\bigotimes_{e\in V_E(K)}
\mathcal H_e(\Gamma_e;\bar\alpha_Y,L)\right)
\otimes
\left(\bigotimes_{j=1}^s H_{j,N}\right),
\]
assembly maps
\[
I^{\rm in}_{K,\mathfrak p,N},
I^{\rm out}_{K,\mathfrak p,N}:
\mathcal M_{K,\mathfrak p}(\Gamma)
\longrightarrow \mathcal H_{K,\mathfrak p,N},
\]
where $\mathcal M_{K,\mathfrak p}(\Gamma)$ is the finite Hilbert space spanned by the residual scalar-channel multiplicity coordinates left after the chosen peeling order.  Its orthonormal basis is the product of the local channel-coordinate bases, the Haar-invariant bases at the residual edge vertices, and the finite half-strand pairing coordinates created by the peeled appendices.  The maps $I^{\rm in}_{K,\mathfrak p,N}$ and $I^{\rm out}_{K,\mathfrak p,N}$ are the corresponding assembly maps from these multiplicity coordinates into the residual edge tensor space.  The residual projection-type operator is
\[
\mathcal P^{\rm res}_{K,\mathfrak p,N}
:=
\left(\bigotimes_{e\in V_E(K)}
\Pi^\Gamma_{e,N}(\bar\alpha_Y,L)\right)
\otimes
\left(\bigotimes_{j=1}^s
\mathsf P_{\lambda_j,H_{j,N}}^{(N)}\right).
\]
The residual scalar is
\begin{equation}\label{eq:residual-scalar-definition}
R_N^L(K;\alpha_Y,\Gamma;\mathfrak p)
:=
\Tr_{\mathcal M_{K,\mathfrak p}(\Gamma)}\!\left(
(I^{\rm out}_{K,\mathfrak p,N})^*
\mathcal P^{\rm res}_{K,\mathfrak p,N}
I^{\rm in}_{K,\mathfrak p,N}
\right).
\end{equation}
We shall also use a core-resolved absolute residual mass.  Set
\[
\mathsf P^{\rm peel}_{\mathfrak p,N}:=\bigotimes_{j=1}^s\mathsf P^{(N)}_{\lambda_j,H_{j,N}} .
\]
By Lemma~\ref{lem:partial-trace-haar-projection}, $\|\mathsf P^{\rm peel}_{\mathfrak p,N}\|_{2\to2}\le 1.$ The absolute mass is defined by resolving only the residual core.  More precisely, fix the deterministic ordering of plaquettes, elementary $V_N/V_N^*$ strands and incidences used below in Lemma~\ref{lem:core-path-encoding}, use the corresponding left-associated orthonormal Clebsch--Gordan resolution at every Haar-projection vertex of $K$, and open every intermediate channel in that canonical resolution.  The remaining fixed local coordinate maps supported on $K$ are kept as bounded operators.  The peeled operator $\mathsf P^{\rm peel}_{\mathfrak p,N}$ is not resolved; it is kept as a norm-one black box.  With this convention the signed residual scalar admits a finite core-resolved expansion
\[
R_N^L(K;\alpha_Y,\Gamma;\mathfrak p)
=
\sum_{\omega\in\Omega_{\rm core}(K,\alpha_Y,\Gamma;\mathfrak p)}
B_N(\omega;\mathsf P^{\rm peel}_{\mathfrak p,N}),
\]
where each core-resolved summand has the form
\[
B_N(\omega;\mathsf P^{\rm peel}_{\mathfrak p,N})
=
c_N(\omega)
\left(\prod_{c\in\mathcal C(\omega)}d_{\eta_c}^{-1}\right)
\left\langle
\xi_N(\omega),
\bigl(U_N(\omega)\otimes\mathsf P^{\rm peel}_{\mathfrak p,N}\bigr)
\zeta_N(\omega)
\right\rangle .
\]
Here $\omega$ records the internal labels visited by the canonical sequential resolution and the corresponding multiplicity history; there is no independent sum over bracketings or fusion schedules.  The vectors $\xi_N(\omega),\zeta_N(\omega)$ are unit vectors after absorbing fixed scalar coordinate normalizations into $c_N(\omega)$, and $U_N(\omega)$ is the remaining bounded core operator.  We define
\begin{equation}\label{eq:absolute-residual-mass}
\mathcal R_N^L(K;\alpha_Y,\Gamma;\mathfrak p)
:=
\sum_{\omega\in\Omega_{\rm core}(K,\alpha_Y,\Gamma;\mathfrak p)}
|c_N(\omega)|
\left(\prod_{c\in\mathcal C(\omega)}d_{\eta_c}^{-1}\right)
\|U_N(\omega)\|_{2\to2} .
\end{equation}
Since $\|\mathsf P^{\rm peel}_{\mathfrak p,N}\|_{2\to2}\le1$, this gives
\[
|R_N^L(K;\alpha_Y,\Gamma;\mathfrak p)|
\le
\mathcal R_N^L(K;\alpha_Y,\Gamma;\mathfrak p).
\]
The definition is independent of cancellations between core-resolved summands.  It is this core-resolved absolute residual mass, rather than the signed residual scalar, which is used in absolute-value estimates.  In particular, the peeled projections never contribute recoupling multiplicities to $\mathcal R_N^L$; their only role is through their norm-one bound.

The original block observable is obtained from the same residual assembly maps before the replacements, by using
\[
\mathcal L_{K,\mathfrak p,N}
:=
\left(\bigotimes_{e\in V_E(K)}
\Pi^\Gamma_{e,N}(\bar\alpha_Y,L)\right)
\otimes
\left(\bigotimes_{j=1}^s
\mathsf L_{j,N}\right)
\]
in place of $\mathcal P^{\rm res}_{K,\mathfrak p,N}$:
\begin{equation}\label{eq:block-observable-leaf-factorization}
\mathcal O_{Y,N}^{\mathrm{red},L}(\alpha_Y,\Gamma)
=
\Tr_{\mathcal M_{K,\mathfrak p}(\Gamma)}\!\left(
(I^{\rm out}_{K,\mathfrak p,N})^*
\mathcal L_{K,\mathfrak p,N}
I^{\rm in}_{K,\mathfrak p,N}
\right).
\end{equation}
Here $\mathsf L_{j,N}$ is placed on the same residual factor as $\mathsf P_{\lambda_j,H_{j,N}}^{(N)}$ and is factorized as in~\eqref{eq:leaf-factorization}.

Thus a tree appendix is peeled as follows.  If a terminal leaf is an unmarked edge leaf, Lemma~\ref{lem:no-dangling-edge} says that the block observable $\mathcal O_{Y,N}^{\mathrm{red},L}(\alpha_Y,\Gamma)$ is zero.  Hence, for non-zero block observables, one may start from a vacuum plaquette leaf.  At such a leaf, Lemma~\ref{lem:partial-trace-haar-projection} proves that $\mathsf P_{\lambda_j,H_{j,N}}^{(N)}$ is a norm-one projection, and Proposition~\ref{prop:relative-peeling} compares \eqref{eq:block-observable-leaf-factorization} with the absolute residual mass \eqref{eq:absolute-residual-mass}. 

\begin{lemma}
\label{lem:no-dangling-edge}
Let $Y$ be a connected active support, let $\alpha_Y$ be a plaquette label field on $Y$, and let $\Gamma\in\mathcal C_Y^L(\alpha_Y)$.  If an active edge vertex $e\in V_E(Y)$ carries no marked Wilson strand and is incident to exactly one active plaquette incidence in $Y$, then $\mathcal O_{Y,N}^{\mathrm{red},L}(\alpha_Y,\Gamma)=0.$
\end{lemma}

\begin{proof}
Let $(p,e)$ be the unique active incidence adjacent to $e$.  In the chosen local-channel basis, all tensor legs entering $e$ through this incidence have been fused into an irreducible channel $V_\eta$, possibly replaced by its dual according to the orientation of $e$ in $\partial p$.  Since the incidence is active, this channel is non-trivial: $\eta\neq0$.  All other plaquette incidences adjacent to $e$ are vacuum incidences, and by assumption there is no marked Wilson strand at $e$.  Hence the Haar factor at $e$ is, up to harmless trivial tensor factors,
\[
\Pi_e
=
\int_{\U(N)} \rho_\eta(g) dg
=
\Proj\bigl(V_\eta^{\U(N)}\bigr),
\]
or the same expression with $V_\eta$ replaced by $V_\eta^*$.  Since $V_\eta$ is a non-trivial irreducible representation of $\U(N)$, $V_\eta^{\U(N)}=\{0\}$.  Therefore the edge projection $\Pi_e$ is the zero operator in the projection-supported realization~\eqref{eq:operator-block-contraction}.  Equivalently, after scalarization, the scalar edge kernel $K_e^N(\Gamma_e;\alpha,L)$ is a matrix coefficient of this zero projection.  Hence $\mathcal O_{Y,N}^{\mathrm{red},L}(\alpha_Y,\Gamma)=0$.
\end{proof}

The preceding lemma is only a vanishing statement.  The quantitative cancellation which performs a peeling step is the partial-trace Haar projection identity of Lemma~\ref{lem:partial-trace-haar-projection}.  In the notation of the peeled leaf, it says precisely that $\mathsf P_{\lambda_j,H_{j,N}}^{(N)}=d_{\lambda_j}\Tr_{V_{\lambda_j}}(\Pi_{\lambda_j,H_{j,N}})$ is a norm-one orthogonal projection on the residual edge space.

The next proposition is the main quantitative output of the peeling process.  It is stated in the relative form needed later: peel all tree appendices towards a prescribed residual subgraph $K$, replace each peeled appendix by the projections $\mathsf P_{\lambda_j,H_{j,N}}^{(N)}$, keep those peeled projections as norm-one black boxes, and estimate the residual matrix coefficient by resolving only the core $K$.  In applications the hypotheses on the existence of a vacuum plaquette-leaf order are guaranteed by the pruning construction of $\operatorname{Core}_*(Y)$, after discarding the zero block observables detected by Lemma~\ref{lem:no-dangling-edge}.

\begin{proposition}[Peeling estimate]\label{prop:relative-peeling}
There exists $C_d<\infty$ such that the following holds.  Let $Y$ be a connected active support and let $K\subset Y$ be a connected active subgraph containing all marked edge vertices of $Y$.  Assume that every connected component of $Y\setminus K$ is a tree attached to $K$, and let $\mathfrak p=(p_1,\ldots,p_s)$ be a plaquette-leaf peeling order for $Y\setminus K$ directed towards $K$.  Fix a reduced stable label field $\alpha_Y$ and a local-channel datum $\Gamma$.  Then
\begin{equation}
\label{eq:relative-peeling}
\left|\mathcal O_{Y,N}^{\mathrm{red},L}(\alpha_Y,\Gamma)\right|
\le
C_d^{|Y\setminus K|+M_{Y\setminus K}}
\mathcal R_N^L(K;\alpha_Y,\Gamma;\mathfrak p),
\end{equation}
where
\[
M_{Y\setminus K}:=\sum_{p\in P(Y\setminus K)}m(\alpha_p).
\]
The estimate is uniform in the peeling order $\mathfrak p$.
\end{proposition}

\begin{proof}
For each $1\le j\le s$, the leaf factor in the block observable is
\[
\mathsf L_{j,N}=B_{j,N} \mathsf P_{\lambda_j,H_{j,N}}^{(N)} A_{j,N}
\]
by \eqref{eq:leaf-factorization}.  Lemma~\ref{lem:partial-trace-haar-projection}, applied to $\lambda_j$ and $H_{j,N}$, gives
\[
\|\mathsf P_{\lambda_j,H_{j,N}}^{(N)}\|_{2\to2}\le 1,
\]
and the fixed-degree local maps satisfy \eqref{eq:leaf-AB-bound}.  Insert this factorization in the trace representation \eqref{eq:block-observable-leaf-factorization} and resolve only the residual core, exactly as in the definition of $\mathcal R_N^L(K;\alpha_Y,\Gamma;\mathfrak p)$.  For each core-resolved summand, the maps $A_{j,N}$ and $B_{j,N}$ act only on the peeled black-box factor, hence only on the vectors adjacent to $\mathsf P^{\rm peel}_{\mathfrak p,N}$.  Since the core operator and the peeled projection are kept between those vectors, Cauchy's inequality gives the summandwise bound
\[
\begin{aligned}
|B_N^{\rm block}(\omega)|
&\le
\left(\prod_{j=1}^s
\|A_{j,N}\|_{2\to2}\|B_{j,N}\|_{2\to2}\right)
|c_N(\omega)|\\
&\qquad\times
\left(\prod_{c\in\mathcal C(\omega)}d_{\eta_c}^{-1}\right)
\|U_N(\omega)\|_{2\to2}.
\end{aligned}
\]
Summing over $\omega$ gives
\[
\left|\mathcal O_{Y,N}^{\mathrm{red},L}(\alpha_Y,\Gamma)\right|
\le
\left(\prod_{j=1}^s
\|A_{j,N}\|_{2\to2} \|B_{j,N}\|_{2\to2}\right)
\mathcal R_N^L(K;\alpha_Y,\Gamma;\mathfrak p).
\]
No cancellation in the signed residual scalar is used here.  Finally,
\[
s\le |Y\setminus K|,
\qquad
\sum_{j=1}^s m(\alpha_{p_j})=M_{Y\setminus K},
\]
and the claimed estimate follows.
\end{proof}

We isolate once and for all the elementary resolved summand which will be used both for terminal pinned blocks and for fully resolved residual cores.  This is the local algebraic unit from which the estimates below are assembled. Let us start with some terminology.

A \emph{resolved elementary channel summand} is one term in the expansion of a fixed matrix-coefficient tensor network after the edge Haar projections have been resolved in orthonormal invariant bases and after a finite set of internal channels has been opened.  Its data are denoted by $\omega$: they include the binary channel trees, the opened channel set $\mathcal C(\omega)$, the labels $\eta_c$ on these channels, and the orthonormal multiplicity indices.  The labelled channel graph $G(\omega)$ has one edge for each opened channel with $\eta_c\ne0$.  Its connected components are the connected components of the labelled channel graph.

A \emph{terminal opened channel} is an opened channel whose label is one of the external plaquette labels already present in the local-channel tensor.  In the residual-core estimates we use only the canonical sequential resolution fixed in Lemma~\ref{lem:core-path-encoding}; a refined channel assignment is therefore simply the sequence of intermediate irreducible labels, equivalently the corresponding multiplicity history, compatible with that fixed resolution.

The labelled channel graph $G(\omega)$ has as vertices the local tensor pieces left after opening the channels, together with the marked Wilson-collar boundary pieces, and has one edge for each opened channel with non-trivial label. A non-trivial opened channel lying on a path between two distinct marked Wilson-collar pieces is called a \emph{separating channel}.

\begin{lemma}\label{lem:degree-pinned-block}
Let $B_N(\omega)$ be a resolved elementary channel summand.  Assume that the labelled channel graph $G(\omega)$ is connected and meets $b(\omega,L)\ge1$ normalized marked Wilson trace components.  Then
\begin{equation}\label{eq:separating-channel-elementary-summand}
B_N(\omega)
=
c_N(\omega)
\left(\prod_{c\in\mathcal C(\omega)} d_{\eta_c}^{-1}\right)
\langle \xi_N(\omega),U_N(\omega)\zeta_N(\omega)\rangle,
\end{equation}
where $c_N(\omega)=O(1)$ has a controlled expansion, $\xi_N(\omega)$ and $\zeta_N(\omega)$ are unit vectors, and $\|U_N(\omega)\|_{2\to2}\le C_\omega$.  In particular
\[
|B_N(\omega)|\le C_\omega\prod_{c\in\mathcal C(\omega)}d_{\eta_c}^{-1}\le C_\omega .
\]
If $b(\omega,L)\ge2$ and $B_N(\omega)\ne0$, then some opened channel $c_*$ has non-trivial reduced label, hence
\begin{equation}\label{eq:separating-channel-dim-loss}
d_{\eta_{c_*}}^{-1}\le N^{-1},
\qquad
|B_N(\omega)|\le C_\omega N^{-1}.
\end{equation}
\end{lemma}

\begin{proof}
All scalar coordinate changes in a fixed resolved summand are fixed-degree Schur--Weyl coefficients, hence are $O(1)$ with controlled expansions by Lemma~\ref{lem:fixed-degree-orthonormalization}.  A normalized Wilson trace is the matrix coefficient of the unit collar vector
\[
\Omega_N=N^{-1/2}\sum_{i=1}^N e_i\otimes e_i^*,
\qquad
N^{-1}\Tr(A)=\langle\Omega_N,(A\otimes\Id)\Omega_N\rangle,
\]
so marked collars produce no positive power of $N$.  Each opened channel contributes the reciprocal factor $d_{\eta_c}^{-1}$ from the contraction--coevaluation projector \eqref{eq:contraction-coevaluation-projector}; terminal channels cancel the plaquette dimension factors contained in the block observable.  This gives \eqref{eq:separating-channel-elementary-summand}.

If the connected labelled channel graph meets two marked trace components, any path between them contains a non-trivial opened channel.  Among reduced stable labels, every non-trivial representation has dimension at least $N$, which gives the extra factor $N^{-1}$.
\end{proof}

\begin{lemma}[Stable label-field count]
\label{lem:local-complexity-bound}
For fixed $Y$ and $M$, the number of reduced stable plaquette-label fields of total projective mass $M$ is bounded by
\[
C_d^{|Y|+M}.
\]
\end{lemma}

\begin{proof}
Write $P=|P(Y)|$.  The number of reduced stable labels of projective mass $m$ is
\[
p_2(m)=\sum_{a+b=m}p(a)p(b).
\]
Since
\[
\sum_{m\ge0}p_2(m)z^m=\phi(z)^{-2}
\]
has radius of convergence $1$, there exists $C_1<\infty$ such that
$p_2(m)\le C_1^{m+1}$ for all $m\ge0$.  Hence the number of fields
$\alpha_Y=(\alpha_p)_{p\in P(Y)}$ with $\sum_pm(\alpha_p)=M$ is bounded by
\[
\sum_{\substack{(m_p)_{p\in P(Y)}\\ \sum_pm_p=M}}
\prod_{p\in P(Y)}C_1^{m_p+1}
\le
C_1^{M+P}\binom{M+P}{P}
\le C_2^{M+P}.
\]
The case $P=0$ is immediate.  Since $P\le |Y|$, the claim follows after increasing $C_d$.
\end{proof}

\subsection{Pinned residual cores and recoupling}\label{subsec:residual-core}

After the peeling process, the remaining graph $K=\operatorname{Core}_*(Y)$ is still pinned but may contain cycles or Wilson-anchored chains.  We now estimate the residual scalar directly by the channel-resolution machinery of Subsections~\ref{subsec:orthonormal-channel-recoupling} and~\ref{subsec:stable-young-graphs}.  The Wilson-collar factors are kept as fixed boundary tensors, while only the plaquette-generated part is fully recoupled.  Thus every refined internal label is produced by the plaquette labels alone.

\begin{lemma}[Canonical stable path encoding of a core resolution]\label{lem:core-path-encoding}
There is a constant $C_d<\infty$ with the following property.  Let $K$ be a residual core with reduced stable plaquette labels of total projective mass $M\le N/2$.  Fix once and for all a deterministic order of the plaquette vertices, of the elementary $V_N/V_N^*$ strands inside each plaquette representation, and of the incidences at every edge vertex; use the corresponding left-associated binary bracketing throughout the core.  Then every free internal irreducible label in the resulting core resolution is generated by elementary stable Pieri moves along the reduced Young graph.  We open the intermediate irreducible channel after every elementary Pieri tensoring, so the multiplicity history is exactly the sequence of visited labels.  The total length of all such paths is at most $C_d(|K|+M)$, and every visited intermediate label contributes the reciprocal dimension factor of its opened channel.

In particular, after the local operator norms have been removed, the sum over \emph{all} internal labels and multiplicity histories in this canonical resolution satisfies
\begin{equation}\label{eq:canonical-core-path-sum}
\sum_{\text{canonical internal histories}}
\prod_{c\in\mathcal C(\omega)}d_{\eta_c}^{-1}
\le C_d^{|K|+M}.
\end{equation}
\end{lemma}

\begin{proof}
For each plaquette $p$, realize $V_{\alpha_p}$ in mixed tensor degree $m(\alpha_p)$ and introduce its elementary $V_N$- and $V_N^*$-strands in the fixed order.  Exactly $m(\alpha_p)$ elementary tensorings are required.  By the Pieri rules~\eqref{eq:Pieri-stable}--\eqref{eq:Pieri-stable2}, every tensoring changes exactly one of the two stable partitions by one box and is multiplicity-free.  Summing over the plaquettes gives $M$ elementary representation-theoretic moves.

The bounded-degree incidence graph contributes only $O_d(|K|)$ further openings and routings.  Because the ordering and bracketing were fixed in advance, there is no sum over permutations of fusion schedules or over binary bracketings.  All multiplicity information is instead the choice of intermediate labels along the resulting stable oscillating paths.  Only plaquette-generated strands are recoupled, so every intermediate projective mass is at most $M\le N/2$.

At each elementary tensoring we insert the orthonormal resolution of the corresponding intermediate irreducible channel on the two sides of the cut.  In the closed tensor contraction this is the contraction--coevaluation projector~\eqref{eq:contraction-coevaluation-projector}, hence contributes $d_{\eta}^{-1}$.  Because the Pieri step is multiplicity-free, a complete multiplicity history is uniquely encoded by the visited labels.  Applying Lemma~\ref{lem:oscillating-path} successively to the finitely many canonical paths, whose total length is at most $C_d(|K|+M)$, gives~\eqref{eq:canonical-core-path-sum} after increasing $C_d$.
\end{proof}

\begin{lemma}[Core connectivity of marked collars]\label{lem:core-collar-connectivity}
Let $Y$ be connected and pinned, and let $K=\operatorname{Core}_*(Y)$. In a non-zero core-resolved summand, the unresolved peeled factors cannot create a connection between two distinct marked Wilson collars which is absent from the labelled channel graph in $K$. Consequently, if a connected summand meets at least two marked Wilson collars, then a connected labelled component of the core meets at least two such collars. In particular, when $b(K,L)\ge2$, Lemma~\ref{lem:degree-pinned-block} produces a non-trivial separating channel inside the core.
\end{lemma}

\begin{proof}
Every connected component of $Y\setminus K$ is a tree appendix attached to $K$ through a single attachment datum, by construction of the pruning procedure. Relative peeling replaces such an appendix by the norm-one operator $\mathsf P^{\rm peel}_{\mathfrak p,N}$ acting on that attachment space; it does not introduce an identification between two different attachment spaces. Hence contracting the unresolved peeled factors can modify the local operator at one attachment but cannot join two labelled connected components of the core. Any connection between two marked collars in the full non-zero summand must therefore already be present in the core-resolved labelled channel graph. The final assertion is precisely the separating-channel conclusion of Lemma~\ref{lem:degree-pinned-block} applied to that connected core component.
\end{proof}

\begin{theorem}[Residual-core estimate]\label{thm:core-plancherel}
Let $K$ be a connected pinned residual core with reduced stable plaquette labels
\[
\alpha=(\alpha_p)_{p\in P(K)},\qquad
M=\sum_{p\in P(K)}m(\alpha_p)\le N/2,
\]
and fix a local-channel datum $\Gamma$ inducing data on $K$.  Then, for every peeling order $\mathfrak p$ which produces the residual data on $K$,
\begin{equation}\label{eq:residual-core-bound}
\mathcal R_N^L(K;\alpha_Y,\Gamma;\mathfrak p)
\le C_L C_d^{|K|+M}N^{-\delta(K,L)},
\end{equation}
where $\delta(K,L)\ge0$, and $\delta(K,L)\ge1$ if $b(K,L)\ge2$.  For fixed $K$ and $M$, the signed residual scalar $R_N^L(K;\alpha_Y,\Gamma;\mathfrak p)$ admits a controlled asymptotic expansion to all orders in powers of $1/N$.
\end{theorem}

\begin{proof}
Resolve only the residual core, as in the definition of $\mathcal R_N^L$.  Thus the edge Haar projections belonging to $K$ are expanded in orthonormal Clebsch--Gordan bases, every external plaquette representation $V_{\alpha_p}$, $p\in P(K)$, is realized in mixed tensor degree $m(\alpha_p)$, and the plaquette-generated elementary $V_N$- and $V_N^*$-legs inside the core are introduced along a one-box fusion schedule.  The Wilson collars are kept as fixed boundary tensors.  The peeled factor $\mathsf P^{\rm peel}_{\mathfrak p,N}$ is not resolved and contributes only through the bound
\[
\|\mathsf P^{\rm peel}_{\mathfrak p,N}\|_{2\to2}\le1 .
\]
Therefore no recoupling multiplicity coming from the peeled tree appendices is counted in this estimate.

By the construction of $\mathcal R_N^L$, each core-resolved summand is bounded by
\[
|c_N(\omega)|
\left(\prod_{c\in\mathcal C(\omega)}d_{\eta_c}^{-1}\right)
\|U_N(\omega)\|_{2\to2}.
\]
Use the canonical ordering and bracketing of Lemma~\ref{lem:core-path-encoding}.  The fixed local coordinate maps are kept as operators; by Lemma~\ref{lem:local-operator-complexity}, the product of their norms is bounded by $C_LC_d^{|K|+M}$.  Terminal opened channels have labels among the external plaquette labels of the core, and their reciprocal dimension factors cancel the plaquette dimensions already included in the block observable.

There is no counting of raw multiplicity bases.  Lemma~\ref{lem:core-path-encoding} gives directly
\[
\sum_{\text{canonical internal histories}}
\prod_{c\in\mathcal C(\omega)}d_{\eta_c}^{-1}
\le C_d^{|K|+M}.
\]
Combining this weighted path sum with the operator-norm bound and enlarging $C_d$ gives, after summing the core-resolved majorants which define $\mathcal R_N^L$,
\[
\mathcal R_N^L(K;\alpha_Y,\Gamma;\mathfrak p)
\le C_LC_d^{|K|+M}.
\]
If $b(K,L)\ge2$, Lemma~\ref{lem:core-collar-connectivity} gives a non-trivial separating channel in the core. Since every non-trivial reduced stable $\U(N)$-label has dimension at least $N$, its reciprocal-dimension factor contributes an additional factor $N^{-1}$. This proves~\eqref{eq:residual-core-bound}.

For the signed residual scalar, one keeps the actual core-resolved expansion with the norm-one peeled factor $\mathsf P^{\rm peel}_{\mathfrak p,N}$ instead of taking absolute majorants.  For fixed $K$ and $M$, all core mixed degrees, core recoupling types and stable core intermediate labels range over finite sets. Lemma~\ref{lem:fixed-degree-orthonormalization} therefore gives controlled asymptotic expansions for all local core matrix coefficients, and the signed finite sum inherits a controlled expansion to all orders.
\end{proof}

\subsection{Proof of the pinned block estimate}

\begin{proof}[Proof of Theorem~\ref{thm:pinned-block}]
Fix a reduced stable label field $\alpha_Y$ of total mass $M$.  By Lemma~\ref{lem:support-resummation}, the exact local-channel resummation~\eqref{eq:resummed-reduced-block} expresses the fixed-support block as a linear combination of at most $C_LC_d^{|Y|}$ projection-supported tensor contractions, with all edge Haar integrals left as orthogonal projections and with coefficients of modulus at most one.  It is therefore enough to estimate one such projection-supported contraction; the support-resummation factor is absorbed into the constants.

Apply the pruning construction to each such projection-supported term.  If an unmarked dangling edge occurs, Lemma~\ref{lem:no-dangling-edge} annihilates the term.  Otherwise every component outside the residual core $K=\operatorname{Core}_*(Y)$ is removed by the plaquette-leaf factorization~\eqref{eq:leaf-factorization}.  Lemma~\ref{lem:partial-trace-haar-projection} replaces the dimension factor of each peeled plaquette together with its adjacent Haar integral by the norm-one projection
$\mathsf P_{\lambda,H}^{(N)}$.  The remaining local incidence maps have norm bounded by Lemma~\ref{lem:local-operator-complexity}.  Consequently the total cost of the peeled appendices is at most
\[
C_d^{|Y\setminus K|+M_{Y\setminus K}}.
\]

Only the residual core is now resolved.  We use the canonical ordering and bracketing of Lemma~\ref{lem:core-path-encoding}; in particular, no raw walled-Brauer basis and no family of fusion schedules is counted.  Lemma~\ref{lem:degree-pinned-block} expresses every resolved core contribution as a bounded matrix coefficient multiplied by the reciprocal dimensions of its opened channels.  Lemma~\ref{lem:core-path-encoding} sums these reciprocal-dimension weights over all internal labels and multiplicity histories and yields a factor at most
\[
C_LC_d^{|K|+M_K}.
\]
The peeled projections remain norm-one black boxes throughout this resolution.  If the connected support meets at least two marked trace components, Lemma~\ref{lem:core-collar-connectivity} supplies a non-trivial separating channel in the core.  Since every non-trivial reduced stable representation has dimension at least $N$, one reciprocal-dimension factor gives an additional $N^{-1}$.

Combining the peeled and core bounds, and using
$M_{Y\setminus K}+M_K=M$ and
$|Y\setminus K|+|K|\le C_d|Y|$, gives, for the resummed fixed-label block,
\[
\left|\mathcal O_{Y,N}^{\mathrm{red},L}(\alpha_Y)\right|
\le
C_LC_d^{|Y|+M}N^{-\delta(Y,L)},
\]
where $\delta(Y,L)\ge0$ and $\delta(Y,L)\ge1$ when $b(Y,L)\ge2$.

It remains only to sum over the external reduced stable label fields of total mass $M$.  Lemma~\ref{lem:local-complexity-bound} gives at most $C_d^{|Y|+M}$ such fields.  Enlarging $C_d$ once more therefore yields
\[
\left\|\mathcal O_{Y,N}^{\mathrm{red},L}\right\|_{M,1}
\le
C_LC_d^{|Y|+M}N^{-\delta(Y,L)},
\]
which is~\eqref{eq:pinned-block-bound}.

Finally fix $Y$ and $M$.  There are only finitely many external stable label fields of total mass $M$.  For each one, the signed projection-supported contraction may be evaluated with the same canonical core resolution.  All mixed degrees are fixed, the local operator coefficients have controlled expansions by Lemma~\ref{lem:local-operator-complexity}, and the finite signed sums preserve such expansions.  Summing over the finitely many external label fields proves the asserted all-order expansion of $\mathcal O_{Y,N}^{\mathrm{red},L}[M]$.
\end{proof}

\begin{corollary}[Coarse non-spectral block bound]\label{cor:coarse-nonspectral-block}
There are constants $C_d<\infty$ and $C_L<\infty$ such that, for every connected pinned support $Y$ and every reduced plaquette-label field $\alpha_Y$ of arbitrary total projective mass
$M=M(\alpha_Y)$,
\begin{equation}\label{eq:coarse-nonspectral-block}
\left|\mathcal O_{Y,N}^{\mathrm{red},L}(\alpha_Y)\right|
\le C_LC_d^{|Y|+M}.
\end{equation}
The same estimate holds after arbitrary determinant shifts are reinserted, provided the central charge selection rule is imposed.
\end{corollary}

\begin{proof}
Repeat the projection-supported resummation, pruning and peeling argument from the proof of Theorem~\ref{thm:pinned-block}.  The leaf estimates and the norm-one partial-trace Haar projections do not use the stable mass restriction.  In the residual core, resolve every reduced plaquette representation into elementary $V_N/V_N^*$ strands in the same fixed canonical order as before, but now regard all intermediate labels as vertices of the full Pieri graph $\mathbb P_N^{\mathrm{full}}$.  Open every intermediate channel.  Lemma~\ref{lem:full-pieri-path} replaces Lemma~\ref{lem:oscillating-path} and bounds the sum of the reciprocal-dimension weights by $C_d^{|K|+M}$ with no restriction on $M$.  This yields~\eqref{eq:coarse-nonspectral-block}; we do not claim the additional $N^{-1}$ separating-channel gain outside the retained stable sector.

A determinant twist is one-dimensional.  In the tensor network it is therefore carried by scalar determinant wires of norm one; Haar invariance either annihilates the term by the central selection rule or leaves the reduced operator norm unchanged.  Hence the same bound holds after determinant shifts are restored.
\end{proof}

\section{Rooted master loop equation}\label{sec:rooted-master-loop}

In this section we recall the coefficientwise master loop equation, root it at a single occurrence, and use the resulting heat-kernel equation to prove a large-$N$ asymptotic expansion of rooted Wilson observables.
\subsection{Coefficientwise master loop equation and extended observables}\label{subsec:coeff-MLE}

We first recall the universal coefficientwise master loop equation of~\cite{Lem26a}.  It is obtained by Haar integration by parts in one edge variable of the topological coefficient~\eqref{eq:kappa-W}.  The equation is independent of the action; the action enters only after multiplying by the spectral weights and summing over plaquette labels.

More explicitly, with the notation of the companion paper, for every oriented edge $e$, every loop family $L$, and every plaquette decoration $\alpha:P(\Lambda)\to\widehat{\U(N)}$, the unrooted coefficientwise equation is
\begin{equation}\label{eq:coeffwise-MLE}
(\mathscr L_e\widehat W_{\Lambda,L})(\alpha)
+
\sum_{p\ni e}(\mathscr B_{e,p}\widehat W_{\Lambda,L})(\alpha)=0.
\end{equation}
Here $\mathscr L_e$ is a cut-and-join operator acting on the loop side, and $\mathscr B_{e,p}$ is a local loop-plaquette recoupling operator acting simultaneously on the Fourier side and the loop side. The rooted versions of these two operations are made explicit in Subsection~\ref{subsec:normalform}, where we prove the single-occurrence form used in the sequel. Rooted, single-location loop equations also appear in the Wilson-action setting of Cao--Park--Sheffield~\cite{CPS25}.

Throughout this section, $\widehat Q^{\HK}_{T,N}$ denotes the heat-kernel coefficient fixed in~\eqref{eq:Fourier-coef-HK}, and $\kappa_{\Lambda,T,N}$ is the corresponding spectral coefficient~\eqref{eq:heat-kernel-kappa}. Let $C^\infty_\Lambda=C^\infty(\U(N)^{E(\Lambda)})$. More generally, all differentiations below take place on the algebra of smooth cylinder functions on $\U(N)^{E(\Lambda)}$ generated by matrix coefficients of edge holonomies, Wilson loop traces and plaquette characters. We equip $\mathfrak u(N)$ with the Hilbert--Schmidt inner product $\langle X,Y\rangle_{\mathfrak u(N)}=-\Tr(XY).$ Let $(X_A)_{A=1}^{N^2}$ be an orthonormal basis of $\mathfrak u(N)$. For $e\in E(\Lambda)$, $X\in\mathfrak u(N)$, and $F\in\mathcal C^1_\Lambda$, define the left-invariant differential operator in the $e$-variable by
\[
(\partial_{e,X}F)(U)=\left.\frac{d}{dt}\right|_{t=0}F\bigl((U_f)_{f\neq e}, e^{tX}U_e\bigr).
\]
We write $\partial_{e,A}:=\partial_{e,X_A}$. With this convention,
\[
\partial_{e,A}U_e=X_AU_e,\qquad\partial_{e,A}U_e^{-1}=-U_e^{-1}X_A,
\]
which is the convention used in the rooted integration-by-parts computation below. If an oriented word contains the formal inverse $e^{-1}$, it is always interpreted through the variable $U_e^{-1}$; no additional edge variable is introduced.

\begin{definition}
Let $0$ denote the trivial representation of $\U(N)$. A \emph{compactly supported plaquette constraint} is a finite partial map $\zeta:S(\zeta)\longrightarrow \widehat{\U(N)}$, where $S(\zeta)\Subset P(\Lambda).$
\end{definition}
The indicator of a compactly supported plaquette constraint on full plaquette decorations is
\[
1_\zeta(\alpha):=\prod_{p\in S(\zeta)}1_{\{\alpha_p=\zeta(p)\}} .
\]
The empty constraint is denoted by $\varnothing$, and corresponds to the ordinary Wilson observable. We also write $D(\zeta):=\{p\in S(\zeta):\zeta(p)\neq 0\}$ for the non-trivial support of $\zeta$, and $\mathcal E(\zeta):=\sum_{p\in D(\zeta)}\mathcal E(\zeta(p))$ for its energy. If $p\in P(\Lambda)$ and $\lambda\in\widehat{\U(N)}$, the update
$\zeta^{p\to\lambda}$ is the compactly supported constraint with support $S(\zeta)\cup\{p\}$, value $\lambda$ at $p$, and the same values as $\zeta$
away from $p$. For a compactly supported plaquette decoration $\zeta$, define the extended Wilson coefficient by
\begin{equation}\label{eq:atomic-extended-coef}
\Phi_{\Lambda,T,N}(L;\zeta)
=
\frac{
\widehat\E_{\Lambda,T,N}\left[\mathbf 1_\zeta(\lambda)\mathcal O^L_{\Lambda,N}(\lambda)\right]
}{
\widehat\E_{\Lambda,T,N}\left[\mathcal O^\varnothing_{\Lambda,N}(\lambda)\right]
}.
\end{equation}
In particular, $\Phi_{\Lambda,T,N}(L;\varnothing)=\Phi_{\Lambda,T,N}(L).$ The formula extends by linearity to compactly supported plaquette-decoration test functions.

\begin{definition}\label{def:rooted-spectral-loop-states}
A \emph{rooted decorated-loop state} is a triple $X=(L,x,\zeta),$ where $L$ is a finite loop family, $x$ is a chosen occurrence of an oriented edge $e_x$ in one of the loops of $L$, and $\zeta$ is a compactly supported plaquette decoration.  We write $\ell(L)$ for the total word length of $L$ and put $|X|_*:=\ell(L)+|S(\zeta)|,$ $\mathcal E(X):=\mathcal E(\zeta).$
\end{definition}

Let $\mathfrak X_\Lambda$ denote the set of rooted decorated-loop states in $\Lambda$. We also use a cemetery symbol $\dagger$ for terminal outputs, which are not elements of $\mathfrak X_\Lambda$. Fix once and for all a deterministic rule $\mathfrak r$ which, to every finite loop family $M$, assigns the first non-trivial edge occurrence in a fixed ordering of the output words, if such an occurrence exists, and assigns $\mathfrak r(M)=\dagger$ otherwise. After every local operation, the outgoing rooted state is defined using the root $\mathfrak r(M)$ when $\mathfrak r(M)\ne\dagger$; if $\mathfrak r(M)=\dagger$, the output is terminal and is placed in the source.

For $a,\eta>0$ we use the weighted supremum norm on functions of rooted states
\begin{equation}\label{eq:raw-weighted-norm}
\|H\|_{a,\eta}:=
\sup_X e^{-a|X|_* -\eta \mathcal E(X)}|H(X)|.
\end{equation}
The corresponding operator estimates below are formulated as weighted row-sum bounds with respect to this norm.
Finally, let
\[
\mu(X):=\#\{\hbox{occurrences of } e_x^{\pm1}\hbox{ in the loop family }L\}
\]
be the rooted-edge multiplicity.  The root is part of the state, and all local steps below are performed at that occurrence. The deterministic rule $\mathfrak r$ is used only to choose the next root after a non-terminal local operation; the value of $\Psi_{\Lambda,T,N}$ itself is independent of this bookkeeping choice.

We denote by $\Psi_{\Lambda,T,N}(L,x,\zeta):=\Phi_{\Lambda,T,N}(L;\zeta)$ the rooted version of the extended Wilson coefficient. Its value does not depend on the root $x$, but the rooted equation computing it does.

\subsection{Rooted coefficientwise normal form}\label{subsec:normalform}

We now prove the rooted coefficientwise identity used in the rest of the paper. It is the single-occurrence version of the coefficientwise master loop equation~\eqref{eq:coeffwise-MLE}. We first recall the local operations and fix the rooted notation.  Choose cyclic representatives
\[
w_{\ell_i}=e_{i,1}^{\varepsilon_{i,1}}\cdots e_{i,m_i}^{\varepsilon_{i,m_i}},
\qquad \varepsilon_{i,r}\in\{\pm1\},
\]
for the loops of $L=(\ell_1,\ldots,\ell_k)$.  For an oriented edge $e$, set
\[
\begin{aligned}
A_i(e;L)&=\{r:e_{i,r}^{\varepsilon_{i,r}}=e\},\\
B_i(e;L)&=\{r:e_{i,r}^{\varepsilon_{i,r}}=e^{-1}\},\\
C_i(e;L)&=A_i(e;L)\cup B_i(e;L).
\end{aligned}
\]
An occurrence $x\in C_i(e;L)$ is said to have sign $+$ if $x\in A_i(e;L)$, and sign $-$ if $x\in B_i(e;L)$.  We write $[w]$ for the loop represented by the word $w$, up to cyclic rotation.  The rooted surgery $L_{x,y}$ is defined as follows, with the cyclic representatives chosen so that the rooted occurrence $x$ appears first.

If $x$ and $y$ lie on the same loop $\ell_i$, then $L_{x,y}$ is obtained by replacing $\ell_i$ by two loops $(\ell_{i,1},\ell_{i,2})$:
\[
\begin{array}{lll}
[ae be c] & \longrightarrow & ([ae c],[be]),\\[2mm]
[ae^{-1}be^{-1}c] & \longrightarrow & ([ae c],[be^{-1}]),\\[2mm]
[ae be^{-1}c] & \longrightarrow & ([ac],[b]),\\[2mm]
[ae^{-1}be c] & \longrightarrow & ([ac],[b]).
\end{array}
\]
The first two cases are the positive splittings, and the last two are the negative splittings.  If $x\in C_i(e;L)$ and $y\in C_j(e;L)$ with $i\neq j$, then $L_{x,y}$ is obtained by replacing $\ell_i$ and $\ell_j$ by one loop $\ell_{ij}$:
\[
\begin{array}{lll}
([ae b],[ce d]) & \longrightarrow & [ae dc e b],\\[2mm]
([ae^{-1}b],[ce^{-1}d])& \longrightarrow & [ae dc e b],\\[2mm]
([ae b],[ce^{-1}d]) & \longrightarrow & [adcb],\\[2mm]
([ae^{-1}b],[ce d]) & \longrightarrow & [adcb].
\end{array}
\]
The first two cases are the positive mergers, and the last two are the negative mergers. Splittings and mergers are illustrated in Figure~\ref{fig:split-merge}.

\begin{figure}[h!]
\centering
\includegraphics[width=0.9\linewidth]{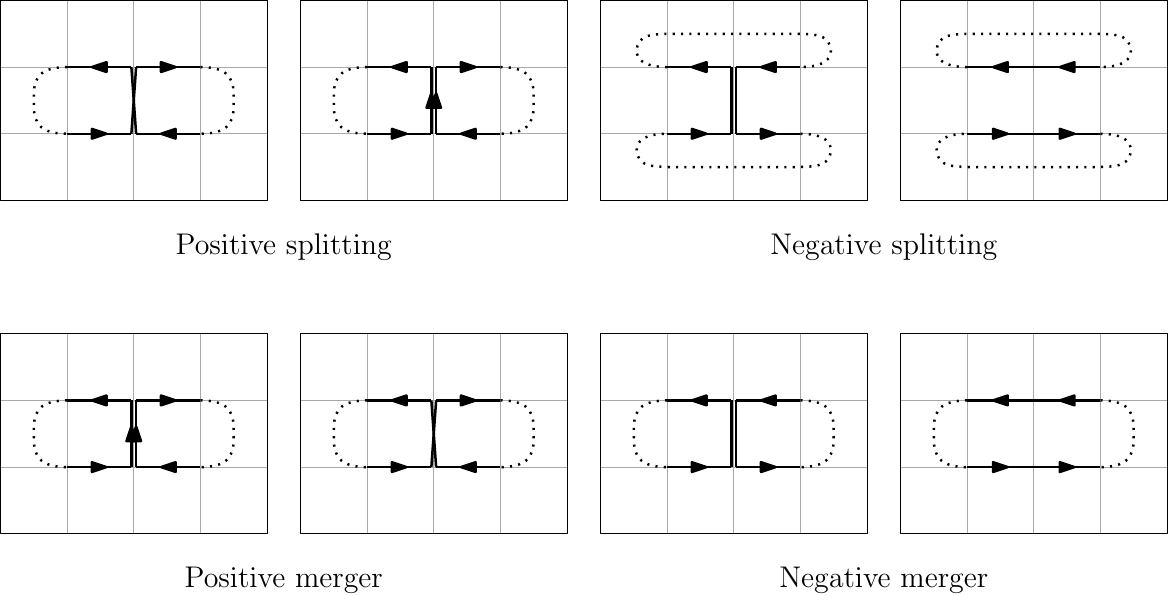}
\caption{Local representations of splittings and mergers. The dotted parts may be arbitrary lattice paths.}
\label{fig:split-merge}
\end{figure}

Define
\[
\sigma_{\mathrm{cj}}(x,y)=
\begin{cases}
-1, & x\text{ and }y\text{ have the same sign},\\
+1, & x\text{ and }y\text{ have opposite signs}.
\end{cases}
\]
Thus positive splittings and positive mergers carry the sign $-1$, whereas negative splittings and negative mergers carry the sign $+1$, exactly as in the unrooted operator $\mathscr L_e$.

For a fixed root $x\in C_i(e;L)$, the rooted loop-only operator is
\begin{equation}\label{eq:rooted-S-operator}
(\mathscr S_x\widehat W_{\Lambda,\bullet})(L,\alpha)
:=
\sum_{\substack{y\neq x\\ y\in \cup_j C_j(e;L)}}
\sigma_{\mathrm{cj}}(x,y) \widehat W_{\Lambda,L_{x,y}}(\alpha).
\end{equation}
If $\mathfrak r(L_{x,y})\ne\dagger$, the outgoing root is $x_{x,y}:=\mathfrak r(L_{x,y})$; if $\mathfrak r(L_{x,y})=\dagger$, the corresponding output is terminal.  The deterministic outgoing-root rule is only a bookkeeping device. The estimate does not rely on monotonicity of the number of occurrences of the rooted edge, which may fail for positive same-orientation splittings. Finiteness at fixed asymptotic order will instead be supplied by the finite-order loop-only entropy estimate, Lemma~\ref{lem:loop-only-entropy}.

We now spell out the rooted loop-plaquette recoupling statement which is used to define the operator $\mathscr B^{\mathrm{root}}_{x,p}$. Fix an oriented incidence $(e,p)$, choose a cyclic representative $\partial p=aeb$, and let $x\in C_i(e;L)$ be a rooted occurrence in the loop $\ell_i$.  Write, in determinant-reduced mixed-tensor notation,
\[
\alpha_p=[\lambda_p^+,\lambda_p^-]_N,
\qquad
n_p^+=|\lambda_p^+|,
\qquad
n_p^-=|\lambda_p^-|,
\qquad
r_p=n_p^++n_p^- .
\]
For $1\le u\le r_p$, set
\[
\gamma_{p,u}:=
\begin{cases}
aeb=\partial p, & 1\le u\le n_p^+,\\[1mm]
b^{-1}e^{-1}a^{-1}=(\partial p)^{-1}, & n_p^+<u\le r_p,
\end{cases}
\]
so that covariant slots carry the positively oriented plaquette word and contravariant slots carry the oppositely oriented plaquette word.  Let $y_u$ be the active occurrence of $e$ or $e^{-1}$ in $\gamma_{p,u}$.  The rooted loop-plaquette surgery
\[
L_{x,p,u}:=\tau_{x,u}(L)
\]
is obtained by replacing $\ell_i$ with the merger of $\ell_i$ and $\gamma_{p,u}$ at the active occurrences $(x,y_u)$: it is the positive merger $\ell_i\oplus_{x,y_u}\gamma_{p,u}$ if the two active occurrences have the same orientation, and the negative merger $\ell_i\ominus_{x,y_u}\gamma_{p,u}$ if they have opposite orientations.  All other loops are left unchanged. Attach to the pair $(x,u)$ the auxiliary mixed tensor space
\[
T^{\mathrm{aux}}_{x,p,u}:=
\begin{cases}
V^{\otimes(n_p^+-1)}\otimes (V^*)^{\otimes n_p^-},
& x\in A_i(e;L),\ 1\le u\le n_p^+,\\[1mm]
V^{\otimes(n_p^++1)}\otimes (V^*)^{\otimes n_p^-},
& x\in B_i(e;L),\ 1\le u\le n_p^+,\\[1mm]
V^{\otimes n_p^+}\otimes (V^*)^{\otimes(n_p^-+1)},
& x\in A_i(e;L),\ n_p^+<u\le r_p,\\[1mm]
V^{\otimes n_p^+}\otimes (V^*)^{\otimes(n_p^--1)},
& x\in B_i(e;L),\ n_p^+<u\le r_p.
\end{cases}
\]
Equivalently, the four cases are: covariant contraction, covariant coevaluation, contravariant coevaluation, and contravariant contraction, respectively. In all cases $T^{\mathrm{aux}}_{x,p,u}$ differs from $T_{\alpha_p}$ by one elementary $V_N$- or $V_N^*$-leg. Denote by $D^x_{e,A}W_L$ the contribution to $\partial_{e,A}W_L$ in which the derivative falls at the single occurrence $x$:
\[
D^x_{e,A}W_L(U)=
\begin{cases}
\displaystyle
\Tr(U_aX_AU_eU_b)\prod_{j\neq i}\Tr(U_{\ell_j}),
& x\in A_i(e;L),\\[2mm]
\displaystyle
-\Tr(U_aU_e^{-1}X_AU_b)\prod_{j\neq i}\Tr(U_{\ell_j}),
& x\in B_i(e;L).
\end{cases}
\]

The following is the rooted specialization of Propositions~4.8 and~4.9 of \cite{Lem26a}, including the explicit isotypic-trace formula for the recoupling coefficients.

\begin{proposition}\label{prop:rooted-local-incidence-recoupling}
For every $1\le u\le r_p$, there is a $\U(N)$-equivariant endomorphism $A_{x,p,u;N}\in\End(T^{\mathrm{aux}}_{x,p,u}),$ independent of the gauge variables, such that
\[
\sum_A D^x_{e,A}W_L(U) \partial_{e,A}\chi_{\alpha_p}(U_{\partial p})
=
\sum_{u=1}^{r_p}
W_{L_{x,p,u}}(U)
\Tr_{T^{\mathrm{aux}}_{x,p,u}}
\bigl(A_{x,p,u;N}\rho_{x,p,u}(U_{\partial p})\bigr).
\]
Moreover, for each $u$, there is a finite set $\mathcal B_{x,p,u}(\alpha_p)$ of stable Pieri moves such that
\begin{equation}\label{eq:rooted-pieri-character-expansion}
\Tr_{T^{\mathrm{aux}}_{x,p,u}}\bigl(A_{x,p,u;N}\rho_{x,p,u}(g)\bigr)=\sum_{\beta\in\mathcal B_{x,p,u}(\alpha_p)}c^N_{x,p,u}(\alpha_p,\beta)\chi_\beta(g),
\end{equation}
where
\begin{equation}\label{eq:rooted-pieri-coefficient}
c^N_{x,p,u}(\alpha_p,\beta)
=
d_\beta^{-1}
\Tr_{T^{\mathrm{aux}}_{x,p,u}}
\bigl(A_{x,p,u;N}P^{T^{\mathrm{aux}}}_{\beta,N}\bigr),
\end{equation}
and $P^{T^{\mathrm{aux}}}_{\beta,N}$ is the orthogonal projection onto the $\beta$-isotypic component.
\end{proposition}

\begin{definition}
The \emph{rooted loop-plaquette operator} is the operator $\mathscr B_{x,p}^{\mathrm{root}}$ defined by
\begin{equation}\label{eq:rooted-B-operator}
(\mathscr B^{\mathrm{root}}_{x,p}\widehat W_{\Lambda,\bullet})(L,\alpha):=\sum_u\sum_{\beta\in\mathcal B_{x,p,u}(\alpha_p)}c^N_{x,p,u}(\alpha_p,\beta) \widehat W_{\Lambda,L_{x,p,u}}(\alpha^{p\to\beta}).
\end{equation}
\end{definition}
The rooted loop-plaquette operator has the finite local kernel
\begin{equation}\label{eq:Broot-local-kernel}
(\mathscr B^{\mathrm{root}}_{x,p}\widehat W_{\Lambda,\bullet})(L,\alpha)
=
\sum_u\sum_{\beta_p}
b^{N,\mathrm{root}}_{x,p,u}(L,\alpha_p;L_u,x_u,\beta_p)
\widehat W_{\Lambda,L_u}(\alpha^{p\to\beta_p}),
\end{equation}
where $u$ ranges over finitely many tensor slots and local orientation types, $L_u=L_{x,p,u}$, and $x_u=\mathfrak r(L_u)$ if $\mathfrak r(L_u)\ne\dagger$, while $x_u=\dagger$ for terminal outputs, setting
\[
b^{N,\mathrm{root}}_{x,p,u}(L,\alpha_p;L_u,x_u,\beta_p)=c^N_{x,p,u}(\alpha_p,\beta_p)
\]
after grouping terms which give the same output.  If $x_u=\dagger$, the term is terminal and is put into the source.  Moreover, if $\alpha_p$ is the trivial representation, the corresponding rooted plaquette-transfer term is zero.

The same central-character bookkeeping gives the local charge balance which will be used in the area-law argument.  Let $\boldsymbol 1_p\in\Z^{P(\Lambda)}$ be the function equal to $1$ at $p$ and to $0$ elsewhere.  Whenever
\[
b^{N,{\rm root}}_{x,p,u}(L,\alpha_p;L_u,\mathfrak r(L_u),\beta_p)\neq0,
\]
one has
\begin{equation}\label{eq:local-reduced-charge-balance}
j_{L_u}-j_L-\partial^*\bigl((s(\alpha_p)-s(\beta_p))\boldsymbol 1_p\bigr)=0.
\end{equation}
Equivalently, if a transfer is recorded by adding the signed charge $s(\alpha_p)-s(\beta_p)$ to the aggregate plaquette current, then the quantity $j-\partial^*s$ is unchanged by the transfer.

The following lemma gives the rooted version of~\eqref{eq:coeffwise-MLE}.

\begin{lemma}[Rooted coefficientwise master loop equation]\label{lem:rooted-coeff-normal-form}
Let $x$ be a rooted occurrence of the oriented edge $e$ in a loop family $L$.  For every plaquette decoration $\alpha:P(\Lambda)\to\widehat{\U(N)}$,
\begin{equation}\label{eq:rooted-coeff-MLE}
N\widehat W_{\Lambda,L}(\alpha)
=
(\mathscr S_x\widehat W_{\Lambda,\bullet})(L,\alpha)
+
\sum_{p\ni e}(\mathscr B^{\mathrm{root}}_{x,p}\widehat W_{\Lambda,\bullet})(L,\alpha).
\end{equation}
\end{lemma}

\begin{proof}
Let $x\in C_i(e;L)$ be the rooted occurrence.  If $x\in A_i(e;L)$, choose a cyclic representative $w_{\ell_i}=a e b,$ where the displayed $e$ is the occurrence $x$.  If $x\in B_i(e;L)$, choose instead $w_{\ell_i}=a e^{-1}b.$ Let $(X_A)_{A=1}^{N^2}$ be an orthonormal basis of $\mathfrak u(N)$ for the Hilbert--Schmidt inner product. We shall use the ``magic formulas'' \cite{DHK13}
\begin{equation}\label{eq:magic-formula}
\begin{aligned}
\sum_A X_A^2&=-NI_N,\\
\sum_A X_AUX_A&=-\Tr(U)I_N,\\
\sum_A\Tr(X_AU)\Tr(X_AV)&=-\Tr(UV).
\end{aligned}
\end{equation}
If we set $Y^x_{e,A}(U):=-D^x_{e,A}W_L(U)$, then Haar invariance in the $U_e$-variable gives, for every $A$,
\[
\int_{\U(N)^{E(\Lambda)}}\partial_{e,A}
\left(
Y^x_{e,A}(U)\prod_{p\in P(\Lambda)}\chi_{\alpha_p}(U_{\partial p})
\right)dU=0.
\]
Summing over $A$, we obtain
\begin{align}\label{eq:rooted-ibp-expanded}
0={}&
\int\sum_A(\partial_{e,A}Y^x_{e,A})(U)
\prod_{p}\chi_{\alpha_p}(U_{\partial p})\,dU\notag\\
&+
\sum_{p\ni e}
\int\sum_A Y^x_{e,A}(U)
\partial_{e,A}\chi_{\alpha_p}(U_{\partial p})
\prod_{q\neq p}\chi_{\alpha_q}(U_{\partial q})\,dU.
\end{align}

We first identify the loop-only part.  Since $Y^x_{e,A}$ is the negative rooted derivative, differentiating the same occurrence $x$ once more gives the positive diagonal term.  Indeed, if $x\in A_i(e;L)$, then
\[
\sum_A \partial_{e,A}\bigl(-\Tr(U_aX_AU_eU_b)\bigr)
=
-\sum_A\Tr(U_aX_A^2U_eU_b)
=
N\Tr(U_{\ell_i}),
\]
and the case $x\in B_i(e;L)$ is identical, using the first magic formula~\eqref{eq:magic-formula} and $\partial_{e,A}U_e^{-1}=-U_e^{-1}X_A$.  After multiplying by the undifferentiated traces, this gives $N W_L(U)$.

If the second derivative hits another occurrence $y\neq x$ of $e^{\pm1}$, the magic formulas give the rooted cut-and-join terms.  Let us spell out the sign using the splittings and mergers defined above.  If $x$ and $y$ have the same sign, the contraction is one of the positive splittings or positive mergers, and its contribution to the right-hand side of the pointwise identity below is $+W_{L_{x,y}}(U)$; this is exactly $-\sigma_{\mathrm{cj}}(x,y)W_{L_{x,y}}(U)$ because $\sigma_{\mathrm{cj}}(x,y)=-1$.  If $x$ and $y$ have opposite signs, the derivative of the inverse occurrence contributes one additional minus sign, the contraction is one of the negative splittings or negative mergers, and the contribution is $-W_{L_{x,y}}(U)$; this is $-\sigma_{\mathrm{cj}}(x,y)W_{L_{x,y}}(U)$ because $\sigma_{\mathrm{cj}}(x,y)=+1$.  Equivalently, the pointwise loop-only identity is
\begin{equation}\label{eq:rooted-loop-only-pointwise}
\sum_A\partial_{e,A}Y^x_{e,A}(U)
=
N W_L(U)
-
\sum_{\substack{y\neq x\\ y\in\cup_j C_j(e;L)}}
\sigma_{\mathrm{cj}}(x,y) W_{L_{x,y}}(U).
\end{equation}
By the definitions of $L_{x,y}$ and $\sigma_{\mathrm{cj}}(x,y)$, same-trace pairs give the positive and negative splittings, while distinct-trace pairs give the positive and negative $\U(N)$-mergers.  Integrating~\eqref{eq:rooted-loop-only-pointwise} against the undifferentiated plaquette characters gives
\begin{equation}\label{eq:rooted-loop-only-integrated}
\int\sum_A(\partial_{e,A}Y^x_{e,A})(U)
\prod_p\chi_{\alpha_p}(U_{\partial p}) dU
=
N\widehat W_{\Lambda,L}(\alpha)
-
(\mathscr S_x\widehat W_{\Lambda,\bullet})(L,\alpha).
\end{equation}

We now identify the mixed terms.  Since $Y^x_{e,A}=-D^x_{e,A}W_L$, the second term in~\eqref{eq:rooted-ibp-expanded} is the negative of the rooted mixed incidence term
\[
I^{\mathrm{root}}_{x,p}(L,\alpha)
:=
\int
\sum_A
D^x_{e,A}W_L(U) 
\partial_{e,A}\chi_{\alpha_p}(U_{\partial p})
\prod_{q\neq p}\chi_{\alpha_q}(U_{\partial q}) dU.
\]
For fixed $x$ and $p$, Proposition~\ref{prop:rooted-local-incidence-recoupling} applies to the single rooted occurrence $x$.  It gives the finite expansion
\begin{align*}
I^{\mathrm{root}}_{x,p}(L,\alpha)
&=
\sum_u\sum_{\beta\in\mathcal B_{x,p,u}(\alpha_p)}
c^N_{x,p,u}(\alpha_p,\beta)
\widehat W_{\Lambda,L_{x,p,u}}(\alpha^{p\to\beta})\\
&=(\mathscr B^{\mathrm{root}}_{x,p}\widehat W_{\Lambda,\bullet})(L,\alpha).
\end{align*}
The construction is local at the incidence $(e,p)$: it changes no plaquette label except $\alpha_p$, and the loop is changed only by the loop-plaquette surgery $L\mapsto L_{x,p,u}$.  The sum is finite because, for fixed $\alpha_p$, there are finitely many active tensor slots in the mixed-tensor realization of $\chi_{\alpha_p}$, and the auxiliary mixed tensor space has only finitely many irreducible summands.  If $\alpha_p=0$, the character $\chi_{\alpha_p}$ is constant, so $\partial_{e,A}\chi_{\alpha_p}(U_{\partial p})=0$, and the whole rooted plaquette-transfer term vanishes.

Substituting~\eqref{eq:rooted-loop-only-integrated} and the last identity into~\eqref{eq:rooted-ibp-expanded} gives
\[
\begin{aligned}
0={}&N\widehat W_{\Lambda,L}(\alpha)
-(\mathscr S_x\widehat W_{\Lambda,\bullet})(L,\alpha)\\
&-\sum_{p\ni e}
(\mathscr B^{\mathrm{root}}_{x,p}\widehat W_{\Lambda,\bullet})(L,\alpha).
\end{aligned}
\]
which is equivalent to~\eqref{eq:rooted-coeff-MLE}.
\end{proof}

By construction,
\[
\mathscr L_e=-N\sum_{i:e\in D_i(L)}m_i(e;L) \Id+
\sum_{x\in\cup_i C_i(e;L)}\mathscr S_x,
\qquad
\mathscr B_{e,p}=
\sum_{x\in\cup_i C_i(e;L)}\mathscr B^{\mathrm{root}}_{x,p},
\]
with the same conventions as in Subsection~\ref{subsec:coeff-MLE}. Summing the rooted coefficientwise master loop equation over all occurrences $x\in\cup_i C_i(e;L)$ gives
\[
N\sum_i m_i(e;L)\widehat W_{\Lambda,L}(\alpha)
=
\sum_x(\mathscr S_x\widehat W_{\Lambda,\bullet})(L,\alpha)
+
\sum_{p\ni e}\sum_x
(\mathscr B^{\mathrm{root}}_{x,p}\widehat W_{\Lambda,\bullet})(L,\alpha),
\]
which is equivalent, after moving the diagonal term to the left, to the unrooted coefficientwise equation recalled in~\eqref{eq:coeffwise-MLE}.

\subsection{The rooted heat-kernel equation}

We now insert the heat-kernel spectral weights into the rooted coefficientwise master equation, and put it in a form that will become exploitable later.

The plaquette term in Lemma~\ref{lem:rooted-coeff-normal-form} is moved onto the compactly supported plaquette decoration by adjunction with respect to the reference law $\widehat\Pbb_{T,N}$. This makes the resulting kernel explicit.

Let $X=(L,x,\zeta)\in\mathfrak X_\Lambda$ and let $p\ni e_x$.  Define
\[
A_p(\zeta)=
\begin{cases}
\{\zeta(p)\},& p\in S(\zeta),\\
\widehat{\U(N)},& p\notin S(\zeta).
\end{cases}
\]
For an input label $\alpha_p\in A_p(\zeta)$ and an output label $\beta_p$, put
\[
X_{u,\beta_p}:=(L_u,\mathfrak r(L_u),\zeta^{p\to\beta_p})
\]
whenever $\mathfrak r(L_u)\ne\dagger$.  If $\beta_p=0$, the update records the constraint that the output plaquette label is trivial.  If $\mathfrak r(L_u)=\dagger$, the output is terminal and is put into the source.

\begin{definition}\label{def:raw-transfer}
The \emph{rooted plaquette-transfer kernel} is the kernel on $\mathfrak X_\Lambda\times\mathfrak X_\Lambda$ defined by
\begin{equation}\label{eq:raw-transfer-kernel}
\begin{aligned}
\mathcal K_{\Lambda,T,N}(X,X')
:={}&\frac1N\sum_{p\ni e_x}
\sum_{\substack{\alpha_p\in A_p(\zeta)\\ \alpha_p\ne0}}
\sum_u\sum_{\beta_p:X'=X_{u,\beta_p}}\\
&\times
\frac{\widehat{Q_{T,N}^\HK}(\alpha_p)}
{\widehat{Q_{T,N}^\HK}(\beta_p)}
 b^{N,\mathrm{root}}_{x,p,u}
 (L,\alpha_p;L_u,\mathfrak r(L_u),\beta_p).
\end{aligned}
\end{equation}
\end{definition}

For fixed $\alpha_p$, the local recoupling support in $\beta_p$ is finite. If $p\notin S(\zeta)$, the sum over $\alpha_p$ is over all input plaquette labels. The summability needed for the exact finite-$N$ identity is proved in Theorem~\ref{thm:raw-rooted-HK-equation}; uniform large-$N$ coefficient estimates are obtained through the truncation argument of Subsection~\ref{subsec:coefficientwise-construction}.

The kernel $\mathcal K_{\Lambda,T,N}$ is obtained by taking the local loop-plaquette recoupling coefficient of the finite-$N$ rooted master loop equation and moving the plaquette label from the observable side to the spectral side by adjunction with respect to the heat-kernel reference law.  The ratio $\widehat{Q_{T,N}^\HK}(\alpha_p)/\widehat{Q_{T,N}^\HK}(\beta_p)$ is precisely the cost of this adjunction.

We also make the loop-only part and the terminal source explicit.  If a non-terminal channel-and-join surgery at the rooted occurrence $x$ and a partner occurrence $y$ produces the rooted state $X_{x,y}:=(L_{x,y},\mathfrak r(L_{x,y}),\zeta)$, set
\begin{equation}\label{eq:loop-only-operator-definition}
(\mathcal T_{\Lambda,T,N}F)(X):=\sum_{\substack{y\neq x: y\in\cup_i C_i(e_x;L)\\ L_{x,y}\ {\rm non\text{-}terminal}}}\sigma_{\rm cj}(x,y) \theta_N(L,L_{x,y}) F(X_{x,y}),
\end{equation}
where $\theta_N(L,L_{x,y}):=N^{\#L_{x,y}-\#L-1}.$ Thus same-trace splittings have coefficient of order $N^0$, while mergers
carry the factor $N^{-2}$.

The terminal source is the part of the same rooted equation whose output loop family has no non-trivial edge occurrence.  More explicitly, define
\begin{align*}
G^0_{\Lambda,T,N}(X)
:={}&
\sum_{\substack{y\ne x: y\in\cup_i C_i(e_x;L)\\ \mathfrak r(L_{x,y})=\dagger}}
\sigma_{\rm cj}(x,y)\theta_N(L,L_{x,y})
\Phi_{\Lambda,T,N}(L_{x,y};\zeta)\\
&+\frac1N\sum_{p\ni e_x}
\sum_{\substack{\alpha_p\in A_p(\zeta)\\ \alpha_p\ne0}}
\sum_{\substack{u,\beta_p:\ \mathfrak r(L_u)=\dagger}}
\frac{\widehat{Q_{T,N}^\HK}(\alpha_p)}
{\widehat{Q_{T,N}^\HK}(\beta_p)}\\
&\hspace{2cm}\times
b^{N,{\rm root}}_{x,p,u}(L,\alpha_p;L_u,\dagger,\beta_p)
\Phi_{\Lambda,T,N}(L_u;\zeta^{p\to\beta_p}).
\end{align*}
Thus $G^0_{\Lambda,T,N}$ is a scalar function on $\mathfrak X_\Lambda$. The first line is the terminal loop-only source and the second line is the terminal plaquette-transfer source; both use exactly the same normalized-output factors as in \eqref{eq:loop-only-operator-definition} and \eqref{eq:raw-transfer-kernel}.

\begin{theorem}\label{thm:raw-rooted-HK-equation}
For every rooted state $X=(L,x,\zeta)\in\mathfrak X_\Lambda$ in a finite volume,
\begin{equation}\label{eq:raw-rooted-HK-equation}
\begin{aligned}
\Psi_{\Lambda,T,N}(X)
={}&G^0_{\Lambda,T,N}(X)
+\mathcal T_{\Lambda,T,N}\Psi_{\Lambda,T,N}(X)\\
&+\sum_{X'\in\mathfrak X_\Lambda}
\mathcal K_{\Lambda,T,N}(X,X')
\Psi_{\Lambda,T,N}(X').
\end{aligned}
\end{equation}
where $\mathcal T_{\Lambda,T,N}$ is the finite loop-only operator coming from $S_x$.
\end{theorem}

\begin{proof}
We first record the finite-$N$ summability which justifies the manipulations below.  For fixed $N$, $T>0$, and finite $\Lambda$, the heat-kernel identity gives
\[
\sum_{\lambda\in\widehat{\U(N)}} d_\lambda^2e^{-\frac{T}{2N}C_2(\lambda)}=p_{T/N}(I_N)<\infty.
\]
Since the heat kernel is smooth, the same estimate remains true after inserting any fixed number of left-invariant derivatives, equivalently after multiplying the Fourier coefficient by a fixed polynomial in $C_2(\lambda)^{1/2}$.  Together with $|\chi_\lambda(g)|\le d_\lambda$ and $|\Tr(U_\ell)|\le N$, this implies absolute convergence of the decorated Wilson sums below and of the one-derivative sums obtained from Lemma~\ref{lem:rooted-coeff-normal-form}. Set $\mathcal Z_{\Lambda,T,N}:=\widehat\E_{\Lambda,T,N}\big[\mathcal O^\varnothing_{\Lambda,N}(\lambda)\big].$ For every loop family $M$ and every compactly supported plaquette constraint $\eta$, the Fourier-side representation gives
\[
\Phi_{\Lambda,T,N}(M;\eta)
=
\mathcal Z_{\Lambda,T,N}^{-1} Z_{T,N}^{-|P(\Lambda)|}
\sum_{\alpha:P(\Lambda)\to\widehat{\U(N)}}
1_\eta(\alpha) \kappa_{\Lambda,T,N}(\alpha)
N^{-\#M}\widehat W_{\Lambda,M}(\alpha).
\]
Indeed, the factor $Z_{T,N}^{-|P(\Lambda)|}$ comes from the reference law $\widehat\Pbb_{\Lambda,T,N}$, while the plaquette dimension factors in $\widehat Q^{\HK}_{T,N}$ are exactly those included in $\mathcal O^M_{\Lambda,N}$.

Fix $X=(L,x,\zeta)$.  Multiply~\eqref{eq:rooted-coeff-MLE} by
\[
Z_{T,N}^{-|P(\Lambda)|}\mathcal Z_{\Lambda,T,N}^{-1} 
1_\zeta(\alpha) \kappa_{\Lambda,T,N}(\alpha) N^{-\#L-1}
\]
and sum over all plaquette decorations $\alpha$.  The diagonal term gives
\[
\mathcal Z_{\Lambda,T,N}^{-1} Z_{T,N}^{-|P(\Lambda)|}
\sum_\alpha 1_\zeta(\alpha)\kappa_{\Lambda,T,N}(\alpha)
N^{-\#L}\widehat W_{\Lambda,L}(\alpha)
=
\Psi_{\Lambda,T,N}(X).
\]

We next identify the loop-only terms.  If the rooted cut-and-join output associated with $y$ is non-terminal and equal to $X_{x,y}=(L_{x,y},\mathfrak r(L_{x,y}),\zeta)$, then its contribution is
\[
\begin{aligned}
&\mathcal Z_{\Lambda,T,N}^{-1} Z_{T,N}^{-|P(\Lambda)|}
\sum_\alpha
1_\zeta(\alpha)\kappa_{\Lambda,T,N}(\alpha)
N^{-\#L-1}\sigma_{\rm cj}(x,y)
\widehat W_{\Lambda,L_{x,y}}(\alpha)
\\
&\qquad=
\sigma_{\rm cj}(x,y)N^{\#L_{x,y}-\#L-1}
\Psi_{\Lambda,T,N}(L_{x,y},\mathfrak r(L_{x,y}),\zeta).
\end{aligned}
\]
Summing over all non-terminal partners $y$ gives $(\mathcal T_{\Lambda,T,N}\Psi_{\Lambda,T,N})(X)$.  The terminal partners give the loop-only part of $G^0_{\Lambda,T,N}(X)$.

We now treat the plaquette-transfer terms.  Fix $p\ni e_x$, an input label $\alpha_p\in A_p(\zeta)$ with $\alpha_p\ne0$, a local slot $u$, and an output label $\beta_p$.  Put $\gamma=\alpha^{p\to\beta_p}.$ The labels away from $p$ are unchanged and $\gamma_p=\beta_p$.  By the definition of the updated constraint, $1_\zeta(\alpha)=1_{\zeta^{p\to\beta_p}}(\gamma)$ on this summand.
The product heat-kernel coefficient factorizes as
\[
\kappa_{\Lambda,T,N}(\alpha)=\kappa_{\Lambda,T,N}(\gamma)\frac{\widehat{Q_{T,N}^\HK}(\alpha_p)}{\widehat{Q_{T,N}^\HK}(\beta_p)}.
\]
Therefore, if the output is non-terminal, its contribution is
\[
N^{\#L_u-\#L-1}\frac{\widehat{Q_{T,N}^\HK}(\alpha_p)}{\widehat{Q_{T,N}^\HK}(\beta_p)} b^{N,{\rm root}}_{x,p,u}(L,\alpha_p;L_u,\mathfrak r(L_u),\beta_p)\Psi_{\Lambda,T,N}(L_u,\mathfrak r(L_u),\zeta^{p\to\beta_p}).
\]
For the rooted loop-plaquette surgery considered here one has $\#L_u=\#L$; hence this is precisely the contribution encoded by the prefactor $1/N$ in \eqref{eq:raw-transfer-kernel}.  Summing over $p$, $\alpha_p$, $u$ and $\beta_p$ gives
\[
\sum_{X'\in\mathfrak X_\Lambda}\mathcal K_{\Lambda,T,N}(X,X')\Psi_{\Lambda,T,N}(X').
\]
The terminal plaquette-transfer outputs give the second part of $G^0_{\Lambda,T,N}(X)$, with the same normalized-output factor.

Combining the diagonal term, the loop-only terms and the plaquette-transfer terms gives \eqref{eq:raw-rooted-HK-equation}.
\end{proof}

The equation \eqref{eq:raw-rooted-HK-equation} is closed only after the terminal outputs have themselves been controlled. We record this control separately, because it will be used as an input in the rooted trajectory expansion.

\begin{lemma}[Linked-cluster cancellation for a local constraint]\label{lem:constraint-component-cancellation}
Let $\zeta$ be a compactly supported plaquette constraint. In the strong-coupling local-channel expansion of
\[
\frac{\widehat\E_{\Lambda,T,N}[\mathbf 1_\zeta(\lambda)\mathcal O^\varnothing_{\Lambda,N}(\lambda)]}
{\widehat\E_{\Lambda,T,N}[\mathcal O^\varnothing_{\Lambda,N}(\lambda)]},
\]
the normalized ratio has an absolutely convergent linked-cluster expansion in which every surviving connected cluster intersects $S(\zeta)$. This remains true when the prescribed label at a plaquette of $S(\zeta)$ is the trivial representation.
\end{lemma}

\begin{proof}
Decompose a scalar local-channel configuration into connected active components in the dual incidence graph. Its non-spectral weight factors over these components, and the reference law~\eqref{eq:product-ref-laws} factors over their plaquette labels. For $T$ in the strong-coupling regime, Theorem~\ref{thm:stable-hk-block} together with Proposition~\ref{prop:nonstable-polymer-tail} gives absolute convergence of the resulting polymer expansion. The standard finite-volume linked-cluster formula may therefore be applied. The indicator $\mathbf 1_\zeta$ changes only local weights of clusters meeting $S(\zeta)$; a connected cluster disjoint from $S(\zeta)$ has exactly its vacuum weight. In the quotient by the vacuum partition function, the connected vacuum clusters disjoint from $S(\zeta)$ cancel, leaving the linked clusters intersecting $S(\zeta)$. The argument depends only on the support of the constraint, not on the prescribed value, so it also applies when that value is the trivial representation.
\end{proof}

\begin{proposition}\label{prop:decorated-terminal-sources}
Let $M$ be a loop family all of whose components have no non-trivial edge occurrence, and let $\zeta$ be a compactly supported plaquette constraint whose prescribed labels are stable and have finite $\mathcal E$-energy.  For $T$ large enough, uniformly in every finite volume $\Lambda$ containing $S(\zeta)$,
\begin{equation}\label{eq:decorated-terminal-expansion}
\Phi_{\Lambda,T,N}(M;\zeta)
=
\sum_{r=0}^R N^{-r} V^{(r)}_{\Lambda,T}(M;\zeta)
+
O_{R,M,\zeta,T}(N^{-R-1})
\end{equation}
for every $R\ge0$.  The coefficients are exponentially local in $\operatorname{dist}(S(\zeta),\partial\Lambda)$ and have infinite-volume limits.  Moreover the retained stable finite-energy terminal source terms appearing in $G^0_{\Lambda,T,N}(X)$ admit controlled expansions of the same kind, while the stable high-energy and large-mass or unstable terminal-source parts satisfy the same exceptional-source bounds as the corresponding plaquette-transfer kernels.
\end{proposition}

\begin{proof}
Since all components of $M$ are trivial after gauge fixing, the normalized Wilson factor is equal to $1$. Thus
\[
\Phi_{\Lambda,T,N}(M;\zeta)
=
\frac{\widehat\E_{\Lambda,T,N}[\mathbf 1_\zeta(\lambda) \mathcal O^\varnothing_{\Lambda,N}(\lambda)]}
{\widehat\E_{\Lambda,T,N}[\mathcal O^\varnothing_{\Lambda,N}(\lambda)]}.
\]
Insert the local-channel expansion of Section~\ref{sec:finiteN}. By Lemma~\ref{lem:constraint-component-cancellation}, after division by the vacuum denominator the linked-cluster expansion contains only connected clusters meeting $S(\zeta)$. Thus every surviving component is pinned to the finite set $S(\zeta)$, irrespective of whether the prescribed representation at a point of the constraint support is trivial or non-trivial. Repeating the proof of Theorem~\ref{thm:pinned-block} and Theorem~\ref{thm:stable-hk-block} with this finite constraint support as the pinning set gives an exponentially summable polymer expansion over connected supports meeting $S(\zeta)$, with the same determinant-shift tail estimates and the same controlled stable Schur--Weyl expansions. This proves~\eqref{eq:decorated-terminal-expansion} and the exponential locality of its coefficients.

A terminal term in $G^0_{\Lambda,T,N}(X)$ is a finite loop-only coefficient, or one rooted plaquette-transfer coefficient with its heat-kernel adjunction factor, multiplied by one of the decorated terminal quantities just considered.  On retained stable finite-energy data all these local factors have controlled expansions by the fixed-degree estimates of Subsections~\ref{subsec:fixed-degree-coordinates} and~\ref{subsec:orthonormal-channel-recoupling}.  The high-energy terminal part is bounded by the same shell estimate as in Lemma~\ref{lem:paid-shell-row}, and the large-mass or unstable terminal part by the tilted-tail estimate used in Lemma~\ref{lem:paid-shell-row}.  Hence terminal sources can be included in the exceptional-sector estimates used below, with the same bounds as the transfer kernels.
\end{proof}

\subsection{The plaquette-transfer estimate}\label{subsec:plaquette-transfer-estimate}

The following estimate is the analytic core of the section. We first record the elementary distortion estimate for a single local transfer.

\begin{lemma}\label{lem:local-transfer-distortion}
There is a constant $C_d<\infty$ with the following property.  Consider one non-zero local term of the rooted plaquette transfer touching the rooted occurrence $x$ and a plaquette $p$.  Let $\alpha_p\ne0$ be the input label and let $\beta_p$ be an output label in the local support.  Write stable coordinates $\alpha_p=q_\alpha \mathbf 1_N+\overline\alpha_p,$ $\beta_p=q_\beta \mathbf 1_N+\overline\beta_p,$ whenever the two labels are in the stable range.  Then, in this range,
\begin{equation}\label{eq:local-transfer-stable-distortion}
m(\overline\beta_p)\le C_d(1+m(\overline\alpha_p)),
\qquad
|q_\beta-q_\alpha|\le C_d.
\end{equation}
Consequently, with $r=\mathcal E(\alpha_p)\ge1$,
\begin{equation}\label{eq:local-transfer-distortion}
\ell(L_u)-\ell(L)\le C_d r,
\qquad
\mathcal E(\beta_p)\le C_d r.
\end{equation}
For the corresponding output state $X'=X_{u,\beta_p}$,
\begin{equation}\label{eq:local-state-distortion}
|X'|_*-|X|_*\le C_d r,
\qquad
\mathcal E(X')-\mathcal E(X)\le C_d r.
\end{equation}
\end{lemma}

\begin{proof}
The rooted transfer is local.  It contracts the rooted fundamental, or dual fundamental, Wilson strand with one tensor slot of the plaquette representation at $p$.  In the local-channel realization, a reduced label of projective mass $m(\alpha_p)$ is realized inside a mixed tensor space of degree $m(\alpha_p)$.  The local operator appearing in the rooted transfer is obtained by tensoring with, contracting with, or removing only a bounded number of copies of $\mathbb C^N$ and $(\mathbb C^N)^*$, the bound depending only on the lattice valence.  By the stable Pieri rules~\eqref{eq:Pieri-stable} and \eqref{eq:Pieri-stable2} and the operator realization of Lemma~\ref{lem:local-operator-complexity}, every stable output label therefore belongs to a bounded finite union of mixed tensor products obtained from $\alpha_p$ by a bounded number of one-box moves.  This gives the projective estimate in \eqref{eq:local-transfer-stable-distortion}.

For the determinant coordinate, recall that the total central character of
$q\mathbf 1_N+[\lambda^+,\lambda^-]_N$ is $Nq+|\lambda^+|-|\lambda^-|$.  Tensoring with $\mathbb C^N$ or $(\mathbb C^N)^*$ changes this total central character by $\pm1$.  Since the rooted local transfer uses only a bounded number of such fundamental or dual fundamental legs, the determinant shift can change only by a bounded amount.  This proves $|q_\beta-q_\alpha|\le C_d$, after possibly increasing $C_d$ to account for the stable-coordinate normalization.

It follows that
\[
\mathcal E(\beta_p)
=m(\overline\beta_p)+q_\beta^2
\le C_d(1+m(\overline\alpha_p))+C_d(1+q_\alpha^2)
\le C_d r,
\]
since $r=\mathcal E(\alpha_p)\ge1$.  The loop surgery creates only the strands produced by this same bounded local mixed-tensor operation, and hence its word-length increase is bounded by $C_d(1+m(\alpha_p))\le C_d r$.  The canonical decoration update changes the decoration support by at most one plaquette, which is absorbed because $r\ge1$.  This proves the state estimates.
\end{proof}

Combining this distortion estimate with the pinned block estimate of Section~\ref{sec:pinned-haar-projection} proves the following adjoint total-variation estimate for one fixed non-trivial input plaquette label in the stable range.

\begin{lemma}\label{lem:paid-local-adjoint-row}
There are constants $C,c>0$, depending only on $d$, such that the following holds uniformly in the finite volume and in $N$.  Fix a rooted state $X=(L,x,\zeta)$, a plaquette $p\ni e_x$, and a non-trivial input label $\alpha_p\in A_p(\zeta)$ in the stable small-mass range $m(\alpha_p)\le \varepsilon_d N,$ with $0<\varepsilon_d<(4d)^{-1}$.  Put $r=\mathcal E(\alpha_p)$. Let $\mathcal B$ be a finite set of stable local output pairs $\mathcal B\subset \{(u,\beta_p):\beta_p\in\mathcal B_{x,p,u}(\alpha_p),\ \beta_p\text{ stable}\}.$ For every family of test coefficients $F_{u,\beta_p}(\gamma)$, supported on this finite output set and satisfying $|F_{u,\beta_p}|\le1$, one has
\begin{align}\label{eq:paid-local-adjoint-row}
&\Biggl|
\sum_{(u,\beta_p)\in\mathcal B}
\sum_{\substack{\gamma:\ \gamma_p=\beta_p}}
1_{\zeta^{p\to\beta_p}}(\gamma)
\kappa_{\Lambda,T,N}(\gamma)
\frac{\widehat Q^\HK_{T,N}(\alpha_p)}
{\widehat Q^\HK_{T,N}(\beta_p)}
\notag\\[-1mm]
&\hspace{2.6cm}\times
b^{N,\mathrm{root}}_{x,p,u}
(L,\alpha_p;L_u,\mathfrak r(L_u),\beta_p)
F_{u,\beta_p}(\gamma)
\Biggr|
\notag\\
&\qquad\le
C^r e^{-cTr}
\sum_{\substack{\alpha:\ \alpha_p\text{ fixed}}}
1_\zeta(\alpha)\kappa_{\Lambda,T,N}(\alpha).
\end{align}
\end{lemma}

\begin{proof}
Insert first a finite energy cutoff so that all sums are finite.  We write
$\mathcal B_A(\alpha_p)$ for the finite set of stable output pairs
$(u,\beta_p)$ with $\beta_p\in\mathcal B_{x,p,u}(\alpha_p)$ and $\mathcal E(\beta_p)\le A$.  Before the adjoint change of variables, the local integration-by-parts term with input label $\alpha_p$ is a connected pinned local block containing only the rooted Wilson collar at $x$, the active plaquette $p$, and the bounded local tensor data created by the rooted surgery. All other Wilson trace components and all parts of the rooted loop away from the differentiated occurrence are spectator factors: they are carried by the outgoing observable $W_{L_u}$ and do not enter the scalar local recoupling coefficient.

Let $B^{\mathrm{top}}_{x,p,u}(L,\alpha_p;L_u,x_u,\beta_p)$ denote the corresponding non-spectral rooted topological coefficient, that is, the local-channel coefficient before the heat-kernel exponential attached to the input label is inserted. By Lemma~\ref{lem:local-transfer-distortion}, every stable output label in the local support has projective mass at most $C_d(1+r)$, and the local rooted support which has to be estimated has size $O_d(1+r)$. The Wilson-collar contribution has bounded mixed degree independent of the ambient loop family. Applying the pinned block estimate of Theorem~\ref{thm:pinned-block} to this local rooted block, together with the local counting already included in its absolute block-observable bound, gives
\begin{equation}\label{eq:paid-local-nonspectral-bound}
\sum_{(u,\beta_p)\in\mathcal B_A(\alpha_p)}
\left|B^{\mathrm{top}}_{x,p,u}(L,\alpha_p;L_u,\mathfrak r(L_u),\beta_p)\right|
\le C^r .
\end{equation}
Here the constant depends only on $d$; in particular it is independent of the ambient loop family $L$, of the finite volume, and of $N$. The dimension factor $d_{\alpha_p}$ appearing in
$\widehat Q^\HK_{T,N}(\alpha_p)$ is one of the plaquette dimensions cancelled
by the Haar projection mechanism of Section~\ref{sec:pinned-haar-projection}; no
positive power of $N$ is produced by this block.  The remaining heat-kernel
factor gives
\[
\exp\left\{-\frac{T}{2N}C_2(\alpha_p)\right\}
\le e^{-cTr}
\]
by the heat-kernel domination bound~\eqref{eq:spectral-decay}, after
decreasing $c$.  Combining this decay with
\eqref{eq:paid-local-nonspectral-bound} shows that the direct heat-kernel
weighted local bilinear form has total absolute mass at most
$C^r e^{-cTr}$, after changing $C$.

Let us spell out the adjunction step on the finite cutoff.  For any bounded family of test coefficients $F_{u,\beta_p}(\gamma)$, supported on the cutoff and with $|F_{u,\beta_p}|\le1$, the change of variables $\gamma=\alpha^{p\to\beta_p}$ gives the exact identity
\begin{align}\label{eq:finite-cutoff-adjunction}
&\sum_{(u,\beta_p)\in\mathcal B_A(\alpha_p)}
\sum_{\substack{\gamma:\ \gamma_p=\beta_p}}
1_{\zeta^{p\to\beta_p}}(\gamma)
\kappa_{\Lambda,T,N}(\gamma)
\frac{\widehat Q^\HK_{T,N}(\alpha_p)}
{\widehat Q^\HK_{T,N}(\beta_p)}
\notag\\[-1mm]
&\hspace{2cm}\times
b^{N,\mathrm{root}}_{x,p,u}
(L,\alpha_p;L_u,\mathfrak r(L_u),\beta_p)
F_{u,\beta_p}(\gamma)
\notag\\
&\qquad=
\sum_{(u,\beta_p)\in\mathcal B_A(\alpha_p)}
\sum_{\substack{\alpha:\ \alpha_p\text{ fixed}}}
1_\zeta(\alpha)\kappa_{\Lambda,T,N}(\alpha)
\notag\\[-1mm]
&\hspace{2cm}\times
b^{N,\mathrm{root}}_{x,p,u}
(L,\alpha_p;L_u,\mathfrak r(L_u),\beta_p)
F_{u,\beta_p}(\alpha^{p\to\beta_p}).
\end{align}
Indeed, all plaquette weights except the one at $p$ are unchanged, and
\[
\kappa_{\Lambda,T,N}(\alpha)
=
\kappa_{\Lambda,T,N}(\gamma)
\frac{\widehat Q^\HK_{T,N}(\alpha_p)}{\widehat Q^\HK_{T,N}(\beta_p)}.
\]
Taking the supremum over phase choices $F$ identifies the total variation of the finite-cutoff adjoint row with the direct local bilinear form just estimated.  Dividing by the reference mass of the input slice gives \eqref{eq:paid-local-adjoint-row}.  The cutoff is removed by monotone convergence of the absolute stable total-variation sums.
\end{proof}

\begin{remark}
For $N$ large compared with $r$, the stable-output restriction of Lemma~\ref{lem:paid-local-adjoint-row} is automatic by Lemma~\ref{lem:local-transfer-distortion}; the finitely many remaining small values of $N$ are absorbed into the constants. The cutoff is removed by monotone convergence of the absolute stable total-variation sums.
\end{remark}

The previous lemma is deliberately a retained-stable-sector statement. However, the full raw kernel also contains labels outside this range. They are treated collectively under the tilted reference law after the exponential non-spectral local cost has been absorbed, and their contribution is exponentially small in $N$, uniformly in the finite volume.

\begin{lemma}\label{lem:paid-shell-row}
There exist constants $C,c>0$, depending only on $d$, such that the following holds.  Let $X=(L,x,\zeta)$ be a rooted state whose compactly supported plaquette decoration is stable and has finite energy, and let $p\ni e_x$.  For every integer $r\ge1$, set
\[
\mathcal A^{\mathrm{st}}_r(p,\zeta)
=\{\alpha_p\in A_p(\zeta):\alpha_p\ne0,
\ m(\alpha_p)\le\varepsilon_dN,
\ r\le \mathcal E(\alpha_p)<r+1\}.
\]
Then, after restricting outputs to any finite stable output set and testing against $|F|\le1$, the shell contribution is bounded by
\[
C^r e^{-cTr}\sum_{\alpha:\alpha_p\in\mathcal A^{\mathrm{st}}_r(p,\zeta)}1_\zeta(\alpha)\kappa_{\Lambda,T,N}(\alpha).
\]
Moreover the two complementary pieces are controlled as follows.
\begin{enumerate}
\item For every stable energy cutoff $A<\infty$, the stable high-energy adjoint tail is bounded by
\begin{equation}\label{eq:stable-energy-tail-paid-row}
\varepsilon_A
\le
C\sum_{r>A} C^r e^{-cTr},
\end{equation}
which tends to zero as $A\to\infty$, uniformly in $N$ and in the finite volume, once $T$ is large.
\item Fix $A<\infty$ and an auxiliary finite representation truncation of the raw label sum.  On this truncation, the part outside the stable small-mass sector but with retained energy $\mathcal E\le A$ is bounded before coefficient extraction by
\begin{equation}\label{eq:paid-unstable-complement}
C_{A,X,T}e^{-\gamma(T)N}
\end{equation}
for some $\gamma(T)>0$, uniformly in $\Lambda$ and in the auxiliary finite truncation.  In particular it is $O_{A,R,X,T}(N^{-R-1})$ for every prescribed asymptotic order $R$.
\end{enumerate}
\end{lemma}

\begin{proof}
For the retained stable shell, apply Lemma~\ref{lem:paid-local-adjoint-row} to each input label and sum the corresponding input slices.  The number of stable labels with $r\le \mathcal E(\alpha_p)<r+1$ is bounded by $C^r$: the reduced part is controlled by the partition-number bound used in Lemma~\ref{lem:nonstable-weighted-tail}, and the determinant coordinate has at most polynomially many choices with $q^2\le r+1$, which is absorbed into $C^r$.  After changing $C,c$, this proves the bound on the stable adjoint total variation. Summing these shell bounds over $r>A$ gives \eqref{eq:stable-energy-tail-paid-row}.

It remains to justify the collective bound \eqref{eq:paid-unstable-complement}.  We use the following consequence of Lemma~\ref{lem:central-cutoff}, Lemma~\ref{lem:nonstable-weighted-tail} and Proposition~\ref{prop:nonstable-polymer-tail}: after paying any fixed exponential non-spectral local cost $A_0^{M}$, where $M$ denotes the sum of the projective masses of the finitely many labels touched by the local transfer, there are constants $C_{A_0,T}<\infty$ and $\gamma(A_0,T)>0$ such that, for the finite set $S$ of labels touched by the local rooted transfer,
\begin{equation}\label{eq:tilted-unstable-tail-used-in-section5}
\widehat\E_{S,A_0,T,N}
\left[
\mathbf 1_{\mathrm{outside\ stable\ small\ mass}}
\right]
\le C_{A_0,T}^{|S|} e^{-\gamma(A_0,T)N},
\end{equation}
provided $T$ is large enough.  The estimate is uniform in the auxiliary finite truncation and in the ambient finite volume.  The stable part with $M>\varepsilon_dN$ is the cutoff estimate of Lemma~\ref{lem:central-cutoff}, while the remaining non-retained part is controlled by the tilted heat-kernel tail of Lemma~\ref{lem:nonstable-weighted-tail}; Lemma~\ref{lem:product-tilted-mass-tail} gives the product tail, and Proposition~\ref{prop:nonstable-polymer-tail} incorporates the additional exponential non-spectral local weight.

For the rooted transfer under consideration, Lemma~\ref{lem:local-transfer-distortion} and Lemma~\ref{lem:local-operator-complexity}, followed by the canonical Pieri-path summation of Lemma~\ref{lem:full-pieri-path} (and by Lemma~\ref{lem:oscillating-path} on the retained stable sector), show that its total absolute local transfer weight is bounded by an exponential factor $A_0^{M}$ for this same local total projective mass $M$, with $A_0$ depending only on $d$ and on the fixed loop family.  The non-spectral pinned block is bounded by Theorem~\ref{thm:pinned-block}.  Applying \eqref{eq:tilted-unstable-tail-used-in-section5} gives \eqref{eq:paid-unstable-complement}.  Since $e^{-\gamma(T)N}=O_{R,T}(N^{-R-1})$ for every fixed $R$, this complement contributes only to the $O_{R,T}(N^{-R-1})$ remainder at every prescribed order.
\end{proof}

\begin{proposition}\label{prop:raw-plaquette-transfer}
There exist constants $C,c,C_0<\infty$, depending only on $d$, such that the following holds.  For every $a,\eta>0$, every finite volume $\Lambda$, every $N$, and every rooted state $X$ whose compactly supported plaquette decoration is stable and has finite energy, the stable part of the raw transfer satisfies the adjoint weighted total-variation bound
\begin{equation}\label{eq:raw-paid-row-sum}
\Theta_T(a,\eta)
:=
C_0\sum_{r\ge1}\exp\{r(\log C-cT+aC+\eta C)\}.
\end{equation}
In particular, $\Theta_T(a,\eta)<1$ for $T$ large enough.
\end{proposition}

\begin{proof}
A rooted plaquette transfer first chooses a plaquette incident to the rooted edge; there are at most $2(d-1)$ such choices.  By Lemma~\ref{lem:rooted-coeff-normal-form}, every nonzero plaquette-transfer term differentiates a non-trivial input plaquette character.  Hence the input label $\alpha_p$ is non-trivial and has $\mathcal E(\alpha_p)\ge1$.

Fix an energy shell $r\le \mathcal E(\alpha_p)<r+1$ in the stable small-mass range.  Lemma~\ref{lem:paid-shell-row} gives the unweighted adjoint shell bound $C^re^{-cTr}$.  Lemma~\ref{lem:local-transfer-distortion} then implies that every stable output state produced from this shell satisfies
\[
|X'|_*-|X|_*\le Cr,
\qquad
\mathcal E(X')-\mathcal E(X)\le Cr.
\]
Multiplying the test family by the corresponding state weight therefore changes the shell estimate by at most the factor $e^{aCr+\eta Cr}$.  Thus the weighted adjoint shell contributes at most
\[
C^r e^{-cTr+aCr+\eta Cr}.
\]
Summing over $r\ge1$ and over the bounded set of incident plaquettes gives \eqref{eq:raw-paid-row-sum}.
\end{proof}

\subsection{Order-truncated rooted trajectories}

The rooted heat-kernel equation~\eqref{eq:raw-rooted-HK-equation} combines loop-only operators and plaquette-transfer operators. The plaquette-transfer part is small at strong coupling, in the weighted adjoint sense of Proposition~\ref{prop:raw-plaquette-transfer}. The loop-only operator has a different nature: it does not change plaquette decorations and has exact monomial coefficients in \(N^{-1}\), but it does not admit a uniformly bounded global resolvent in the exponential loop-length norm. We therefore organize the expansion into finite loop-only segments separated by plaquette transfers. A rooted trajectory is an alternating sequence
\[
\begin{aligned}
\text{loop-only}
&\longrightarrow \text{plaquette-transfer}
\longrightarrow \text{loop-only}\\
&\longrightarrow\cdots\longrightarrow \text{terminal source}.
\end{aligned}
\]
The loop-only segments are controlled by a finite-order entropy estimate, whereas the plaquette transfers are controlled by the strong-coupling weighted total-variation bound. Since all local coefficients have non-negative asymptotic order after the normalizations used above, we may fix an order \(R\) and keep only trajectories whose total order is at most \(R\). The complementary stable high-energy and large-mass or unstable sectors are excluded from the retained histories and are controlled by the truncation argument of Subsection~\ref{subsec:coefficientwise-construction}.

We begin by isolating the loop-only dynamics. A loop-only elementary move is one of the non-terminal channel-and-join terms in $\mathcal T_{\Lambda,T,N}$.  It uses the rooted occurrence and another occurrence of the same unoriented edge. The two differentiated insertions at these occurrences are contracted by the magic formulas~\eqref{eq:magic-formula}, and the adjacent strands are reconnected according to the four surgeries displayed above.  In the positive same-orientation cases the resulting words may still contain occurrences of the same oriented edge.  If the output loop family $M$ satisfies $\mathfrak r(M)\ne\dagger$, we use the root $\mathfrak r(M)$; otherwise the move is terminal.  Loop-only moves never change the compactly supported plaquette decoration.

For fixed loop words, the loop-only expansion is exact: every loop-only coefficient is one of the monomials $N^{\#L_{x,y}-\#L-1}$.  Thus
\begin{equation}\label{eq:T-Laurent}
\mathcal T_{\Lambda,T,N}=\sum_{j\ge0}N^{-j}\mathcal T^{[j]}_{\Lambda,T},
\end{equation}
where only finitely many terms are present for the fixed input words. A same-trace splitting has $\#L_{x,y}=\#L+1$ and therefore coefficient of order $N^0$, whereas a different-trace merger has $\#L_{x,y}=\#L-1$ and therefore coefficient of order $N^{-2}$. Thus the order-zero loop-only moves are precisely the same-trace splittings, while mergers have order $2$.

The finite-order entropy estimate below is based on the following elementary picture.  Order-zero loop-only histories are encoded by rooted non-crossing forests on the occurrences of the initial loop words: every same-trace split chooses the root occurrence and a non-crossing partner.  Positive-order loop-only moves, namely mergers, are exceptional.  At total asymptotic order at most $R$, there are at most $R/2$ such exceptional events.  Thus the order-zero part has merely exponential entropy in the initial length, and the positive-order part adds only a polynomial factor depending on $R$.  In the estimate below this entropy is absorbed by allowing a slightly larger exponential weight on the input loop length than on the endpoint loop length.

\begin{definition}
Let $X=(L,x,\zeta)$ be a rooted state.  A loop-only history starting from $X$ is a finite sequence
\[
\mathfrak l=(X=X_0,s_1,X_1,\ldots,s_m,X_m)
\]
where each $s_i$ is one non-terminal loop-only elementary move.  The endpoint is $Y=X_m$.  The history is allowed to have length $m=0$.  If after the endpoint one chooses a terminal loop-only output, that terminal coefficient is included in the terminal source rather than in $\mathfrak l$.  If $s_i$ contributes through the coefficient of order $j_i$ in \eqref{eq:T-Laurent}, the asymptotic order of $\mathfrak l$ is $\operatorname{ord}(\mathfrak l)=j_1+\cdots+j_m$, and its coefficient is the product of the corresponding local coefficients.
\end{definition}

\begin{lemma}\label{lem:loop-only-entropy}
There exists $\kappa<\infty$ with the following property.  For every $R\ge0$ and every $\gamma\ge0$ there exists $C_R(\gamma)<\infty$ such that, for every rooted state $X=(L,x,\zeta)$, if $\mathcal L_{\le R}(X)$ denotes the set of loop-only histories starting from $X$ whose total asymptotic order is at most $R$, then
\begin{equation}\label{eq:loop-only-entropy}
\sum_{j=0}^R\sum_{\substack{\mathfrak l\in\mathcal L_{\le R}(X)\\
\operatorname{ord}(\mathfrak l)=j}}
|w^{[j]}_{\mathrm{loop}}(\mathfrak l)| e^{\gamma\ell(L_{Y(\mathfrak l)})}
\le
C_R(\gamma)e^{(\gamma+\kappa)\ell(L)}.
\end{equation}
The constants are independent of $N$, $\Lambda$, the compactly supported plaquette decoration $\zeta$, and the ambient finite volume.
\end{lemma}

\begin{proof}
Loop-only moves do not change the compactly supported plaquette decoration, and they do not increase the total loop length. The estimate does not rely on monotonicity of the number of occurrences of the rooted edge, which may fail in the positive same-orientation cases. The finiteness and entropy bound instead come from the same-trace/merger degree bookkeeping described above.

We first count an order-zero segment.  Such a segment uses only same-trace splittings.  At every order-zero step, the root occurrence and its partner lie on the same trace component, and the normalized magic formula splits this component into two trace components without crossing the cyclic order.  Record the two occurrences paired by the split, viewed as two marked points on the cyclic word on which the split was performed.  Pulling these pairs back to the word on which the segment starts gives a collection of non-crossing chords: nested splittings give nested chords and disjoint splittings give disjoint chords.  The deterministic outgoing-root rule records which child interval continues the rooted evolution.  Therefore the collection of rooted non-crossing chords determines the order-zero segment up to bounded local multiplicity.  Consequently the number of order-zero segments starting from total loop length $n$ is bounded by $C^n$; for instance, after cutting the cyclic words at the relevant roots and forgetting some structure, such segments inject with bounded multiplicity into plane rooted forests, which are Catalan-bounded.  The order-zero local coefficients are universal normalized trace-contraction coefficients and have absolute value bounded by a fixed constant, so the total absolute weight of all order-zero segments starting from length $n$ is at most $C^n$, after changing $C$.

Now allow total asymptotic order at most $R$.  A loop-only move which is not an order-zero same-trace split is exceptional: it is a trace-merging term.  Each exceptional move has positive asymptotic order, hence at most $R$ exceptional moves can appear in a history of total order at most $R$.  During the whole history all loop words have length at most $n=\ell(L)$.  Thus the exceptional moves contribute at most a polynomial factor $n^{q_R}$, together with a coefficient bound depending only on $R$.  Between successive exceptional moves there are order-zero segments, and each of these segments is bounded by the preceding non-crossing-forest estimate on subwords of total length at most $n$.  Equivalently, the whole history can be encoded by one rooted non-crossing forest, together with at most $R$ exceptional marked trace-merging operations and their attachment data.  The exceptional data have only polynomially many choices in $n$, with exponent depending on $R$, while the non-crossing forest has Catalan-type exponential growth with a universal base.  Therefore the total absolute loop-only weight of all histories of order at most $R$ is bounded by $C_R n^{q_R}B^n$ for constants $C_R,q_R<\infty$ depending on $R$ and a universal constant $B<\infty$.

Finally, every endpoint $Y(\mathfrak l)$ satisfies $\ell(L_{Y(\mathfrak l)})\le n$.  Hence
\[
\sum_{j=0}^R\sum_{\substack{\mathfrak l\in\mathcal L_{\le R}(X)\\
\operatorname{ord}(\mathfrak l)=j}}
|w^{[j]}_{\mathrm{loop}}(\mathfrak l)| e^{\gamma\ell(L_{Y(\mathfrak l)})}
\le C_R n^{q_R}B^n e^{\gamma n}.
\]
Choose $\kappa>\log B$.  The polynomial factor $n^{q_R}$ is absorbed into $e^{(\kappa-\log B)n}$, after increasing the constant to $C_R(\gamma)$.  Thus
\[
C_R n^{q_R}B^n e^{\gamma n}
\le C_R(\gamma)e^{(\gamma+\kappa)n},
\]
which proves \eqref{eq:loop-only-entropy}.
\end{proof}

\begin{definition}
Fix an initial rooted state $X_0$.  An order-truncated rooted history of order $r$ starting from $X_0$ is a finite alternating sequence
\[
\mathfrak h=(\mathfrak l_0,\tau_1,\mathfrak l_1,\ldots,\tau_m,\mathfrak l_m,\sigma),
\]
where each $\mathfrak l_i$ is a loop-only history, each $\tau_i$ is one plaquette-transfer term of the stable raw kernel, and $\sigma$ is one coefficient in the controlled expansion of a retained terminal source term, as supplied by Proposition~\ref{prop:decorated-terminal-sources}.  The total asymptotic order of all local and terminal-source coefficients is $r$.  Its coefficient is denoted by $w^{[r]}_{\Lambda,T}(\mathfrak h)$.  The support of $\mathfrak h$ is the union of all loops, decoration supports and plaquettes touched by the local terms appearing in the history.  The energy of a transfer $\tau_i$ is $\mathcal E(\alpha_i)$, where $\alpha_i$ is its input plaquette label.  Thus the histories counted in coefficients are stable finite-energy histories; the complementary sectors are treated as remainders in Theorem~\ref{thm:rooted-coeff-construction}.
\end{definition}

For a stable rooted history $\mathfrak h$ appearing in the coefficientwise expansion, we define its projective mass by
\[
M(\mathfrak h)=\sum_p m_{\mathfrak h}(p),
\]
where $m_{\mathfrak h}(p)$ is the aggregate projective mass at $p$ obtained by summing, with multiplicity one for each local contribution of the history, the projective masses of the input labels at $p$ and of the labels appearing at $p$ in the terminal source. Labels merely transported unchanged through intermediate loop-only segments are not counted again.

\begin{proposition}\label{prop:trajectory-estimate}
Let $R\ge0$ and fix an initial rooted state $X_0$ whose compactly supported plaquette decoration is either zero, or fixed, stable and of finite $\mathcal E$-energy.  There exists $T_0(d)<\infty$ such that, for every $T>T_0(d)$, the sum of the absolute values of all stable rooted histories of order at most $R$ starting from $X_0$ is finite, uniformly in $N$ and $\Lambda$.  More precisely, there are constants $C_{R,X_0,T},c>0$ such that
\begin{equation}\label{eq:history-locality}
\sum_{r=0}^R\sum_{\substack{\mathfrak h\text{ stable}: \operatorname{ord}(\mathfrak h)=r\\
\operatorname{dist}(\operatorname{supp}\mathfrak h,\operatorname{supp}X_0)\ge A}}
|w^{[r]}_{\Lambda,T}(\mathfrak h)|
\le
C_{R,X_0,T}e^{-cA}.
\end{equation}
Moreover the same bound remains true, on retained stable finite-energy local data, if one replaces one local coefficient by its order-$(R+1)$ remainder, after multiplying the right-hand side by $N^{-R-1}$.
\end{proposition}

\begin{proof}
Choose any $\gamma\ge0$, let $\kappa$ be the constant of Lemma~\ref{lem:loop-only-entropy}, choose $a>\gamma+\kappa$, and choose any $\eta>0$.  We work throughout this proposition with the stable retained transfer, estimated in the adjoint weighted total-variation sense of Proposition~\ref{prop:raw-plaquette-transfer}; the large-mass/unstable complement has already been separated as the exponentially small remainder of that proposition.  Then choose $T_0(d)$ so large that
\begin{equation}\label{eq:paid-geometric-small}
\Theta_T:=\sum_{r\ge1}\exp\{r(\log C-cT+aC+\eta C)\}<1.
\end{equation}
Here $C,c$ are the constants of Proposition~\ref{prop:raw-plaquette-transfer}. Consider a history with input energies $r_1,\ldots,r_m$.  Between two transfers, or before the first transfer, the sum over all loop-only segments of total order at most $R$ is controlled by Lemma~\ref{lem:loop-only-entropy}.  The factor $e^{(\gamma+\kappa)\ell}$ produced on the input side of this lemma is dominated by the $e^{a\ell}$ weight transported by the adjoint plaquette-transfer estimate because $a>\gamma+\kappa$.  More explicitly, if $X_i^-$ is the state just before the $i$-th transfer and $X_i^+$ the state just after it, then the adjoint shell estimate of Proposition~\ref{prop:raw-plaquette-transfer} gives the corresponding weighted total-variation bound
\[
C^{r_i}e^{-cTr_i+aCr_i+\eta Cr_i}.
\]
After the loop-only entropy factors are inserted, the contribution of all histories with these energies is bounded by
\[
C_{R,\gamma,X_0}
\prod_{i=1}^m C'e^{-c'Tr_i},
\]
where $C'$ and $c'$ depend only on $d$ and on the choice of $a,\eta$, while $C_{R,\gamma,X_0}$ absorbs the initial factor $e^{(\gamma+\kappa)\ell(L_0)}$ and the finitely many exceptional loop-only moves of total order at most $R$.  Summing over $m\ge0$ and $r_i\ge1$ gives a convergent geometric series by \eqref{eq:paid-geometric-small}. This proves absolute summability, uniformly in $N$ and $\Lambda$.

For locality, observe first that a loop-only segment stays inside the support already created by the preceding steps: it only reconnects existing occurrences of the rooted edge and never introduces a new plaquette.  A transfer is local around the current root.  More precisely, by Lemma~\ref{lem:local-transfer-distortion}, a transfer with energy $r$ can create only $O_d(r)$ new tensor strands and can move the outgoing root only by a distance $O_d(r)$ in the lattice.  Consequently, if a history reaches distance at least $A$ from $\operatorname{supp}X_0$, then either the sum of its energies or the total length created by its transfers is at least $c_dA$.  Loop-only pieces cannot be responsible for such a displacement unless this support has already been created by previous transfers.  Keeping an extra factor $e^{\rho\operatorname{diam}(\operatorname{supp}\mathfrak h)}$ in the previous row estimate, with $\rho>0$ small, therefore only changes the plaquette-transfer factor by $e^{\rho C_dr}$.  After increasing $T_0(d)$, the strict inequality \eqref{eq:paid-geometric-small} is preserved.  Removing this extra factor gives \eqref{eq:history-locality}.

The assertion about remainders follows from the same majorant.  On every fixed stable finite-mass local datum, the Schur--Weyl, Haar, recoupling and heat-kernel factors have controlled asymptotic expansions; after truncation at order $R$, their remainders are bounded by $N^{-R-1}$ times the same local majorant used above.  If one local factor is replaced by such a remainder, the preceding summation gives the same history bound multiplied by $N^{-R-1}$.
\end{proof}

\begin{proposition}\label{prop:mass-sensitive-summability}
There exists a constant $c>0$, depending only on $d$, such that the following holds.  For every coefficient order $r\ge0$, there exist $T_{\mathrm{mass}}(d,r)<\infty$ and $C_r<\infty$ such that, for every $T>T_{\mathrm{mass}}(d,r)$, every finite loop family $L$, every root $x$, every finite volume $\Lambda$, and every $A\ge0$,
\begin{equation}\label{eq:mass-sensitive-summability}
\sum_{\substack{\operatorname{ord}(\mathfrak h)=r\\ M(\mathfrak h)\ge A}}
\left|w^{[r]}_{\Lambda,T}(\mathfrak h)\right|
\le C_r^{\ell(L)}\exp\{-(cT-\log C_r)A\},
\end{equation}
where the sum is over stable finite-energy rooted histories starting from $(L,x,\varnothing)$.  After increasing $T_{\mathrm{mass}}(d,r)$ if necessary, the exponent $cT-\log C_r$ is positive.  The same estimate holds in infinite volume.
\end{proposition}

\begin{proof}
Let $\mathfrak h$ be a stable rooted history starting from $(L,x,\varnothing)$, and let $\alpha_1,\ldots,\alpha_m$ be the input plaquette labels of its transfers. Loop-only histories do not change plaquette decorations. By the local distortion estimate of Lemma~\ref{lem:local-transfer-distortion}, a transfer with input energy $\mathcal E(\alpha_i)$ can create only $O_d(\mathcal E(\alpha_i))$ projective mass in its output decoration and in the terminal source produced at that step. Therefore
\begin{equation}\label{eq:mass-controlled-by-paid-energy}
M(\mathfrak h)
\le C_d\sum_{i=1}^m \mathcal E(\alpha_i).
\end{equation}
The constant is independent of the ambient loop family; the initial state has empty plaquette decoration, and the Wilson collar carries no projective plaquette mass.

Repeat the proof of Proposition~\ref{prop:trajectory-estimate} with the additional factor $e^{\theta M(\mathfrak h)}$. Using \eqref{eq:mass-controlled-by-paid-energy}, this multiplies the plaquette-transfer shell with input energy $s$ by at most $e^{\theta C_d s}$. Thus the stable adjoint plaquette-transfer estimate has the same form as in Proposition~\ref{prop:raw-plaquette-transfer}, but with the exponent $cT$ replaced by $cT-C_d\theta$, after changing constants depending only on $d$.

Choose $\theta=\theta_T$ of the form
\[
\theta_T=c_1T-\log C_r,
\]
with $c_1>0$ small enough, depending only on $d$, so that the weighted transfer remains summable in the proof of Proposition~\ref{prop:trajectory-estimate}. For every fixed order $r$, this is possible for all $T>T_{\mathrm{mass}}(d,r)$, after enlarging $C_r$. The finite-order loop-only entropy estimate, Lemma~\ref{lem:loop-only-entropy}, is unchanged by the mass weight. Hence
\[
\sum_{\operatorname{ord}(\mathfrak h)=r}
|w^{[r]}_{\Lambda,T}(\mathfrak h)|e^{\theta_T M(\mathfrak h)}
\le C_r^{\ell(L)}.
\]
Markov's inequality gives
\[
\sum_{\substack{\operatorname{ord}(\mathfrak h)=r\\ M(\mathfrak h)\ge A}}
|w^{[r]}_{\Lambda,T}(\mathfrak h)|
\le C_r^{\ell(L)}e^{-\theta_TA}.
\]
Renaming the resulting positive constant as $c$ gives \eqref{eq:mass-sensitive-summability}.

The large-mass and unstable complementary sectors are separated in the coefficientwise construction and contribute $O(N^{-R-1})$ at every prescribed order. They therefore do not enter the fixed coefficient $w^{[r]}$. The infinite-volume estimate follows from the same bound and the exponential locality of Theorem~\ref{thm:rooted-coeff-construction}, or directly by the same support-counting argument on $\Z^d$.
\end{proof}

\subsection{Coefficientwise construction}\label{subsec:coefficientwise-construction}

We now make precise the limiting protocol used to remove the auxiliary finite representation cutoffs.  The raw plaquette-label sums are first interpreted with an auxiliary finite representation truncation $F$.  Inside $F$, we separate three parts of the plaquette-transfer kernel: a retained stable finite-energy kernel, a stable high-energy tail, and a large-mass or unstable complement.  Exceptional transfers are controlled at their first occurrence, so the error bounds remain uniform after summing over an arbitrary number of retained transfers.

Fix $a,\eta>0$ as in Proposition~\ref{prop:raw-plaquette-transfer}, and choose $T$ so large that
\begin{equation}
\label{eq:good-transfer-contraction}
\theta_T
:=C_0\sum_{r\ge1}\exp\{r(\log C-cT+aC+\eta C)\}<1.
\end{equation}
Let $F\subset\widehat{\U(N)}$ be a finite representation set containing the labels of the fixed input decoration.  The finite-truncation unknown $\Psi^F_{\Lambda,T,N}$ is obtained from the raw equation by restricting every free plaquette-label sum, both in the reference-law adjunction and in terminal sources, to labels in $F$.  For $A<\infty$, set
\[
\begin{aligned}
S_A(N)&=\{\lambda\in\widehat{\U(N)}:
\lambda\text{ is stable},\ m(\lambda)\le\varepsilon_dN,
\ \mathcal E(\lambda)\le A\},\\
S_A^F(N)&:=F\cap S_A(N).
\end{aligned}
\]
We decompose the finite-cutoff raw plaquette-transfer kernel according to the input and output labels touched by the one-step transfer. A one-step plaquette transfer
\[
X=(L,x,\zeta)\longrightarrow X'=(L_u,\mathfrak r(L_u),\zeta^{p\to\beta_p})
\]
with input label $\alpha_p\in A_p(\zeta)$ is called $(A,F)$-retained if
\[
\alpha_p\in S_A^F(N)\qquad\text{and}\qquad \beta_p\in S_A^F(N).
\]
The retained kernel $K_A^F$ is the restriction of the finite-truncation kernel to $(A,F)$-retained transfers. The stable high-energy kernel $H_A^F$ is the stable small-mass part inside $F$ for which at least one of $\alpha_p$ and $\beta_p$ is stable but not in $S_A(N)$. The kernel $U_A^F$ is the complement inside $F$, namely the part in which at least one of $\alpha_p$ and $\beta_p$ is unstable or outside the stable small-mass range.  The same decomposition is applied to terminal plaquette-transfer sources:
\[
G^F=G_A^F+G_{H,A}^F+G_{U,A}^F,
\]
where $G_A^F$ is the retained terminal source, $G_{H,A}^F$ the stable high-energy terminal source, and $G_{U,A}^F$ the large-mass or unstable terminal source.  The loop-only terminal source is included in $G_A^F$, since it carries no plaquette label.  The estimates for these terminal pieces are those of Proposition~\ref{prop:decorated-terminal-sources}; in particular they enter as terminal coefficients in the finite-cutoff expansion. By Proposition~\ref{prop:raw-plaquette-transfer}, each retained transfer satisfies the adjoint weighted total-variation estimate
\begin{equation}
\label{eq:retained-good-kernel-bound}
\theta_T
:=C_0\sum_{r\ge1}\exp\{r(\log C-cT+aC+\eta C)\}<1,
\end{equation}
uniformly in $A,F,N$ and $\Lambda$.  Moreover, Lemma~\ref{lem:local-transfer-distortion} implies that if a stable transfer is not retained by $K_A^{F}$, then either its input energy is larger than $A/C_d$, or its output energy is larger than $A$ and therefore its input energy is larger than $A/C_d$ after changing $C_d$.  Hence Lemma~\ref{lem:paid-shell-row} gives the adjoint first-exceptional-transfer bound
\begin{equation}
\label{eq:stable-exceptional-kernel-bound}
\varepsilon_A\le C\sum_{r>A/C_d} C^r e^{-cTr}
\xrightarrow[A\to\infty]{}0,
\end{equation}
uniformly in $F,N$ and $\Lambda$.  Finally, for every prescribed expansion order $R$, the collective large-mass/unstable estimate \eqref{eq:paid-unstable-complement}, combined with the exponential non-spectral local bound, gives the adjoint first-exceptional-transfer estimate
\begin{equation}
\label{eq:unstable-exceptional-kernel-bound}
C_{A,R,X_0,T}N^{-R-1}
\end{equation}
after coefficient extraction up to order $R$, uniformly in $F$ and in the finite volume.

For later reference we also record the corresponding finite-cutoff identity.  With the notation above,
\begin{equation}
\label{eq:finite-cutoff-rooted-equation}
\begin{aligned}
\Psi^F_{\Lambda,T,N}
={}&G_A^F+G_{H,A}^F+G_{U,A}^F
+\mathcal T_{\Lambda,T,N}\Psi^F_{\Lambda,T,N}\\
&+K_A^F\Psi^F_{\Lambda,T,N}
+H_A^F\Psi^F_{\Lambda,T,N}
+U_A^F\Psi^F_{\Lambda,T,N}.
\end{aligned}
\end{equation}
A retained finite-cutoff history from $X_0$ is a rooted trajectory obtained by alternating finite loop-only segments generated by $\mathcal T_{\Lambda,T,N}$ with retained transfers from $K_A^F$, and ending with the retained source $G_A^F$.  Its weight $w_{\Lambda,T,N}(\mathfrak h)$ is the product of the local coefficients, including the terminal source coefficient.  We denote by $\mathcal H^{F,A,\mathrm{ret}}_{\le m}(X_0)$ the finite set of retained histories with at most $m$ retained transfers.  Iterating only the retained kernel $K_A^F$ and the loop-only operator $\mathcal T_{\Lambda,T,N}$, and stopping after $m$ transfers, gives the finite algebraic identity
\begin{equation}
\label{eq:finite-cutoff-history-identity}
\Psi^F_{\Lambda,T,N}(X_0)
=
\sum_{\mathfrak h\in\mathcal H^{F,A,\mathrm{ret}}_{\le m}(X_0)}
w_{\Lambda,T,N}(\mathfrak h)
+
\mathcal T^{F,A,m}_{\Lambda,T,N}(X_0)
+
\mathcal B^{F,A,m}_{\Lambda,T,N}(X_0),
\end{equation}
where $\mathcal T^{F,A,m}$ is the retained-history tail with more than $m$ retained transfers, and $\mathcal B^{F,A,m}$ is the contribution of histories in which the first non-retained input object is either an $H_A^F$-transfer or terminal source, or a $U_A^F$-transfer or terminal source.

The first-exceptional-transfer decomposition gives the bounds
\begin{align}
\label{eq:first-exceptional-stable-bound}
|\mathcal B^{F,A,m}_{H}(X_0)|
&\le C_{R,X_0,T}\frac{\varepsilon_A}{(1-\theta_T)^2},\\
\label{eq:first-exceptional-unstable-bound}
|\mathcal B^{F,A,m}_{U}(X_0)|
&\le C_{A,R,X_0,T}\frac{N^{-R-1}}{(1-\theta_T)^2}
\end{align}
after coefficient extraction up to order $R$, uniformly in $m,F$ and $\Lambda$.  Indeed, mark the first exceptional transfer or exceptional terminal source.  The retained histories before it and after it are summed by the geometric majorant generated by~\eqref{eq:retained-good-kernel-bound}; the loop-only pieces are already included in the finite-order loop-only estimate of Lemma~\ref{lem:loop-only-entropy} and Proposition~\ref{prop:trajectory-estimate}.  This removes any dependence of the exceptional-sector constants on the number of transfers.  The retained tail satisfies
\begin{equation}
\label{eq:retained-history-tail-bound}
\lim_{m\to\infty}\sup_{F,A,N,\Lambda}
|\mathcal T^{F,A,m}_{\Lambda,T,N}(X_0)|=0,
\end{equation}
by Proposition~\ref{prop:trajectory-estimate}.

\begin{theorem}\label{thm:rooted-coeff-construction}
Let $d\ge2$.  There exists $T_0(d)<\infty$ such that the following holds for $T>T_0(d)$.  Fix either a Wilson initial state $X_0=(L,x,\varnothing)$, or more generally a rooted state $X_0=(L,x,\zeta)$ whose compactly supported plaquette decoration is fixed, stable and has finite $\mathcal E$-energy, and fix an integer $R\ge0$.  Then, uniformly in the finite volume $\Lambda$,
\begin{equation}\label{eq:rooted-asymptotic-expansion}
\Psi_{\Lambda,T,N}(X_0)
=
\sum_{r=0}^R N^{-r}\Psi^{(r)}_{\Lambda,T}(X_0)
+
O_{R,X_0,T}(N^{-R-1}).
\end{equation}
The coefficient $\Psi^{(r)}_{\Lambda,T}(X_0)$ is the absolutely convergent sum of all stable order-$r$ rooted histories starting from $X_0$:
\begin{equation}\label{eq:rooted-history-coeff}
\Psi^{(r)}_{\Lambda,T}(X_0)
=
\sum_{\operatorname{ord}(\mathfrak h)=r}w^{[r]}_{\Lambda,T}(\mathfrak h).
\end{equation}
The coefficients are exponentially local in the sense of \eqref{eq:history-locality}, and therefore have infinite-volume limits along every exhaustion.
\end{theorem}

\begin{proof}
We start from the finite-cutoff identity \eqref{eq:finite-cutoff-rooted-equation}.  For fixed $F$, $A$ and transfer-count cutoff $m$, the iterated identity \eqref{eq:finite-cutoff-history-identity} is finite: all label sets, local channel sets and history sets are finite.  At this stage the identity is a finite algebraic expansion.

On every retained stable finite-energy local datum, expand the local Schur--Weyl coefficients, Haar coefficients, recoupling coefficients, heat-kernel Casimir factors and local adjoint coefficients in powers of $1/N$ up to order $R$.  The local asymptotic orders are non-negative after the normalizations used throughout the rooted equation.  A same-trace loop-only splitting has degree zero, whereas every merger has degree $2$. For a pinned plaquette block, Lemma~\ref{lem:degree-pinned-block} and Theorem~\ref{thm:pinned-block} ensure that, after trace normalization and after the dimension cancellations of Section~\ref{sec:pinned-haar-projection}, no positive power of $N$ remains.  If at least two marked trace components remain in one connected block, the separating-channel reciprocal dimension factor gives one additional negative power of $N$.  Hence no local retained coefficient starts with a positive power of $N$.

By Lemma~\ref{lem:fixed-degree-orthonormalization} and Subsection~\ref{subsec:orthonormal-channel-recoupling}, all retained local factors have controlled expansions in fixed stable degree. Proposition~\ref{prop:trajectory-estimate} then sums uniformly the contributions in which one factor is replaced by its order-$(R+1)$ remainder. The same majorant controls products whose accumulated order first exceeds $R$, after extracting the factor $N^{-R-1}$ at that first excess.

We now remove the auxiliary cutoffs.  First let $m\to\infty$.  The retained-history tail tends to zero uniformly by \eqref{eq:retained-history-tail-bound}.  Histories with at least one stable high-energy transfer are controlled uniformly in $m$ by the first-exceptional-transfer estimate \eqref{eq:first-exceptional-stable-bound}; since $\varepsilon_A\to0$ by \eqref{eq:stable-exceptional-kernel-bound}, their total contribution disappears when $A\to\infty$.  Histories with at least one large-mass or unstable transfer are controlled by \eqref{eq:first-exceptional-unstable-bound}; for fixed $A$ they contribute $O_{A,R,X_0,T}(N^{-R-1})$, uniformly in $F$ and $\Lambda$.  Finally, for each fixed $A$ and $N$, the retained set $S_A(N)$, together with all one-step local recoupling outputs from $S_A(N)$, is finite.  Therefore every sufficiently large finite truncation $F$ gives the same retained coefficients, and the complement is already included in the preceding stable high-energy and large-mass/unstable estimates.  We may let $F\uparrow\widehat{\U(N)}$.

After sending $m\to\infty$, $A\to\infty$ and $F\uparrow\widehat{\U(N)}$, the coefficient $\Psi^{(r)}_{\Lambda,T}(X_0)$ is therefore well-defined as the absolutely convergent sum of all stable finite-energy histories of total asymptotic order $r$, and the preceding paragraph shows that this definition is independent of the auxiliary truncation procedure.  Collecting the terms of orders $0,\ldots,R$ gives \eqref{eq:rooted-asymptotic-expansion}, with a remainder uniform in the finite volume.  The exponential locality is exactly the locality estimate \eqref{eq:history-locality}, obtained by the same history majorant with the additional weight $e^{\rho\operatorname{diam}(\operatorname{supp}\mathfrak h)}$ used in Proposition~\ref{prop:trajectory-estimate}.
\end{proof}

\section{The master field: construction and area law}\label{sec:master-field}

In this section we construct the strongly coupled large-$N$ heat-kernel master field and prove its area law. The construction will follow from the coefficientwise rooted construction of Theorem~\ref{thm:rooted-coeff-construction}.  We first extract root-independent Wilson-loop coefficients, then prove determinant reduction at leading order, infinite-volume locality, and factorization through connected cumulants.

Let us start with the following extended version of Theorem~\ref{thm:main}.

\begin{theorem}\label{thm:main-expanded}
Let $d\ge2$.  There exists $T_0(d)<\infty$ such that, for every $T>T_0(d)$ and every fixed finite loop family $L$, the following hold.
\begin{enumerate}
\item For every finite $\Lambda$ containing $L$ and every $R\ge0$,
\begin{equation}\label{eq:finite-volume-expansion}
\Phi_{\Lambda,T,N}(L)
=
\sum_{r=0}^R N^{-r}\Phi^{(r)}_{\Lambda,T}(L)
+
O_{R,L,T}(N^{-R-1}),
\end{equation}
with a remainder uniform in $\Lambda$.
\item Each coefficient has an infinite-volume limit
\begin{equation}\label{eq:coefficient-limit}
\Phi^{(r)}_{\infty,T}(L)=\lim_{\Lambda\uparrow\Z^d}\Phi^{(r)}_{\Lambda,T}(L),
\end{equation}
and the convergence is exponentially fast in $\operatorname{dist}(\operatorname{supp}L,\partial\Lambda)$.
\item More generally, if $N_n\to\infty$ and $\Lambda_n$ is any sequence of finite volumes with $\operatorname{dist}(\operatorname{supp}L,\partial\Lambda_n)\to\infty$, then
\begin{equation}\label{eq:largeN-limit}
\lim_{n\to\infty}\Phi_{\Lambda_n,T,N_n}(L)=\Phi^{(0)}_{\infty,T}(L).
\end{equation}
In particular, this implies the iterated limit along every exhaustion $\Lambda_n\uparrow\Z^d$.
\item The leading coefficient factorizes:
\begin{equation}\label{eq:factorization}
\Phi^{(0)}_{\infty,T}(\ell_1,\ldots,\ell_k)
=
\prod_{i=1}^k\Phi^{(0)}_{\infty,T}(\ell_i).
\end{equation}
\end{enumerate}
We call $\Phi^{(0)}_{\infty,T}$ the heat-kernel master field at strong coupling.
\end{theorem}

By Proposition~\ref{prop:determinant-reduction-leading}, the leading coefficient in Theorem~\ref{thm:main-expanded} may be computed in determinant-reduced form.

Throughout this section $T>T_0(d)$, where $T_0(d)$ is chosen so that Theorem~\ref{thm:rooted-coeff-construction} holds.  Let $L=(\ell_1,\ldots,\ell_k)$ be fixed.  If all components of $L$ are trivial, then $\Phi_{\Lambda,T,N}(L)=1$ and there is nothing to prove.  Otherwise choose a non-trivial occurrence $x_L$ of an oriented edge in one of the loops and set $X_L=(L,x_L,\varnothing)$.  By definition of the extended coefficient,
\begin{equation}\label{eq:ext-coef}
\Psi_{\Lambda,T,N}(X_L)=\Phi_{\Lambda,T,N}(L).
\end{equation}
Although $X_L$ contains a root, the quantity in \eqref{eq:ext-coef} does not.

\subsection{Extraction of the Wilson loop expansion}

The rooted construction applies to a rooted decorated-loop state, while the Wilson expectation itself is unrooted.  We first remove this auxiliary choice and obtain root-independent Wilson-loop coefficients.

\begin{proposition}\label{prop:wilson-expansion-from-rooted}
For every finite $\Lambda$ containing $L$ and every $R\ge0$, there are coefficients $\Phi^{(r)}_{\Lambda,T}(L)$ such that
\begin{equation}\label{eq:expansion-phi}
\Phi_{\Lambda,T,N}(L)
=
\sum_{r=0}^R N^{-r}\Phi^{(r)}_{\Lambda,T}(L)
+
O_{R,L,T}(N^{-R-1}),
\end{equation}
with a remainder uniform in $\Lambda$.  The coefficients do not depend on the choice of root.
\end{proposition}

\begin{proof}
Apply Theorem~\ref{thm:rooted-coeff-construction} to $X_L=(L,x_L,\varnothing)$ and use \eqref{eq:ext-coef}.  Define
\[
\Phi^{(r)}_{\Lambda,T}(L):=\Psi^{(r)}_{\Lambda,T}(X_L).
\]
If $x$ and $y$ are two possible roots, then for every $N$,
\[
\Psi_{\Lambda,T,N}(L,x,\varnothing)=\Phi_{\Lambda,T,N}(L)=
\Psi_{\Lambda,T,N}(L,y,\varnothing).
\]
Both sides have controlled asymptotic expansions with remainders of arbitrary order.  Uniqueness of asymptotic coefficients gives equality of the coefficients for the two roots.
\end{proof}

\subsection{Determinant reduction at leading order}

Recall that a stable label is written
\[
\lambda_p=q_p\mathbf 1_N+[\lambda_p^+,\lambda_p^-]_N,
\qquad
s_p=|\lambda_p^+|-|\lambda_p^-|.
\]
The determinant variables $q_p$ are one-dimensional charges, whereas $[\lambda_p^+,\lambda_p^-]_N$ is the projective, or reduced, part.

\begin{proposition}\label{prop:determinant-reduction-leading}
At order $N^0$, the determinant-shift sector factors as a vacuum determinant partition function independent of $L$ and of the reduced labels. Consequently, the leading normalized coefficient is computed by the determinant-reduced expansion.
\end{proposition}

\begin{proof}
In the stable small-mass range, Lemma~\ref{lem:central-selection} gives the split central selection rule~\eqref{eq:split-selection}.  The large-mass/unstable complement is exponentially small in $N$, by Lemma~\ref{lem:central-cutoff}, Lemma~\ref{lem:nonstable-weighted-tail} and Proposition~\ref{prop:nonstable-polymer-tail}; hence it does not contribute to the coefficient of order $N^0$.

For labels satisfying \eqref{eq:split-selection}, the determinant part of the character contributes
\[
\prod_p\det(U_{\partial p})^{q_p}
=\prod_e\det(U_e)^{(\partial^*q)(e)}=1.
\]
Thus the finite-$N$ topological coefficient depends only on the reduced labels.  The heat-kernel factor is the only remaining place where $q$ appears.  Using the Casimir identity~\eqref{eq:casimir-det-shift}, we get the exact factorization
\begin{align*}
&\exp\left[-\frac{T}{2N}C_2
\bigl(q_p\mathbf 1_N+[\lambda_p^+,\lambda_p^-]_N\bigr)\right]\\
&\qquad=
e^{-Tq_p^2/2}
\exp\left[-\frac{T}{2N}C_2([\lambda_p^+,\lambda_p^-]_N)\right]
e^{-Tq_ps_p/N}.
\end{align*}
On every fixed reduced-mass sector, $|s_p|$ is bounded independently of $N$, and
\[
\bigl|e^{-Tq_ps_p/N}-1\bigr|
\le \frac{T|q_p||s_p|}{N}e^{T|q_p||s_p|/N}.
\]
Using $a|q|\le (T/4)q^2+a^2/T$ with $a=T|s_p|/N$, the product of the right-hand side with $e^{-Tq_p^2/2}$ is bounded by $N^{-1}$ times a constant depending only on the fixed reduced-mass sector and the summable Gaussian $|q_p|e^{-Tq_p^2/4}$. Hence the replacement of $e^{-Tq_ps_p/N}$ by $1$ has an absolutely summable $O(N^{-1})$ error, uniformly under the determinant-sector sum.  Therefore the order-$N^0$ determinant contribution is not the single configuration $q=0$, but the whole determinant-sector sum
\begin{equation}\label{eq:z-det}
Z^{\det}_{\Lambda,T}=\sum_{\partial^*q=0}\prod_{p\in P(\Lambda)}e^{-Tq_p^2/2}.
\end{equation}
This factor depends only on the finite volume and on $T$.  It is independent of the Wilson loop family and independent of the projective labels satisfying the second equation in~\eqref{eq:split-selection}.  The same determinant factor appears in the vacuum denominator for every reduced vacuum label field satisfying $\partial^*s=0$: the determinant constraint is again $\partial^*q=0$, and it is independent of the reduced labels. Thus, at order $N^0$, $Z^{\det}_{\Lambda,T}$ factors out of the entire vacuum denominator, not merely out of the trivial reduced sector. No uniform bound on $Z^{\det}_{\Lambda,T}$ is needed, because this leading determinant factor is common to numerator and denominator and cancels before the normalized leading coefficient is read off.  After this cancellation the leading coefficient is exactly the reduced one.
\end{proof}

\begin{remark}
Proposition~\ref{prop:determinant-reduction-leading} is the lattice analogue of the determinant decoupling which appears in the large-$N$ analysis of two-dimensional Yang--Mills theory on compact surfaces \cite{Lem22,Lem25}. In the latter setting, the $\U(N)$ partition function asymptotically factors into an abelian determinant sector and a determinant-reduced, or $\SU(N)$-type, projective sector. Here the same mechanism appears locally on the plaquette field: the determinant shifts form an integer field $q:P(\Lambda)\to\Z$ constrained by $\partial^*q=0$, and their leading heat-kernel contribution is the Gaussian factor $Z^{\det}_{\Lambda,T}$. Since this factor is independent of the Wilson insertion and of the reduced plaquette labels, it cancels against the vacuum denominator before the leading normalized coefficient is read off.
\end{remark}

\subsection{Infinite-volume locality}

For a rooted state $X=(L,x,\zeta)$, write $\operatorname{supp}(X)$ for the union of the loop support, the root edge and the decoration support $D(\zeta)$.

\begin{proposition}\label{prop:coefficient-locality}
For every fixed rooted state $X$ and every $r\ge0$, the coefficient $\Psi^{(r)}_{\Lambda,T}(X)$ has an infinite-volume limit
\[
\Psi^{(r)}_{\infty,T}(X)=\lim_{\Lambda\uparrow\Z^d}\Psi^{(r)}_{\Lambda,T}(X).
\]
Moreover, there exist constants $C_{r,X,T},c>0$ such that
\begin{equation}\label{eq:bound-psi-r}
\left|\Psi^{(r)}_{\Lambda,T}(X)-\Psi^{(r)}_{\infty,T}(X)\right|
\le
C_{r,X,T}e^{-c\operatorname{dist}(\operatorname{supp}X,\partial\Lambda)}.
\end{equation}
\end{proposition}

\begin{proof}
By Theorem~\ref{thm:rooted-coeff-construction}, $\Psi^{(r)}_{\Lambda,T}(X)$ is an absolutely convergent sum of order-$r$ rooted histories, and those histories satisfy the locality estimate~\eqref{eq:history-locality}.  If two finite volumes agree on the $A$-neighbourhood of $\operatorname{supp}(X)$, all histories contained in this neighbourhood have identical weights in the two volumes.  The difference is therefore bounded by the total contribution of histories reaching distance at least $A$, which is $O(e^{-cA})$.  Taking $A=\operatorname{dist}(\operatorname{supp}X,\partial\Lambda)$ proves that the coefficients are Cauchy along every exhaustion and gives \eqref{eq:bound-psi-r}.
\end{proof}

\begin{corollary}\label{cor:infinite-volume-coefficients}
For every fixed loop family $L$ and every $r\ge0$, the limit
\[
\Phi^{(r)}_{\infty,T}(L):=\lim_{\Lambda\uparrow\Z^d}\Phi^{(r)}_{\Lambda,T}(L)
\]
exists, is independent of the exhaustion, and satisfies
\[
\left|\Phi^{(r)}_{\Lambda,T}(L)-\Phi^{(r)}_{\infty,T}(L)\right|
\le C_{r,L,T}e^{-c\operatorname{dist}(\operatorname{supp}L,\partial\Lambda)}.
\]
\end{corollary}

\begin{proof}
Apply Proposition~\ref{prop:coefficient-locality} to $X_L=(L,x_L,\varnothing)$ and use \eqref{eq:ext-coef} at the coefficient level.
\end{proof}

\begin{corollary}\label{cor:joint-largeN-volume-limit}
If $N_n\to\infty$ and $\operatorname{dist}(\operatorname{supp}L,\partial\Lambda_n)\to\infty$, then
\[
\Phi_{\Lambda_n,T,N_n}(L)\longrightarrow\Phi^{(0)}_{\infty,T}(L).
\]
\end{corollary}

\begin{proof}
Take $R=0$ in Proposition~\ref{prop:wilson-expansion-from-rooted}. Uniformity in the finite volume gives
\[
\Phi_{\Lambda_n,T,N_n}(L)=\Phi^{(0)}_{\Lambda_n,T}(L)+O_{L,T}(N_n^{-1}).
\]
The first term converges to $\Phi^{(0)}_{\infty,T}(L)$ by Corollary~\ref{cor:infinite-volume-coefficients}, with an exponentially small error in the distance to the boundary. Both errors tend to zero.
\end{proof}

\subsection{Connected cumulants and factorization}\label{subsec:connected-cumulants-factorization}

We now prove leading-order factorization by showing that connected cumulants involving at least two trace components vanish at order $N^0$.

If $L=(\ell_1,\ldots,\ell_k)$, define the normalized connected cumulant by
\begin{equation}\label{eq:norm-conn-cumulant}
\Phi^{\conn}_{\Lambda,T,N}(L)
=
\sum_{\pi\in\mathcal P(k)}(|\pi|-1)!(-1)^{|\pi|-1}
\prod_{B\in\pi}\Phi_{\Lambda,T,N}(L_B),
\end{equation}
where $L_B=(\ell_i)_{i\in B}$.

\begin{lemma}[Linked-cluster form of Wilson cumulants]\label{lem:wilson-cumulant-linked-cluster}
For $T$ in the strong-coupling regime, the normalized connected cumulant~\eqref{eq:norm-conn-cumulant} admits an absolutely convergent local-channel expansion over connected polymer clusters which meet every marked trace component. In particular, all contributions whose active support splits according to a non-trivial partition of $\{1,\ldots,k\}$ cancel in the M\"obius sum defining the cumulant.
\end{lemma}

\begin{proof}
For each subfamily $L_B$, use the defect-ratio/local-channel representation and the strong-coupling block bounds of Section~\ref{sec:pinned-haar-projection} to write the normalized moment as an absolutely convergent polymer expansion. The weight of a collection of mutually disconnected active components factors over those components, since both the local channel kernels and the product reference law do. The usual linked-cluster formula therefore expresses moments as sums over set partitions of connected cluster weights attached to the corresponding trace labels. Substituting this expression into~\eqref{eq:norm-conn-cumulant} performs M\"obius inversion on the lattice of set partitions: any family of clusters whose marked trace set decomposes according to a non-trivial partition has total coefficient zero. The surviving terms are exactly the connected clusters meeting every marked trace component. Absolute convergence justifies all regroupings.
\end{proof}

\begin{proposition}\label{prop:factorization}
For every finite loop family $L=(\ell_1,\ldots,\ell_k)$,
\begin{equation}\label{eq:factorization-phi}
\Phi^{(0)}_{\infty,T}(L)=\prod_{i=1}^k\Phi^{(0)}_{\infty,T}(\ell_i).
\end{equation}
\end{proposition}

\begin{proof}
It is enough to prove that every leading connected cumulant involving at least two trace components vanishes. Fix $k\ge2$. By Lemma~\ref{lem:wilson-cumulant-linked-cluster}, the connected cumulant has an absolutely convergent expansion over connected local-channel clusters which meet every marked trace component. In particular, every surviving cluster meets at least two marked Wilson collars. The resummed connected-block estimate of Theorem~\ref{thm:pinned-block}, together with Lemma~\ref{lem:core-collar-connectivity}, therefore supplies an additional reciprocal dimension from a non-trivial separating channel. Since every non-trivial reduced stable label has dimension at least $N$, every such connected cluster is $O(N^{-1})$ after trace normalization. The stable high-energy and large-mass/unstable sectors are exponentially small in $N$ by the estimates used in Theorem~\ref{thm:rooted-coeff-construction}, so they do not contribute at order $N^0$.

Consequently the coefficient of order $N^0$ in $\Phi^{\conn}_{\Lambda,T,N}(\ell_1,\ldots,\ell_k)$ vanishes uniformly in the finite volume. Passing to the infinite-volume coefficients by Corollary~\ref{cor:infinite-volume-coefficients} gives
\[
\Phi^{\conn,(0)}_{\infty,T}(\ell_1,\ldots,\ell_k)=0, \qquad k\ge2.
\]
Finally, the inverse moment--cumulant formula expresses \(\Phi^{(0)}_{\infty,T}(L)\) as a sum over partitions of products of leading connected cumulants. Since all leading connected cumulants with blocks of size at least two vanish, only the partition into singletons contributes. This gives the factorization~\eqref{eq:factorization-phi}.
\end{proof}

We can now complete the proof of the main theorem.

\begin{proof}[Proof of Theorem~\ref{thm:main-expanded}]
Proposition~\ref{prop:wilson-expansion-from-rooted} gives the uniform finite-volume expansion \eqref{eq:finite-volume-expansion}.  Corollary~\ref{cor:infinite-volume-coefficients} gives the infinite-volume limits of all coefficients and their exponential locality, proving \eqref{eq:coefficient-limit}.  The joint large-$N$/infinite-volume statement is Corollary~\ref{cor:joint-largeN-volume-limit}; in particular it contains the iterated limit along every exhaustion.  Finally, determinant reduction follows from Proposition~\ref{prop:determinant-reduction-leading}, and factorization follows from Proposition~\ref{prop:factorization}.  This completes the proof.
\end{proof}

\subsection{Area law}\label{sec:area-law}

We will now derive an area-law upper bound from the strong-coupling expansion constructed above. As usual in strong-coupling arguments, the decay comes from two ingredients: a filling constraint imposed by the central charge selection rule, and the exponential heat-kernel decay in the projective mass of the plaquette labels.

\begin{definition}\label{def:lattice-area}
Let $\ell$ be a closed lattice loop in $\Z^d$.  We write $j_\ell\in\Z^{E(\Z^d)}$ for its signed edge-current.  If $\Lambda\subset\Z^d$ is finite and contains $\ell$, define the finite-volume lattice area of $\ell$ by
\begin{equation}
\mathcal A_\Lambda(\ell):=\inf\left\{\sum_{p\in P(\Lambda)} |n_p|:n\in\Z^{P(\Lambda)},\ \partial^*n=-j_\ell\right\}.
\end{equation}
If no such $n$ exists in $\Lambda$, we put $\mathcal A_\Lambda(\ell)=+\infty$.  In infinite volume, set
\begin{equation}
\mathcal A(\ell):=\inf\left\{\sum_{p\in P(\Z^d)} |n_p|:n\in\Z^{P(\Z^d)},\ n \text{ finitely supported},\ \partial^*n=-j_\ell\right\}.
\end{equation}
\end{definition}
In other terms, $\mathcal A(\ell)$ is the minimal plaquette $\ell^1$-area of an integral lattice surface spanning the current of $\ell$.

When $\ell$ is a simple lattice loop, its edge-current has $\|j_\ell\|_1=|\ell|$.  Since every plaquette has four boundary edges, every filling $n$ satisfies
\[
|\ell|=\|j_\ell\|_1
\le
4\sum_p |n_p|.
\]
Therefore
\begin{equation}\label{eq:bound-l}
|\ell|\le 4\mathcal A(\ell).
\end{equation}

The following elementary lemma is the place where the area enters.

\begin{lemma}\label{lem:mass-fills-loop}
Let $\ell$ be a single loop and let $\mathfrak h$ be a stable small-mass rooted history contributing to the coefficient $\Phi^{(r)}_{\Lambda,T}(\ell)$. For each plaquette $p$, let $s_{\mathfrak h}(p)\in\Z$ be the aggregate reduced central charge carried by the plaquette labels of the history at $p$, counted with the signs inherited from the local loop--plaquette transfers. Let $m_{\mathfrak h}(p)$ be the corresponding aggregate projective mass. Then
\begin{equation}\label{eq:history-filling-current}
\partial^*s_{\mathfrak h}+j_\ell=0,
\end{equation}
and
\begin{equation}\label{eq:history-charge-mass}
|s_{\mathfrak h}(p)|\le m_{\mathfrak h}(p)\qquad\text{for every }p.
\end{equation}
Consequently,
\begin{equation}\label{eq:mass-fills-loop}
M(\mathfrak h):=\sum_p m_{\mathfrak h}(p)\ge \mathcal A_\Lambda(\ell).
\end{equation}
The same inequality holds in infinite volume with $\mathcal A(\ell)$.
\end{lemma}

\begin{proof}
We work on the stable small-mass sector, so the central selection rule splits as in Lemma~\ref{lem:central-selection}. Undo the adjunction changes used to define the rooted heat-kernel transition kernel. Each rooted-history term then comes from a finite sequence of coefficientwise topological coefficients before summing over the heat-kernel reference law.

Loop-only cut-and-join operations preserve the total edge-current of the loop family.  For a rooted loop--plaquette transfer at a plaquette $p$, the local reduced charge balance~\eqref{eq:local-reduced-charge-balance} says that the change of loop current is exactly the coboundary of the signed charge $s(\alpha_p)-s(\beta_p)$ removed from the plaquette label at that step.  Thus, if this signed charge is added to the aggregate plaquette current $s_{\mathfrak h}$, the quantity
\[
j-\partial^*s_{\mathfrak h}
\]
is invariant along the rooted history, where $j$ denotes the current of the current loop family.  Initially $s_{\mathfrak h}=0$ and $j=j_\ell$.

At a terminal coefficientwise source, the remaining loop family has zero current.  The terminal decorated vacuum coefficient may still contain plaquette labels, but the split central selection rule applied to this terminal coefficient makes their total reduced charge divergence-free.  Hence these terminal labels can only add a closed plaquette current to $s_{\mathfrak h}$ and do not change the boundary.  The invariant identity therefore gives
\[
-\partial^*s_{\mathfrak h}=j_\ell,
\]
which is \eqref{eq:history-filling-current}.

For every plaquette label
\[
\lambda_p=q_p\mathbf 1_N+[\lambda_p^+,\lambda_p^-]_N
\]
appearing in the history, its reduced central charge and projective mass satisfy
\[
\bigl||\lambda_p^+|-|\lambda_p^-|\bigr|
\le |\lambda_p^+|+|\lambda_p^-|.
\]
Summing this inequality over all signed contributions at the plaquette gives \eqref{eq:history-charge-mass}. Therefore $s_{\mathfrak h}$ is an integral plaquette filling of $\ell$. In particular, if $\mathcal A_\Lambda(\ell)=+\infty$, no such stable small-mass history can exist. Otherwise, by the definition of the finite-volume filling area,
\[
\mathcal A_\Lambda(\ell)
\le \sum_p |s_{\mathfrak h}(p)|
\le \sum_p m_{\mathfrak h}(p)
=M(\mathfrak h).
\]
The same finite support $s_{\mathfrak h}$ is an admissible filling in $\Z^d$, so the infinite-volume inequality follows directly from the definition of $\mathcal A(\ell)$.
\end{proof}

We can now state the area-law upper bound for all coefficients. Theorem~\ref{thm:main-area-law} will follow easily.

\begin{theorem}[Coefficientwise area law upper bound]
\label{thm:area-law}
Let $d\ge2$.  For every coefficient order $r\ge0$ there exists $T_{\mathrm{area}}(d,r)<\infty$ such that, for every $T>T_{\mathrm{area}}(d,r)$, there are constants $\sigma_r(T)>0$ and $C_r(T)<\infty$, with $\sigma_r(T)\to\infty$ linearly as $T\to\infty$, such that for every lattice loop $\ell$,
\begin{equation}\label{eq:infinite-volume-arealaw}
\left|
\Phi_{\infty,T}^{(r)}(\ell)
\right|
\le
C_r(T)^{|\ell|}
\exp\{-\sigma_r(T)\mathcal A(\ell)\}.
\end{equation}
The same finite-volume estimate holds with $\mathcal A_\Lambda(\ell)$:
\begin{equation}\label{eq:finite-volume-arealaw}
\left|\Phi_{\Lambda,T}^{(r)}(\ell)\right|\le C_r(T)^{|\ell|}\exp\{-\sigma_r(T)\mathcal A_\Lambda(\ell)\}.
\end{equation}
\end{theorem}

\begin{proof}
We first prove the finite-volume estimate.  By the rooted coefficientwise expansion of Theorem~\ref{thm:rooted-coeff-construction}, and since $\Phi_{\Lambda,T}^{(r)}(\ell)$ is independent of the choice of root, we may choose any non-trivial root $x$.

If $\mathcal A_\Lambda(\ell)=+\infty$, Lemma~\ref{lem:mass-fills-loop} shows that there is no stable finite-energy rooted history contributing to the coefficient. The exceptional sectors are remainders which are $O(N^{-R})$ for every $R$ in the coefficient construction, hence the coefficient itself is zero; with the convention $e^{-\sigma\cdot\infty}=0$, the desired estimate follows. We may therefore assume $\mathcal A_\Lambda(\ell)<\infty$.

By Lemma~\ref{lem:mass-fills-loop}, every stable small-mass rooted history $\mathfrak h$ contributing to the coefficientwise expansion satisfies
\begin{equation}
M(\mathfrak h)\ge \mathcal A_\Lambda(\ell).
\end{equation}
On the other hand, Proposition~\ref{prop:mass-sensitive-summability} gives, for every fixed coefficient order
\(r\), constants \(c>0\) and \(C_r<\infty\), depending only on \(d,r\), such that
\[
\sum_{\substack{\operatorname{ord}(h)=r\\ M(h)\ge A}}
\left|w^{[r]}_{\Lambda,T}(h)\right|
\le
C_r^{|\ell|}
\exp\{-(cT-\log C_r)A\},
\]
uniformly in the finite volume. Applying this with \(A=A_\Lambda(\ell)\) gives
\[
\left|\Phi^{(r)}_{\Lambda,T}(\ell)\right|
\le
C_r^{|\ell|}
\exp\{-(cT-\log C_r)A_\Lambda(\ell)\}.
\]
Thus~\eqref{eq:finite-volume-arealaw} holds with
\[
C_r(T):=C_r,\qquad \sigma_r(T)=cT-\log C_r.
\]
Choosing $T_{\mathrm{area}}(d,r)$ large enough makes $\sigma_r(T)>0$.  The infinite-volume estimate~\eqref{eq:infinite-volume-arealaw} follows by taking $\Lambda\uparrow\Z^d$, using the exponential locality and the convergence of the coefficients proved in Corollary~\ref{cor:infinite-volume-coefficients}.
\end{proof}

\begin{proof}[Proof of Theorem~\ref{thm:main-area-law}]
Apply Theorem~\ref{thm:area-law} with $r=0$ and set $T_{\mathrm{area}}(d)=T_{\mathrm{area}}(d,0)$, $C(T)=C_0\vee1$ and $\sigma(T)=\sigma_0(T)$. This gives~\eqref{eq:mf-area-law1}.

If $\ell$ is simple, then $|\ell|\le4\mathcal A(\ell)$ by~\eqref{eq:bound-l}.  Hence
\[
C(T)^{|\ell|}
e^{-\sigma_0(T)\mathcal A(\ell)}
\le
\exp\left\{
-\left(\sigma_0(T)-4\log C(T)\right)\mathcal A(\ell)
\right\}.
\]
After increasing $T_{\mathrm{area}}(d)$, the quantity
\[
\widetilde\sigma(T):=
\sigma_0(T)-4\log C(T)
\]
is positive.  This proves~\eqref{eq:mf-area-law2}.
\end{proof}

\subsection{Concluding remarks}\label{subsec:concluding-remarks}

We close the construction of the master field by discussing two complementary aspects which were not needed in the proof: its interpretation as a non-commutative probabilistic object, and the relation between the present coefficientwise expansion and gauge/string duality.

\subsubsection*{The master field as a non-commutative process}

The leading coefficient constructed above has a useful interpretation in non-commutative probability.  Fix a base point $o\in\Z^d$, and let $\mathcal L_o$ be the group of reduced lattice loops based at $o$, with product given by concatenation and inverse given by reversal.  The formula
\[
\tau_{\infty,T}^{(o)}(u_\ell):=\Phi^{(0)}_{\infty,T}(\ell),\qquad \ell\in\mathcal L_o,
\]
extended linearly to the group $*$-algebra $\mathbb C[\mathcal L_o]$, defines a normalized positive tracial state.  Indeed, it is the pointwise large-$N$ limit of the normalized matrix-trace states
\[
\tau_{\Lambda,T,N}^{(o)}(u_\ell)
=
\frac1N\mathbb E_{\Lambda,T,N}\Tr(U_\ell).
\]
Positivity is checked on finite linear combinations: if $a=\sum_i c_i u_{\ell_i}$, then all loops appearing in $a^*a$ are contained in a common finite support, and $\tau_{\infty,T}^{(o)}(a^*a)$ is the limit along finite volumes and $N\to\infty$ of $\tau_{\Lambda,T,N}^{(o)}(a^*a)\ge0$.  Thus the master field is not only a scalar Wilson-loop functional; it is the tracial non-commutative law of limiting lattice holonomies.  The factorization statement of Proposition~\ref{prop:factorization} says that scalar normalized traces become deterministic at leading order, but it does not imply freeness of the underlying loop-holonomy variables.

\subsubsection*{The master field and gauge/string duality}

Let us clarify the relation between the expansion proved above and the different gauge/string
duality statements which appear in the Wilson-action literature.  There are several distinct
notions which should not be conflated.

First, the finite-$N$ surface-sum of Cao--Park--Sheffield \cite{CPS25} is an exact surface expansion for the Wilson-action loop expectation, with the partition-function factor kept outside the surface sum.  Its surfaces are genuine embedded maps, but the formula is not a normalized connected expansion: vacuum components are not cancelled inside the surface sum itself.  The finite-$N$ decorated-surface expansion of the companion paper \cite{Lem26a} has the same structural feature in the Fourier setting.  It expands fixed topological coefficients, before division by the partition function.

Second, the string trajectories of Chatterjee and Jafarov \cite{Cha19,Jaf16} arise from a different mechanism. For the Wilson action, the large-$N$ master loop equation closes on Wilson loop observables. At strong coupling it can be solved by a contraction/Picard expansion, and the terms of this Neumann series may be interpreted as vanishing string trajectories.  This gives a normalized large-$N$ expansion, but it is an algebraic trajectory expansion rather than a surface formula by embedded maps.

Third, the planar surface formula of Borga--Cao--Shogren-Knaak \cite{BCSK24,BCSK25} should be viewed as a further Wilson-specific bridge between these two viewpoints.  It is not obtained by taking a termwise large-$N$ limit of the finite-$N$ Weingarten surface sum, nor by literally thickening Chatterjee's trajectories.  Rather, one proposes a planar embedded-map ansatz, motivated by the finite-$N$ surface formula, and then proves that this ansatz satisfies the limiting Wilson master loop equation.  This works because two simplifications occur at leading order for the Wilson action: the limiting equation is trace-closed, and the relevant planar weights reduce, after the required cancellations, to explicit signed Catalan weights.

The present heat-kernel/Fourier setting lacks both simplifications.  The coefficientwise master loop equation does not close on Wilson loops alone; it lives on rooted decorated-loop states and non-zero plaquette transfers change Fourier labels.  The rooted histories constructed above therefore keep track of spectral data and of an auxiliary order of resolution.  Their loop-only steps are trace-algebra operations; this auxiliary order has no geometric meaning as a time direction in the lattice, and the histories cannot be turned into embedded maps by the Borga--Cao--Shogren-Knaak ansatz.

Conversely, the decorated surfaces of \cite{Lem26a} are available only at the level of fixed finite-$N$ topological coefficients.  Since the partition-function normalization is external to that expansion, vacuum components created by the Fourier sum are not removed at the surface level.  Passing from those finite-$N$ surfaces to a normalized large-$N$ surface expansion would require an additional connected, or pinned, surface-level cluster expansion together with large-$N$ control of its local weights.  This is not provided by the present argument.

Thus the result proved here should be understood as a normalized, coefficientwise expansion over rooted strong-coupling histories, built from the same local-channel and recoupling calculus as the companion paper. The addition of local spectral data makes it much harder to give a gauge/string interpretation in the spirit of \cite{Cha19,Jaf16}.

\bibliographystyle{alpha}
\bibliography{master-field-Zd}

\end{document}